\documentclass{article}

\usepackage{enumitem}

\usepackage[utf8]{inputenc}
\usepackage{amsmath}
\usepackage{cancel}
\usepackage{amssymb}
\usepackage{amssymb}
\usepackage{authblk}
\usepackage{amsfonts}
\usepackage{mathrsfs}
\usepackage{diagbox}

\usepackage{graphics}
\usepackage{amsthm}
\usepackage[all]{xy}
\usepackage[margin=1.4in,marginparwidth=1in,marginparsep=.05in]{geometry}
\PassOptionsToPackage{hyphens}{url}
\usepackage{tikz-cd}
\usepackage{titlesec}
\usepackage{xcolor,graphicx}
\usepackage[abs]{overpic}
\usepackage{verbatim}
\usepackage{extarrows}
\usepackage{tikz}
\usetikzlibrary{shapes,arrows,positioning}
\usepackage{mathtools}

\usepackage{subfig}
\usepackage{hyperref}
\usetikzlibrary{arrows}

\makeatletter
\newcommand*{\rom}[1]{\expandafter\@slowromancap\romannumeral #1@}

\titlespacing*{\section}{0pt}{\baselineskip}{\baselineskip}
\titleformat{\subsection}{\normalfont\bfseries}{\thesubsection.}{1em}{}
\titleformat{\subsubsection}{\normalfont}{\thesubsubsection.}{1em}{\itshape}

\theoremstyle{plain}
\newtheorem{theorem}{Theorem}[section]
\newtheorem{proposition}[theorem]{Proposition}
\newtheorem{lemma}[theorem]{Lemma}
\newtheorem{corollary}[theorem]{Corollary}

\theoremstyle{definition}

\newtheorem{assumption}{Assumption}
\newtheorem{definition}[theorem]{Definition}

\theoremstyle{remark}
\newtheorem{remark}[theorem]{Remark}

\numberwithin{equation}{section}

\let\pa\partial

\newcommand{\R}{\mathbb R}
\newcommand{\bA}{\mathbf A}
\newcommand{\bB}{\mathbf B}
\newcommand{\bC}{\mathbf C}
\newcommand{\bD}{\mathbf D}

\newcommand{\bF}{\mathbf F}

\newcommand{\bH}{\mathbf H}
\newcommand{\bI}{\mathbf I}

\newcommand{\bP}{\mathbf P}

\newcommand{\bN}{\mathbf N}
\newcommand{\bQ}{\mathbf Q}
\newcommand{\bR}{\mathbf R}
\newcommand{\bS}{\mathbf S}

\newcommand{\ba}{\mathbf a}
\newcommand{\bb}{\mathbf b}

\newcommand{\bm}{\mathbf m}
\newcommand{\bn}{\mathbf n}
\newcommand{\be}{\mathbf e}

\newcommand{\bt}{\mathbf t}
\newcommand{\bq}{\mathbf q}

\newcommand{\bT}{\mathbf T}
\newcommand{\bu}{\mathbf u}

\newcommand{\bv}{\mathbf v}
\newcommand{\bw}{\mathbf w}
\newcommand{\bW}{\mathbf W}
\newcommand{\bx}{\mathbf x}

\newcommand{\bzero}{\mathbf 0}
\newcommand{\bbf}{\mathbf f}

\newcommand{\Div}{\mathop{\rm div}}
\newcommand{\divG}{{\mathop{\,\rm div}}_{\Gamma}}
\newcommand{\gradG}{\nabla_{\Gamma}}
\newcommand{\nablaG}{\nabla_{\Gamma}}

\newcommand{\divM}{{\mathop{\,\rm div}}_{M}}
\newcommand{\gradM}{\nabla_{M}}
\newcommand{\nablaM}{\nabla_{M}}

 \renewcommand{\div}{{\mathop{\rm div}}\,}

\newcommand{\Rn}{\mathbb{R}^n}

\newcommand{\tr}{{\rm tr}}

\DeclareRobustCommand{\tripledot}{%
  \mathrel{\vbox{\baselineskip.65ex\lineskiplimit0pt\hbox{.}\hbox{.}\hbox{.}}}%
}

\DeclareGraphicsExtensions{.pdf,.eps,.ps,.eps.gz,.ps.gz,.eps.Y}

\newcommand{\bSigma}{\boldsymbol{\Sigma}}
\newcommand{\bOmega}{\boldsymbol{\Omega}}
\newcommand{\bLambda}{\boldsymbol{\Lambda}}

\usepackage[normalem]{ulem} 
\newcommand{\LB}[1]{{\color{black}#1}}

\newcommand{\VY}[1]{{\color{black}#1}}

\newcommand{\RHN}[1]{{\color{black}#1}}

\title{A hydrodynamical model of nematic liquid crystal films with a general state of orientational order}
\author{Lucas Bouck\thanks{Department of Mathematics,
University of Maryland,  College
Park - 20742, MD, USA. Email: lbouck@umd.edu.},\, Ricardo H. Nochetto\thanks{Department of Mathematics and Institute for Physical Science and Technology, University of Maryland, College
Park - 20742, MD, USA. Email: rhn@umd.edu.},\, Vladimir Yushutin\thanks{Department of Mathematics,
University of Maryland,  College
Park - 20742, MD, USA. Email: yushutin@umd.edu.}}

\begin{document}

\maketitle
\begin{abstract}
    We develop a Q-tensor model of nematic liquid crystals occupying a \RHN{stationary} surface which represents a fluidic material film in space. In addition to the evolution due to Landau--de\,Gennes energy the model includes a tangent viscous \RHN{incompressible} flow along the surface. \RHN{A thermodynamically consistent} coupling of a two-dimensional flow and a three-dimensional Q-tensor dynamics is derived from the generalized Onsager principle following the Beris--Edwards system known in the flat case. The main novelty of the model is that it allows for a flow of an arbitrarily oriented  liquid crystal so the Q-tensor is not anchored to the tangent plane of the surface, \RHN{and also obeys an energy law}. Several numerical experiments explore kinematical and dynamical properties of the novel model.

\end{abstract}

\tableofcontents

\section{Introduction}\label{intro}

Modeling of materials with orientational order is a challenging task. There is a variety of approaches  including particle theories as well as continuum director field and Q-tensor theories which may be customized by specifying the energy landscape to accommodate different types of phase transitions. 

In this paper we focus on the Q-tensor approach \cite{borthagaray2021q}, \cite{mottram2014introduction} . The main idea is to characterize the nematic liquid crystal state by averaging the probability density $\rho_p$ at a point $p$ over the unit sphere $\mathbb{S}$ which 
is the statistical distribution of the orientations of liquid crystal molecules at $p$. More specifically, a Q-tensor at $p$ is a symmetric, traceless matrix 
\begin{align*}
    \bQ=\int_{\mathbb{S}^2}\rho_p(s)s\otimes{}s\,ds -\frac13\bI
\end{align*}
defined as the difference between the second moments (the first moments are trivial due to the so-called tail-to-head symmetry) and the isotropic state $\frac13 \bI$. The physically relevant information derived from a Q-tensor simulation is the eigenframe of the matrix 
\begin{align}\label{spectral}
\bQ&= \lambda_1 (\bq_1\otimes \bq_1) + \lambda_2 ( \bq_2\otimes \bq_2)+ \lambda_3(\bq_3\otimes \bq_3) 
\end{align}
with the most significant orientations $\bq_i$ and corresponding eigenvalues $\lambda_i$. Then the Landau--de\,Gennes energy of a liquid crystal occupying a domain $\Omega$ combines the elastic energy with a material constant $L>0$ and the double-well potential $F[\bQ]$ with material constants $c>0, a, b$:
\begin{align}\label{bulk_e}
    E_{LdG}[\bQ]&=\int_\Omega \frac{L}2|\nabla\bQ|^2+\int_\Omega F[\bQ]\,,\\
    F[\bQ]&=\frac{a}{2}|{\bf Q}|^2 - \frac{b}{3}({\bf Q}:{\bf Q}^2)+\frac{c}{4}|{\bf Q}|^4\,,\label{dw}
\end{align}
where the so-called \textit{one-constant} approximation \cite{borthagaray2021q, de1993physics, sonnet2012dissipative} is considered for simplicity.
At the same time Onsager reciprocity principle \cite{onsager1931reciprocal1}, \cite{onsager1931reciprocal2}  suggests that space variations of the orientational order should be matched with a macroscopic flux of momentum \cite{yang2016hydrodynamic} to have an energy law provided by thermodynamics.  The coupling of the Q-tensor dynamics with the transport of the momentum is a  delicate matter due to possible non-conservative behavior of the total energy of the system. A well-known thermodynamically consistent model is the Beris--Edwards system \cite{beris1994thermodynamics}, \cite{zarnescu2012topics} in which the transport of the Q-tensor exerts \textit{Ericksen stress} $\bSigma=\bQ\bH-\bH\bQ$ and  \textit{Leslie force} $\bH:\nabla\bQ$ to the momentum flow. Here the \textit{molecular field} $\bH$ is the traceless and symmetric variation of $E_{LdG}[\bQ]$ with respect to $\bQ$. 

\LB{A surface Beris--Edwards model would help in understanding the dynamics of thin nematic liquid crystals shells. Thin nematic liquid crystal shells are potential candidates for self-assembling colloids due to the configuration of defects in a thin LC shell \cite{nelson2002toward}, which can be tuned by varying the thickness of the LC shell \cite{lopez2011frustrated}.}

\LB{The analytical properties of the Beris--Edwards system for flat domains in $\R^2$ and $\R^3$ has been studied extensively. A nonexhaustive list of references include works on existence of weak solutions in $\mathbb{R}^2,\mathbb{R}^3$ \cite{paicu2012energy}, weak-strong uniqueness and higher regularity in $\mathbb{R}^2$ \cite{paicu2012energy}, short time existence for strong solutions in bounded domains \cite{Abels2015StrongSF}, existence of weak solutions and short time well-posedness  with mixed boundary conditions \cite{abels2014well}, and
physical eigenvalue preservation of the $Q$-tensor in the corotational Beris--Edwards system \cite{wu2019dynamics}. To the best of our knowledge, many of these questions of existence, possible uniqueness, regularity, and eigenvalue preservation remain open for hydrodynamical models of liquid crystals on curved surfaces.}

Regarding numerical methods for the Landau--de\,Gennes dynamics and for the Beris--Edwards model for flat domains we refer to  \cite{zhao2017novel}, \cite{cai2017stable}, \cite{gudibanda2020convergence}  and references therein. In addition, we refer to \cite{sonnet2012dissipative} for modeling of dissipative ordered fluids.

We aim to extend the Beris--Edwards model to stationary curved surfaces. For an example of a situation in which the model is relevant, one may think of a liquid crystal material confined between two parallel surfaces which may be far enough apart to fit normally oriented rod-like particles but are sufficiently close so Q-tensor states are constant along the thickness. Yet tangent distortions of the orientational order generate a tangent macroscopic flow of matter. The main challenge is thus to establish a thermodynamically consistent energy law for a curved surface while still having a generically oriented Q-tensor.

We start with a discussion of situations where a generically oriented Q-tensor description of a liquid crystal may be desirable. In the case of long bulk cylinders with homeotropic anchoring on the cylinder wall, experiments show that the liquid crystal may experience what is called an escape to the third dimension \cite{crawford1991escaped,meyer1973existence,williams1972nonsingular}. This behavior has also been shown theoretically using director field models \cite{meyer1973existence}. For flat 2D disks, escape to the third dimension has also been observed to be energetically favorable for Q-tensor models due to the complex Landau-de Gennes energy landscape \cite{ignat2016stability,ignat2020symmetry}. We point to \cite{hu2016disclination} for a numerical exploration of this landscape. Moreover, in thin flat domains where tangential anchoring is present on the top and bottom boundary, numerical experiments suggest that the liquid crystal orientation may not stay planar \cite{chiccoli2002topological}. For thin shells with varying thickness and a bead inside, the numerics in \cite{gharbi2013microparticles} provide a plausible explanation of experiments of a metastable configuration with an escape to the third dimension near defects. \LB{Numerical experiments of three-dimensional LC shells show the escape to the third dimension near defects when the shell increases in thickness \cite{bates2010defects, koning2013bivalent}. Additionally, \cite{koning2016spherical} provides a plausible explanation of the presence of two $+1/2$ defects and one $+1$ defect in experiments \cite{lopez2011frustrated}: the $+1$ defect is composed of two boojums on the confining surfaces and escape to the third dimension occurs in the thickness of the shell.} \LB{We finally point to the experiments in \cite{murray2017decomposition} where escape to the third dimension is observed near topological defects in a thin LC cell. This escape to the third dimension depends on the strength of surface anchoring.} All these situations do not involve coupling with a fluid, but do suggest that a generically oriented Q-tensor description may be warranted.

Existing models of liquid crystals films differ from each other in the assumed structure of the Q-tensor eigenframe and its relation to the tangent plane of the surface. For example, in the thin-film models of \cite{kralj2011curvature}, \cite{napoli2012surface} the Q-tensor is assumed to be \textit{conforming} and \textit{flat-degenerate} with zero eigenvalue in the normal direction.
We define these concepts as follows for a general $Q$-tensor $\bQ\in \R^{3\times3}$ on a surface $\Gamma$ :
\begin{equation}\label{conf_def}
\begin{minipage}{0.84\linewidth}
{\em $\bQ$ is called conforming to $\Gamma$ at a point $\bx\in\Gamma$ if one of the eigenvectors equals the unit normal vector $\bn(\bx)$ to $\Gamma$;}
\end{minipage}
\end{equation}
\begin{equation}\label{fd_def}
\begin{minipage}{0.84\linewidth}
{\em $\bQ$ is called flat-degenerate  at a point $\bx\in\Gamma$ if one of the eigenvalues is zero.}
\end{minipage}
\end{equation}

Physically a conforming and flat-degenerate Q-tensor with zero eigenvalue in the normal eigendirection  corresponds to the case of liquid crystal molecules being located strictly in the tangent plane essentially forming a two-dimensional liquid crystal state in each tangent plane. The assumption that the Q-tensor is conforming and flat-degenerate at each point requires the presence of an ad-hoc large interface force\VY{, e.g. a reaction force from rigid walls surrounding the thin film from one or both sides,} which dominates all the other forces \VY{because otherwise the Q-tensor field would evolve to a uniform, uniaxial state in $\mathbb{R}^3$ violating the conformity assumption}.

The conformity assumption reduces the number of independent coefficients for a general traceless, symmetric Q-tensor \cite{nestler2020properties}. In the conforming case, when 
one eigenvector is normal to the surface, the tangent orientational order of the liquid crystal state is described by a tangent director field and a scalar order parameter. The normal orientational order, i.e. the eigenvalue corresponding to the normal eigenvector, is described by a scalar field which often (e.g. \cite{nitschke2019hydrodynamic}) has a prescribed constant value; see \cite{nestler2020properties} for a general discussion of conforming models. \VY{Thus, the conformity assumption facilitates the reduction of the number of Q-tensor unknowns from 5 to 3}.
The expression of the elastic energy $|\nabla{}\bQ|^2$ from \eqref{bulk_e} in terms of a tangent director field contains \VY{several}  geometric terms 
which complicate the numerical implementation.
Note that strategies for dimension reduction of the Landau--de\,Gennes energy other than \cite{nestler2020properties} are possible, e.g. \cite{novack2018dimension}.

In addition to conformity, the assumption of flat-degeneracy in the normal eigendirection further reduces the number of independent variables of the Q-tensor field by 1 - from 3 to 2. 
It should be noted that such flat-degenerate Q-tensor fields are biaxial from the perspective of $\R^3$ while the minimizers of a double-well potential have to be uniaxial; see Definition \ref{uni-bi}. Consequently, \VY{the assumption of flat-degeneracy means that there exists a force} with a special structure to prevent the evolution of a Q-tensor towards a uniaxial state.

Besides \cite{kralj2011curvature} and \cite{napoli2012surface}, where conformity and flat-degeneracy in the normal eigendirection are assumed, other models of liquid crystal films may relax some of these assumption but not entirely. In \cite{nitschke2018nematic} the eigenvalues of a Q-tensor may be general but the eigenframe is assumed to be conforming to the surface. The normal eigenvalue and the tangent order parameter undergo interrelated $L^2$ gradient flows which are formulated in the \VY{language of local tensor calculus}. In \cite{golovaty2017dimension} the Q-tensor does not have to be conforming but no evolution laws are discussed. The bulk Landau--de\,Gennes energy  of a thick film is combined in \cite{golovaty2017dimension} with an anchoring energy of the film interfaces, and the minimizers of the resulting landscape are studied via $\Gamma$-convergence.

The aforementioned references involve liquid crystal models with no hydrodynamical properties. To the best of the authors' knowledge the only paper in which a coupling of a Q-tensor and linear momentum is considered for curved liquid crystal films is \cite{nestler2021active}. In such a paper the Q-tensor is assumed to be conforming but the evolution laws
are not shown to have an energy structure. An important conclusion of the present paper is the requirement of thermodynamical consistency, namely the existence of a proper energy law, strongly ties the kinematics and the dynamics of curved liquid crystal films. Nevertheless, it is physically reasonable that there are regimes for which the anchoring of Q-tensors can be justified if it is posed weakly with an energy term that penalizes the non-conformity.

\LB{ We also note that the model derived in this paper reduces to surface Navier--Stokes equations when $\bQ\equiv 0$. There has been extensive work on surface Navier--Stokes for modeling and numerics. We point to \cite{arroyo2009relaxation,reuther2015interplay, koba2017energetic, jankuhn2018incompressible, mietke2019self} for works on modeling and \cite{reuther2015interplay,nitschke2012finite, fries2018higher, fries2018higher, reuther2020numerical, brandner2022finite, pearce2019flow} for works on numerics of surface fluid flows.}

The goal of this paper is three-fold. The first one is to derive a surface model of the liquid crystal flow where the orientational order is not anchored to the surface or, in other words, the Q-tensor is not conforming to the surface. This model is derived via the generalized Onsager principle mainly following \VY{\cite{sonnet2012dissipative}}, \cite{yang2016hydrodynamic}, \cite{unpublishedWang} and a private communication with Qi Wang. A similar approach called Lagrange-Rayleigh principle has been applied to Ericksen--Leslie theory involving a director field tangent to the film \cite{napoli2016hydrodynamic}. \VY{The formalism of Onsager is quite general and it does not involve any assumptions on the relation between the dimensions of the model and its environment, and, hence, is suitable for the modelling of embedded surfaces. So, the applicability of the generalized Onsager principle as a guiding physical principle of thin-film modeling is assumed in this paper. We refer to \cite{doi2011onsager} for the principle's thermodynamical premises and to  \cite{doi2015onsager}, \cite{wang2021onsager} for further details of its application to the particular physical systems.}
The second goal is to use the language of \LB{differential geometry in cartesian coordinates} 
instead of the language of differential geometry that refers to local parametric coordinate systems  \cite{nestler2021active}, thus \VY{simplifying implementation of the model in standard computational packages}. The application of  this approach to Q-tensors on surfaces appears to be new. The third goal is to explore computationally the action of three forces, one new to our surface model, and the consequences of non-conformity for the dynamics of Q-tensors on surfaces.

The outline of the paper is as follows.   
In Section \ref{notation}, following \cite{jankuhn2018incompressible}, we give the preliminaries of tangential calculus and introduce two tensor 
derivatives on a surface $\Gamma$: the \VY{external} surface derivative \eqref{nablam} and the covariant surface derivative \eqref{nablaG} - both will be used in our surface model; \RHN{in Appendices \ref{apx:calc} and \ref{apx:by_parts} we provide further discussion and proofs.} Section \ref{kinematics} is devoted to the development of the kinematical properties of the surface Beris--Edwards system. We introduce the new notion of passive transport of generically oriented Q-tensors along a surface flow; see Definition \ref{circle_def}. This novel concept is based on  Assumption \ref{assum_Q} which possesses a clear physical meaning. In Section \ref{Onsager} we apply systematically  the generalized Onsager principle \cite{Wang2021} to derive a thermodynamically consistent surface Beris--Edwards model based on the kinematical properties defined in  Section \ref{kinematics}, and establish the underlying energy law. In Section \ref{S:representation} we discuss the biaxiality parameter $\beta[\bQ]$ and the non-conformity parameter $r_\Gamma[\bQ]$ and study their properties. 
We conclude in Section \ref{numerics} with several numerical simulations of the surface Beris--Edwards model derived in Section \ref{Onsager} to demonstrate its basic properties, and investigate the action of the three induced forces and role of non-conformity. We do observe nonconforming dynamics connecting conforming states in Section \ref{S:normal_anchoring}. The parameters $\beta[\bQ]$ and $r_\Gamma[\bQ]$ are crucial to describe the numerical experiments.

\section{Preliminaries in tangential calculus}\label{notation}

In this section the surface and all the fields are assumed to be sufficiently smooth.
Although we are concerned with a surface model, we intentionally work with tensor fields in $\Rn$, the ambient space to the surface, to avoid the less practical parametric approach. \VY{This section summarizes and clarifies two types of surface derivatives of tensors of order up to two, see e.g. \cite{jankuhn2018incompressible}, \cite{nestler2019finite}, which are both relevant to the surface Beris--Edwards model to be derived in this paper. Some preliminary notations are given in Appendix \ref{apx:calc}.} \VY{Integration by parts for the two types of derivatives is discussed in Appendix~\ref{apx:by_parts}.}

\subsection{{External} surface derivatives}

Here we introduce some standard operators of \VY{calculus on embedded surfaces}. Intuitively, these operators replicate standard \VY{Cartesian} operators with an assumption that their tensor arguments are extended  from the surface constantly along the normal direction.
\VY{Although the concept is certainly not novel, in this paper we will call  such operators \textit{external} to highlight the difference from similar operators which are based on the \textit{covariant derivative}, see Section~\ref{cov_subsection}.}

 Consider a closed surface $\Gamma\subset\Rn$ defined as the zero level set of its distance function $d\in{}C^2(\Omega_\delta)$ where $\Omega_\delta=\{\bx\in\Rn: |d(\bx)|<\delta\}$ is a tubular neighborhood of $\Gamma$ of thickness $\delta>0$. The boundary $\pa\Omega_\delta$ consists of two parallel surfaces, $\Gamma_\delta^+=\{\bx\in\Rn: d(\bx)=\delta\}$ and $\Gamma_\delta^-=\{\bx\in\Rn: d(\bx)=-\delta\}$. \LB{By means of the unit vector field $\bn=\nabla{}d\in{}C^1(\Omega_\delta)^n$, which is orthogonal to level sets of $d(\bx)$, we define the projectors $\bN$ and $\bP$ onto the normal and tangent subspaces to such level sets} as well as the shape operator $\bB$ (or Weingarten map) to be
\begin{align*}
    \bN=\bn\otimes\bn\in{}C^1(\Omega_\delta)^{n\times{}n}\,,\quad \bP=\bI-\bn\otimes\bn\in{}C^1(\Omega_\delta)^{n\times{}n}\,,\quad \bB=\nabla\bn\in{}C(\Omega_\delta)^{n\times{}n}\,,\quad\bx\in\Omega_\delta
\end{align*}
where $\bI$ is the identity operator.

\VY{Consider a tensor (scalar, vector, matrix) field $\bT$ on $\Omega_\delta$. The so-called \VY{\textit{external}} derivative $\nablaM \bT$ (see Definition~\ref{nablam} below) guarantees that} $$(\nablaM{}\bT)\bv=(\nabla\bT)\bP\bv\,,\qquad\forall\bv\in\Rn\,,\qquad\bx\in\Omega_\delta.$$ Essentially, this condition prescribes a \VY{non-standard Cartesian derivative} in $\Omega_\delta$ which disregards variations of the tensor field in the normal direction: \RHN{if $\bv=\bP\bv$ then $\nabla$ and $\nablaM{}$ coincide.} \VY{We stress that the external derivative $\nablaM \bT$ evaluated on $\Gamma$  depends only on the values of $\bT$ on $\Gamma$. The latter is straightforward for scalar tensor fields, and, therefore, holds for the external derivatives of a vector and a matrix fields as well since they are based on the external derivatives of scalar components.} 

\begin{definition}[\VY{external} derivative] For a scalar field $f$, a vector field $\bu$, and a matrix field $\bA$ on $\Omega_\delta$ the \VY{external} surface derivative is given by
 \begin{equation}\label{nablam}
 \begin{aligned}
\nabla_M{}f&=\bP\nabla{}f=\sum_{j=1}^n(\bP\be_j)\partial_jf\,,
\\
\nabla_M\bu
&=(\nabla\bu)\bP=\sum_{j=1}^n\partial_j\bu\otimes\bP\be_j=\sum_{j=1}^n\be_j\otimes{}\nabla_M\bu_j\,,
\\
 \nabla_M\bA&=\sum_{j=1}^n\partial_j\bA\otimes\bP\be_j=\sum_{j=1}^n\be_j\otimes\nablaM(\bA^T)_j\,,\qquad\bx\in\Omega_\delta\,.
  \end{aligned}
  \end{equation}
 \end{definition}
 Therefore, the \textit{\VY{external} surface directional derivative} along a vector field $\bv$ is then given by
  \begin{equation}\label{surf_dir_der}
  \begin{aligned}
  (\nablaM^Tf)\bv&=(\nabla^Tf)\bP\bv=\sum_{j=1}^n\bv\cdot(\bP\be_j)\partial_jf\,,  
  \\
    (\nabla_M\bu)\bv
&=(\nabla\bu)\bP\bv=\sum_{j=1}^n(\bv\cdot\bP\be_j)\partial_j\bu=\sum_{j=1}^n(\nabla_M\bu_j\cdot\bv)\be_j\,, 
  \\(\nabla_M\bA)\bv&=\sum_{j=1}^n(\bv\cdot\bP\be_j)\partial_j\bA=\sum_{j=1}^n(\nablaM\bA_j)\bv\otimes{}\be_j\, .
  \end{aligned}
  \end{equation}
 \VY{Note that $\nablaM^T \! f$ is a short notation for $(\nablaM f)^T$. We note that the surface directional derivative $(\nablaM\bu)\bv$ of a vector field  $\bu$ may have non-zero normal components.}

\begin{remark}[\VY{normal extension}]
   \VY{Consider a tensor field $\bT$ with values on $\Gamma$ only.} \VY{The normal extension $\bT^e$ on $\Omega_\delta$ is defined by}
\begin{align*}
    \bT^e(\bx):=\bT(\bx-d(\bx)\bn(\bx))\,,\qquad \bx\in\Omega_\delta;
\end{align*}
\RHN{thus $\bn^e=\bn$. A key property of the normally extended tensor fields is the vanishing} of the derivative in the normal direction which can be expressed via \eqref{dir_der} as:
\begin{align}\label{e-n}
    (\nabla{}\bT^e(\bx))\bn(\bx)=0\,,\qquad \bx\in\Omega_\delta.
\end{align}
\VY{Consequently, the external derivative of the normal extensions $\bT^e$ satisfies    \begin{align}\label{nablaM-normalextension}
    (\nablaM\bT)\bv=(\nabla\bT)\bP\bv=
   (\nabla\bT^e)\bv\qquad\forall{}\bv\in \Rn\,,\qquad\bx\in\Gamma\, ,
\end{align}
whence, for a normally extended tensor field the directional derivatives due to $\nabla$ and $\nablaM$ coincide for any $\bv$ (not just for $\bv=\bP\bv)$. Note that   \ref{nablaM-normalextension} could be seen as a way to compute $\nablaM \bT$ on $\Gamma$: given a field on $\Omega_\delta$, restrict its values to $\Gamma$, extend it normally to $\Omega_\delta$ and find its Cartesian gradient.}
\end{remark}

  The \textit{surface divergence} operator is defined in the same spirit: for normally extended tensors the result corresponds to the \VY{Cartesian} divergence \eqref{bulk_div}:
  \begin{align*}
 \divM \bu  &= \tr (\gradM \bu)
 \,,\quad
 \divM \bA  = \sum_{j=1}^n\be_j\,\divM (\bA^T)_j
\,,\quad\divM\nablaM\bA=\sum_{j=1}^n\be_j\otimes\divM\nablaM(\bA^T)_j\,.
\end{align*}
\LB{\begin{remark}[alternative definition of divergence]\label{rmk:alt-divG}We note that the surface divergence is the trace of the surface gradient. This definition coincides with those in differential geometry and other works in finite element methods for surface PDEs. For instance, we note that our definition of $\divM$ is consistent with previous work on finite elements for surface PDEs \cite[Eq. (2.7)]{dziuk2007finite}, and our definition of $\divG$ below is consistent with the definition of surface divergence found in previous work on modeling of elastic surfaces \cite[Eq. (2.8)]{gurtin1975continuum}. This definition of surface divergence is not the $L^2$-adjoint of $\nabla_\Gamma$ for a vector field with nonzero normal component as seen in Proposition \ref{vector_parts} (covariant integration by parts) below.
\end{remark}
}
\begin{definition}[\VY{external} strain-rate and spin tensors] For a given velocity field $\bv$ on $\Gamma$ we define the   \textit{\VY{external} strain-rate} and \textit{\VY{external} spin} tensors, respectively, as follows:
\begin{align} \label{strainM}
 \bD_M(\bv)= \frac12  (\nablaM \bv +\nablaM^T\bv) \,,\qquad
  \bW_M(\bv)=\frac12  (\nablaM \bv -\nablaM^T\bv) \, .
 \end{align}
 \end{definition}
 \VY{Again, $\nablaM^T \! \bv$ is a short notation for $(\nablaM \bv)^T$.} These rates correspond to symmetric and antisymmetric parts of the \VY{Cartesian} gradient of the vector field normally extended from the surface $\Gamma$ to its neighborhood $\Omega_\delta$. 
Using \eqref{contractions} and \eqref{surf_dir_der} we define the contraction  $\bC:\nablaM\bA    $ of a second order tensor $\bC$ and the surface gradient of a second order tensor $\bA$ as a vector such that for all $\bv\in\Rn$ we have
\begin{align}\label{contr_definition}
    (\bC:\nablaM\bA)\cdot\bv&=\bC:(\nablaM\bA)\bv\,.
    \end{align}
Since  $\bC:(\nabla_M\bA)\bv=\sum_{j=1}^n(\bv\cdot\bP\be_j)\bC:\partial_j\bA
    =\left(\sum_{j=1}^n(\bC:\partial_j\bA)\bP\be_j\right)\cdot\bv$ we obtain
    \begin{align}\label{tangent_contr}
    \bC:(\nabla_M\bA)&
    =\sum_{j=1}^n(\bC:\partial_j\bA)\bP\be_j
\end{align}
which shows that vector $\bC:\nablaM\bA$ is tangent to $\Gamma$.

\subsection{Covariant surface derivatives}\label{cov_subsection}

Here we discuss covariant operators  on the surface $\Gamma$, considered as an isometric embedding of a $C^2$ Riemannian manifold, which are intrinsic in the  sense that they only depend on the Riemannian structure and the surface values of the argument \VY{if it belongs to the tangent plane}. Intuitively,  the covariant derivative of a tangent object measures the tangent part of the change of the object in a tangent direction. 

For each $\bx_0\in\Gamma$ the subspace $\{\bx\in\Rn: \bx-\bx_0=\bP(\bx_0)(\bx-\bx_0)\}$ is identified with the tangent space of the manifold at $\bx_0$.  A linear operator $\bA$ is called \textit{tangent} if the normal $\bn$ is in its kernel and its range belongs to the tangent plane, or $\bA=\bP\bA\bP$. 
Given a tangent vector $\bv$ at $\bx_0$   consider a regular curve $\gamma:(a,b)\rightarrow\Gamma$, $\gamma(t_0)=\bx_0$, $t_0\in(a,b)$ such that $\gamma'(t_0)=\bv$. \VY{In the following we assume that all tangent planes of $\mathbb{R}^3$ are identified with itself as usual so the addition of tensors from different points is meaningful.} The tangent component of the  variation along the curve defines the action of the covariant directional derivative:
\begin{align*}
(\nablaG^Tf(\bx_0))\bv&=  \hskip 1.25cm \lim_{t\rightarrow t_0}\frac{1}{t-t_0}\left(f(\gamma(t))-f(\bx_0)\right)
\\
 (\nablaG\bu(\bx_0))\bv&= \bP(\bx_0)\left(  \lim_{t\rightarrow t_0}\frac{1}{t-t_0}\left(\bu(\gamma(t))-\bu(\bx_0)\right)\right)
 \\
 (\nablaG\bA(\bx_0))\bv&= \bP(\bx_0)\left(  \lim_{t\rightarrow t_0}\frac{1}{t-t_0}\left(\bA(\gamma(t))-\bA(\bx_0)\right)\right)\bP(\bx_0)
\end{align*}
which, by extending the tensor fields normally and applying \eqref{surf_dir_der}, can be shown to be equivalent to the following expressions for all $\bx\in \Gamma$
\begin{equation}\label{convG}
\begin{aligned}
(\nablaG^Tf)\bv&=(\nablaM^Tf)\bv
\,,\qquad
 (\nablaG\bu)\bv=\bP(\nablaM\bu)\bv
 \\
 (\nablaG\bA)\bv&=\bP(\nabla_M\bA)\bv\bP
 =\sum_{j=1}^n(\nablaG\bA_j)\bv\otimes{}\bP\be_j\,.
\end{aligned}
\end{equation}

 \VY{Note that $\nablaG^T \! f$ is a short notation for $(\nablaG f)^T$.} Finally, we give the definition of the \textit{covariant surface derivative}, which applies to fields that are not necessarily constant along the normal direction:
\begin{definition}[covariant derivatives] For a scalar field $f$, a vector field $\bu$, and a matrix field $\bA$ on $\Omega_\delta$ the covariant surface derivative is given by
\begin{equation}\label{nablaG}
\begin{aligned}
\nablaG{}f&=\nablaM{}f=\sum_{j=1}^n(\bP\be_j)\partial_jf\,,
\\
\nablaG\bu
&=\bP\nablaM\bu=\sum_{j=1}^n\bP\partial_j\bu\otimes\bP\be_j=\sum_{j=1}^n\bP\be_j\otimes{}\nabla_M\bu_j\,,
\\
 \nablaG\bA&=\sum_{j=1}^n\bP(\partial_j\bA)\bP\otimes\bP\be_j\,.
  \end{aligned}
  \end{equation}
  \end{definition}
  
 In fact, it can be shown \cite{jankuhn2018incompressible} that for points $\bx\in\Gamma$ the covariant derivatives \eqref{nablaG}   of tensors extended from $\Gamma$ are independent of the chosen extension.
 
 The covariant divergence of a vector field $\bu$ and of
a matrix field $\bA$ is defined as:
\begin{align}
 \divG \bu  &= \tr (\gradG \bu)\,,\quad
 \divG \bA  = \left( \divG (\bA^T)_1,\,
               \divG (\bA^T)_2,\,
               \divG (\bA^T)_3\right)^T\,,\quad \bx\in\Omega_\delta\, ,\label{surf_div}
\end{align}
namely the divergence of $\bA$ is computed by rows.
However, because of the cyclic property of traces we have
\begin{align}\label{divergence}
\divG\bu=\tr(\bP\nabla\bu\bP)=\tr (\nabla\bu\bP)=\divM\bu\,,\qquad \divG\bA=\divM\bA\,.
\end{align}
 
\begin{definition}[covariant strain-rate and spin tensors] For a given velocity field $\bv$ on $\Gamma$ we define the   \textit{covariant strain-rate} and \textit{covariant spin} tensors, respectively, as follows:
\begin{equation} \label{strain}
  \bD_\Gamma(\bv)= 
 \frac12(\nabla_\Gamma \bv + \nabla_\Gamma^T \bv)\,,\qquad
 \bW_\Gamma(\bv)=
 \frac12(\nabla_\Gamma \bv - \nabla_\Gamma^T \bv).
 \end{equation}
 \end{definition}
These tensors correspond to symmetric and antisymmetric tangent parts of the instant deformation of a tangent plane due to the flow $\bv$. Essentially, formulas in \eqref{strain} as well as the $\nablaG$ operator represent objects intrinsic to $\Gamma$ which one may compute using the Riemannian structure only. 

Finally we note the relation
of the \VY{external} and covariant rates because of \eqref{nablam} and \eqref{nablaG}:
\begin{align}
    \label{strain_relation}
 \bD_\Gamma(\bv)= \bP\bD_M(\bv)\bP \,,\qquad
 \bW_\Gamma(\bv)=\bP\bW_M(\bv)\bP \, .
\end{align}

\section{Kinematics of Q-tensors on surfaces}\label{kinematics}

In this section we discuss kinematic properties of the model of  surface flows of liquid crystals developed in this paper. The kinematic properties are introduced by defining the dependence of the state variables, which are the momentum and the Q-tensor, on the prescribed deformation of the domain, caused by a vector field $\bv$ in the absence of any forces. The resulting operators, if set equal to zero, are called \textit{passive transport} equations.  For example, the passive transport of a scalar field $f$, e.g. the density, along a tangent flow $\bv$ on a surface $\Gamma$ is usually given by
\begin{align}
    \dot{f}=\pa_t{}f+(\nablaG^Tf)\bv=\pa_t{}f+(\nablaM^Tf)\bv=0\label{scalar_trans}
\end{align}
which attaches scalar values to the flow $\bv$.
Similarly, we will define a notion of surface transport of the linear momentum and a notion of surface  transport of generically oriented Q-tensor fields based on the derivatives introduced in previous sections. While the former is well-established in the literature \cite{jankuhn2018incompressible}, the latter is new. The definition of the Q-tensor passive transport will be motivated by kinematic assumptions with a  clear physical meaning: a tangent eigenvector is embedded into a surface flow similar to the flat case in $\R^2$ and a normal eigenvector has to stay normal along the flow.

\subsection{Momentum transport}

We assume the density $\rho$ is constant and often omitted it in this section for clarity. To express the
rate of change of the linear momentum field $\rho\bu$ in the ambient space, we need to use the Euclidean parallel transport equation \eqref{nablam} of a velocity field $\bu$ normally extended from $\Gamma$:
\begin{equation}\label{cons_mom}
\begin{aligned}
     \pa_t\bu+(\nablaM\bu)\bv&= \pa_t(\bu_T+u_N\bn)+(\nablaM\bu_T+u_N\nablaM\bn+\bn\otimes\nablaM{}u_N)\bv 
     \\
    &=\pa_t\bu_T+(\nablaM\bu_T)\bv + \bn(\pa_t{}u_N+(\nablaM^T{}u_N)\bv)+u_N(\nablaM\bn)\bv 
    \\
    &=\left[\pa_t\bu_T+(\nablaG\bu_T)\bv\right]+\bN(\nablaM\bu_T)\bv+\dot{u}_N\bn+u_N\bB\bv 
        \\
    &=\left[\pa_t\bu_T+(\nablaG\bu_T)\bv\right]-(\bv\cdot\bB\bu_T)\bn+\dot{u}_N\bn+u_N\bB\bv\,,
\end{aligned}
\end{equation}
where we split the velocity $\bu$ into the normal and tangent components as in \eqref{u_decomp}, and used \eqref{Bu} and  Proposition \ref{appendix}.
We consider films which are stationary in space, whence the surface $\Gamma$ does not evolve in time and the velocity $\bu=\bu_T+u_N\bn$ is tangent to $\Gamma$, i.e. $u_N=0$. We express the transport of linear momentum by setting $\bv=\bu_T$ in \eqref{cons_mom}.
  What remains in \eqref{cons_mom} is the \VY{tangential material acceleration} 
  \begin{align}\label{pt}
  \pa_t\bu_T+(\nablaG\bu_T)\bu_T
  \end{align} and the normal centripetal acceleration $-\bu_T\cdot\bB\bu_T$. Since we require $u_N=0$ the centripetal acceleration has to be balanced by the reaction forces which enforce the constraint that the surface does not evolve. This suggests that the tangent part of the rate of change, $\bP(\pa_t\bu+(\nablaM\bu)\bu)=\pa_t\bu_T+(\nablaG\bu_T)\bu_T$, should represent the passive transport of momentum in a surface model. This idea is summarized in the following assumption.
\begin{assumption}[kinematics of momentum]\label{assum_vel}
The passive transport of the linear momentum field $\rho\bu=\rho\bu_T$ along the velocity $\bu_T$
is the parallel transport \eqref{pt}, i.e.:
\begin{align*}
     \pa_t\bu_T+(\nabla_\Gamma\bu_T)\bu_T=0\,.
\end{align*}
\end{assumption}
Based on this kinematic assumption on how the momentum $\rho\bu$ is transported in the absence of any forces we define the surface material derivative, which suits the Assumption \ref{assum_vel}, as follows:
\begin{definition}
\label{vector_transp}
The \textit{surface material derivative} $\dot{\bu}$ of a vector field $\bu:\Gamma\to\R^3$ along a given tangent vector field $\bv$ is a vector field with the following normal and tangent components:
   \begin{align*}
      \bN\dot{\bu}&= \dot{u}_N\bn\,,\quad \bP\dot{\bu}=\pa_t\bu_T+(\nablaG\bu_T)\bv\,.
\end{align*}
\end{definition}

\begin{remark}
    The passive transport $\dot{\bu}=0$ by the surface material derivative given in  Definition \ref{vector_transp} has the following properties.  The normal component $\bu\cdot\bn=u_N$  of the  vector field $\bu$  is transported by \eqref{scalar_trans} as a scalar field. Consequently, if at the initial moment of time the vector field $\bu$ is tangent, then it remains tangent along the passive flow by a vector field $\bv$. The tangent component $\bu_T=\bP\bu$ satisfies the \textit{Riemannian parallel transport} equation
$
\pa_t\bu_T+(\nablaG\bu_T)\bv=0.
$
\end{remark}

Finally we would like to show that the passive transport of the velocity field along itself has the property of preserving the kinetic energy and the linear momentum in the ambient space $\R^3$. We start with a property of the convective term, which is well known in flat domains.

\begin{lemma}[vanishing of the convective term]\label{Temam}
Let $\bu_T$, $\bv_T$, $\bw_T$ be tangent vector field on a closed surface $\Gamma$. Then 
trilinear convective form satisfies
\begin{align*}
    &((\nablaG\bv_T)\bw_T, \bu_T)_\Gamma+((\nablaG\bu_T)\bw_T, \bv_T)_\Gamma=-(\bv_T\cdot\bu_T,\divG\bw_T)_\Gamma\,.
\end{align*}
and it vanishes provided $\divG\bw_T=0$ and $\bu_T=\bv_T$, namely
\begin{align}\label{conv_cancel}
    ((\nablaG\bu_T)\bw_T, \bu_T)_\Gamma=0 \, .
\end{align}
\end{lemma}
\begin{proof} Using Lemma \ref{vector_parts} (covariant integration by parts) followed by \eqref{divergence} and Proposition \ref{appendix} (product rules) for normally extended $\bu^e_T, \bv^e_T, \bw^e_T$ we deduce
 \begin{align*}
      0 &= \big(\divG((\bv_T\cdot\bu_T)\bw_T),1 \big)_\Gamma
      = \big((\bv_T\cdot\bu_T)\divM\bw_T+ \bw_T\cdot\nablaM(\bu_T\cdot\bv_T), 1 \big)_\Gamma
    \\
    &= \big((\bv_T\cdot\bu_T)\divG\bw_T+ \bu_T\cdot(\nablaM\bv_T)\bw_T+\bv_T\cdot(\nablaM\bu_T)\bw_T, 1\big)_\Gamma \, .
\end{align*}
Invoking \eqref{nablaG}, namely $\bP\nablaM\bu=\nablaG\bu$, and reordering yields the assertion.
\end{proof}

\begin{corollary}[preservation of kinetic energy and linear momentum] 
Let $\bu_T$ be a tangent and incompressible velocity field that is passively transported, namely $\bu$ satisfies $\divG\bu_T=0$ and $\pa_t\bu_T+(\nabla_\Gamma\bu_T)\bu_T=0$, for sufficiently smooth initial condition on a closed surface $\Gamma$. Then the kinetic energy $\frac12 \int_\Gamma \rho\bu_T^2$  and the total linear momentum $\int_\Gamma {\rho}\bu_T$ are preserved over time.
\end{corollary}
\begin{proof} Taking into account that density $\rho$ is constant we compute 
\begin{align*} 
    \frac{d}{dt}\int_\Gamma \frac{\rho}2\bu_T^2=\rho\big(\pa_t\bu_T, \bu_T\big)_\Gamma = -\rho\big((\nablaG\bu_T)\bu_T, \bu_T\big)_\Gamma = 0\,,
\end{align*}
according to \eqref{conv_cancel}. To treat the vector of total linear momentum, we consider its $x$-component
\begin{align*}
   \frac{d}{dt}\left(\int_\Gamma {\rho}\bu_T \cdot{} \be_x\right)=\rho \big(\pa_t\bu_T, \be_x\big)_\Gamma=-\rho\big((\nabla_\Gamma\bu_T)\bu_T, \be_x\big)_\Gamma =-\rho\big(\nablaG(\bu_T\cdot\be_x), \bu_T\big)_\Gamma=0
\end{align*}
where used Lemma \ref{vector_parts} (covariant integration by parts) with $\bu=\bu_T$ and $f=\bu_T\cdot\be_x$ at the last step. Other components of $\bu_T$ are dealt with similarly. 
This completes the proof.
\end{proof}

\subsection{Q-tensor transport}
\VY{Although many objective rates are available for modeling of Q-tensor flows even in $\R^3$, we aim to choose one, the corotational derivative \eqref{bulk_corotation}, and show how it should be adapted for the case of a fixed surface $\Gamma$ resulting in Definition \ref{circle_def}.} 
Modeling flows of liquid crystal material in $\R^3$ often involves the following objective rates of change \VY{\cite{sonnet2012dissipative}}, \cite{xiao1998objective} 
\begin{align}\label{bulk_corotation}
\partial_t\bq+(\nabla\bq)\bv-\frac{ (\nabla\bv-\nabla^T\bv)}{2}\bq\,,\quad\partial_t\bQ+(\nabla\bQ)\bv+\bQ \frac{\nabla\bv-\nabla^T\bv}{2}-\frac{\nabla\bv-\nabla^T\bv}{2}\bQ
\end{align}
which are the \textit{corotational derivatives}  of a vector field $\bq$ and of a matrix field $\bQ$ along the flow $\bv$. These {corotational derivatives} express the rate of change of tensors with respect to the (Lagrangian) coordinate system embedded in the fluid, and sometimes they should be chosen over the parallel Euclidean transport to model the physics adequately. For example, one uses the parallel Euclidean transport $\partial_t\bu+(\nabla\bu)\bv$ to express the rate of change of the non-material momentum vector $\rho\bu$; while if one works with the rate of change of the Q-tensor, which provides the statistical description of the liquid crystal orientation, the objective rate \eqref{bulk_corotation} should be used. \LB{ We refer to \cite[Section 10]{nitschke2022observer} for discussion of different time derivatives for vectors and $Q$-tensors on surfaces.}

\begin{remark} \label{emb} A zero corotational derivative \eqref{bulk_corotation} of a vector field $\bq$ means that the vector field is embedded in the flow $\bv$ \VY{in the following kinematical sense \cite{sonnet2012dissipative}}: it is transported parallelly in $\mathbb{R}^3$ by the flow $\bv$ and\VY{, in addition,} is rotated along it by the spin tensor $\frac12(\nabla\bv-\nabla^T\bv)$. 
\LB{To make this point concrete, we consider ${\bf X}(t)$ to satisfy $\frac{d}{dt}{\bf X}(t) = v(t,{\bf X}(t))$. Computing the time derivative of $\tilde{\bq}(t)=\bq(t,{\bf X}(t))$ yields
\[
\frac{d}{dt}\tilde{\bq} = \partial_t\bq+(\nabla\bq)\bv=\frac{ (\nabla\bv-\nabla^T\bv)}{2}\bq = \frac{(\nabla\bv-\nabla^T\bv)}{2}\tilde{\bq}.
\]
}
At the same time, a zero corotational derivative \eqref{bulk_corotation} of a Q-tensor field $\bQ$ means that its eigenframe is embedded in the flow $\bv$: the eigenvectors are transported and rotated as vectors \VY{by \eqref{bulk_corotation}} and the eigenvalues are transported as scalars by \eqref{scalar_trans}.
\end{remark}

In this paper the flow $\bv$ is two-dimensional and is tangent to a surface $\Gamma$, but the Q-tensor is three-dimensional. This  raises the question of what rate of change should be utilized to model the passive motion of Q-tensors along the surface flow.

\LB{
It is natural to try the following transport equation
\begin{align}
    \partial_t\bQ+(\nablaM\bQ)\bv+\bQ \frac{\nablaM\bv-\nablaM^T\bv}{2}-\frac{\nablaM\bv-\nablaM^T\bv}{2}\bQ=0\,,\label{Qtrans_M}
\end{align}
 where all the derivatives correspond to \VY{Cartesian} derivatives of the normal extensions of arguments. Unfortunately, the Q-tensor of type $\bn\otimes\bn-\frac13\bI$ is not in the kernel of the operator \eqref{Qtrans_M}. In fact, using \eqref{e-n}, \eqref{surf_dir_der} and Lemma \ref{Gauss} we compute 
\begin{align*}
  0&= (\nabla_M(\bn\otimes\bn))\bv+(\bn\otimes\bn) \frac{\nablaM\bv}{2}+\frac{\nablaM^T\bv}{2}(\bn\otimes\bn)
   \\
   &=
   \sum_{j=1}^n(\bv\cdot\bP\be_j)\partial_j(\bn\otimes\bn)+\frac12\left(  \bn\otimes({\nablaM^T\bv})\bn +(\nablaM^T\bv)\bn\otimes\bn\right)
   \\
   &=(\nablaM\bn)\bv\otimes\bn+(\nablaM\bn)\otimes\bB\bv
 +\frac12\left(  \bn\otimes(-\bB\bv) +(-\bB\bv)\otimes\bn\right)
 =\frac12{}\left(\bB\bv\otimes\bn+\bn\otimes\bB\bv\right)
\end{align*}
whence $\bB=0$ or, equivalently, the surface has to be flat. We have just shown that transport of a 
Q-tensor by \eqref{Qtrans_M} does not keep alignment with respect to the normal direction.
    These considerations suggest that there are many ways to define the transport of a Q-tensor field and that we need to choose one based on some kinematic assumptions similar to Assumption \ref{assum_vel}.
    In fact, to transport a three-dimensional tensor we need a three dimensional spin tensor which cannot be provided by a two-dimensional surface flow. 
    }
    
    At a point $\bx\in\Gamma$ consider an orthonormal basis $\bt_1,\bt_2, \bn$ of $\R^3$ where the first two vectors belong to the tangent plane of $\Gamma$. 
     A general three-dimensional skew-symmetric {\it spin tensor} $\bW$ can be represented in this local basis with the help of a tangent vector $\bw_T=(w_{2}\bt_1-w_{1}\bt_2)$ and a tangent tensor $\bW_T=w_{3}(\bt_2\otimes\bt_1-\bt_1\otimes\bt_2)$:
     \begin{equation}\label{W}
     \begin{aligned}
         \bW&=w_{3}(\bt_2\otimes\bt_1-\bt_1\otimes\bt_2)+w_{2}(\bt_1\otimes\bn-\bn\otimes\bt_1)+w_{1}(\bn\otimes\bt_2-\bt_2\otimes\bn)
          \\
          &=\bP\bW\bP +(w_{2}\bt_1-w_{1}\bt_2)\otimes\bn-\bn\otimes(w_{2}\bt_1-w_{1}\bt_2)=\bW_T + \bw_T\otimes\bn-\bn\otimes\bw_T \, .
     \end{aligned}
     \end{equation}
     
     Motivated by the three-dimensional corotational derivatives \eqref{Qtrans_M} we are in position to determine the structure of the \textit{surface corotational derivatives of Q-tensors}. Consider a vector field $\bv$ tangent to $\Gamma$, a matrix field $\bQ$ and its eigenvector field $\bq$. A surface corotational derivative along the flow $\bv$ can be expected to be given by
     \begin{align}\label{corotation}
\overset{\circ}{\bq}&=\partial_t\bq+(\nablaM\bq)\bv-\bW\bq\,,
\qquad
\overset{\circ}{\bQ}=\partial_t\bQ+(\nablaM\bQ)\bv+(\bQ\bW-\bW\bQ) \, ,
\end{align}
where the spin tensor $\bW$ is yet to be defined via ${\bw_T}$ and ${\bW_T}$ in \eqref{W}. 
For completeness, we define the corotational derivative of a scalar field by its material derivative \eqref{scalar_trans}, $$\overset{\circ}{f}=\dot{f}.$$ We specify $\bW_T$ and $\bw_T$ in \eqref{W} by making the following assumption.

\begin{assumption}[kinematics of Q-tensors]\label{assum_Q}  The tangent vector $\bw_T$ and the tangent spin tensor $\bW_T$ are such that the normal eigenvector $\bq=\bn$ of a conforming Q-tensor  is in the kernel of the passive transport operator \eqref{corotation} along a tangent flow $\bv=\bu_T$. Also, the passive transport of a tangent eigenvector $\bq=\bt$ of a conforming Q-tensor field  is a combination of the parallel transport \eqref{pt} and the instant rotation by the covariant spin tensor \eqref{strain}:
\begin{align*}
    \overset{\circ}{\bn}=0\,,\qquad \overset{\circ}{\bt}=  \partial_t\bt+(\nablaG\bt)\bu_T-\frac12(\nablaG\bu_T-\nablaG^T\bu_T)\bt\,.
\end{align*}
\end{assumption}
In view of Remark \ref{emb}, $\overset{\circ}{\bt}=0$ means that the tangent vector field $\bt$ is also embedded in the flow $\bu_T$ but this time in the sense of the Riemannian structure on $\Gamma$. \VY{This rate of change $\overset{\circ}{\bt}$ of a tangent vector field $\bt$ is known as \textit{surface Jaumann derivative} \cite{nitschke2022observer}.} From this modelling assumption on how the eigenframe of a conforming Q-tensor is transported, we immediately find what should $\bw_T$ and $\bW_T$ be like.

\begin{lemma}[characterization of spin tensor] In order to satisfy Assumption \ref{assum_Q}, the spin tensor $\bW$ in \eqref{corotation} should be given by \eqref{W} with  
\begin{align}\label{r1}
{\bw_T}=\bB\bu_T\,, \qquad {\bW_T}=\frac12(\nablaG{\bu_T}-\nablaG^T{\bu_T})
\end{align}
\end{lemma}

\begin{proof}Indeed, using the structure of \eqref{corotation} and the definition of the shape operator $\bB=\nablaM\bn$ we compute
    \begin{align*}
\overset{\circ}{\bn}&=\partial_t\bn+(\nablaM\bn)\bu_T-\bW\bn=\bB\bu_T-\bw_T
\end{align*}
and, because $\overset{\circ}{\bn}$ should vanish, we have to have $\bw_T=\bB\bu_T$ in \eqref{W}. Similarly, using Definition~\ref{vector_transp}, we compute
\begin{align*}
\overset{\circ}{\bt}&=\partial_t\bt+(\nablaM\bt)\bu_T-\bW\bt= \partial_t\bt+(\nablaG\bt)\bu_T+\bN(\nabla\bt)\bu_T-\bW\bt
\\
&=\dot{\bt}+\bN(\nabla\bt)\bu_T-\bW_T\bt +(\bw_T\cdot\bt)\bn=\dot{\bt}-\bW_T\bt-(\bn\otimes\bB\bt)\bu_T +(\bw_T\cdot\bt)\bn
=\dot{\bt}-\bW_T\bt
\end{align*}
which leads to the choice $\bW_T=\bW_\Gamma(\bu_T)$ to suit Assumption \ref{assum_Q}. Consequently, the spin tensor \eqref{W} can be expressed in terms of the covariant spin tensor $\bW_\Gamma$ from \eqref{strain} or of the \VY{external} spin tensor from $\bW_M$ \eqref{strainM} as follows (using \eqref{skewBu} in the last equality)
 \begin{align}\label{WM-W}
         \bW&=\bW_\Gamma({\bu_T})+\bB{\bu_T}\otimes\bn-\bn\otimes\bB{\bu_T}=\bW_\Gamma({\bu_T}) + \bW_*({\bu_T})=\bW_M({\bu_T}) + \frac12\bW_*({\bu_T})
\end{align}
where the \textit{star spin tensor} $\bW_*$ is defined for a given velocity $\bu_T$ by
\begin{align}\label{star_spin}
\bW_*({\bu_T}) := \bB{\bu_T}\otimes\bn-\bn\otimes\bB{\bu_T}\,.
\end{align}
This concludes the proof.
\end{proof}
Based on the Assumption \ref{assum_Q} (kinematics of Q-tensors) on how the eigenframe of a conforming Q-tensor is transported in the absence of any forces we define the surface corotational derivative of \textit{general} Q-tensors in the following definition.
\begin{definition}[\VY{external surface corotational
derivative}] \label{circle_def} The \textit{surface corotational derivatives} $\overset{\circ}{\bq}$ of a vector field $\bq:\Gamma\rightarrow\R^3$ and $\overset{\circ}{\bQ}$  of a matrix field $\bQ:\Gamma\rightarrow\R^{3\times3}$ along a tangent vector field $\bv$ are given by \looseness=-1
\begin{align}\label{corotation_defintiion}
\overset{\circ}{\bq}&=\partial_t\bq+(\nablaM\bq)\bv-\left(\bW_\Gamma(\bv)+\bW_*(\bv)\right)\bq\,,\\
\label{corotation_defintiion2}
\overset{\circ}{\bQ}&=\partial_t\bQ+(\nablaM\bQ)\bv+\bQ\left(\bW_\Gamma(\bv)+\bW_*(\bv)\right)-\left(\bW_\Gamma(\bv)+\bW_*(\bv)\right)\bQ
\end{align}
\end{definition}
Clearly, the structure of \VY{Cartesian} corotational derivative \eqref{bulk_corotation} can be recognized in Definition \ref{circle_def} but  the special spin tensor \eqref{WM-W} is used because the domain of definition of all objects is a surface, and the flow $\bv$ is two-dimensional while the Q-tensor is three-dimensional.

 \subsection{Properties of the surface corotational derivative}
 
 In this section we characterize the surface corotational derivative $\overset{\circ}{\bq}$ of general vector fields $\bq$ using the splitting \eqref{u_decomp}. We also explain the structure of the surface corotational derivative $\overset{\circ}{\bQ}$ of a Q-tensor $\bQ$ by relating the passive transport equation $\overset{\circ}{\bQ}=0$ to the passive transport equations $\overset{\circ}{\bq}_i=0$, $\overset{\circ}{\lambda}_i=0$ of the eigenvectors and eigenvalues of $\bQ$.
 
 We first present some intuitive properties of the vector \eqref{corotation_defintiion} and the matrix \eqref{corotation_defintiion2} surface corotational derivatives. It is shown that if a vector field $\bq=q_N\bn+\bq_T$ satisfies the passive transport equation $\overset{\circ}{\bq}=0$ then the normal component $q_N$ is transported by \eqref{scalar_trans} as a scalar field while the tangent component  $\bq_T$ undergoes a combination of the parallel transport \eqref{pt} and the instant rotation by the covariant spin tensor \eqref{strain} (the tangent component is embedded in the tangent two-dimensional flow). At the same time, we recover the usual meaning of the  corotational derivative of matrices but for the case of a surface
 (see Remark \ref{emb}): if $\overset{\circ}{\bQ}=0$ then the eigenvalues and the eigenvectors are embedded (in the sense described above) into the flow along a surface. 
 
 We start with the basic properties of corotational derivatives in the next lemma.
 
\begin{lemma}[properties of corotational derivatives]\label{co_properties}
 The surface corotational derivative \eqref{corotation_defintiion} along a tangent flow $\bv$ has the following distributive properties for vector fields $\ba, \bb$ and matrix field $\bA$
 \begin{align*}
      \overset{\circ}{(\ba\cdot\bb)}&=\overset{\circ}{\ba}\cdot\bb+\overset{\circ}{\bb}\cdot\ba\,,
   \quad
   \overset{\circ}{(f\bA)}=\dot{f}\bA+f\overset{\circ}{\bA}
   \,,
   \quad
    \overset{\circ}{(\ba\otimes\bb)}=\overset{\circ}{\ba}\otimes\bb+{\ba}\otimes\overset{\circ}{\bb}\,.
 \end{align*}
\end{lemma}
\begin{proof} The properties follow from Lemma \ref{appendix}, \eqref{nablam}, \eqref{WM-W} and $\bv=\bP\bv$:
\begin{align*}
      \overset{\circ}{(\ba\cdot\bb)}&=\partial_t\ba\cdot\bb+\ba\cdot\partial_t\bb+\bv\cdot((\nablaM\ba)^T\bb+(\nablaM\bb)^T\ba)
   =(\partial_t\ba+(\nablaM\ba)\bv)\cdot\bb+\ba\cdot(\partial_t\bb+(\nablaM\bb)\bv)
    \\
   &=(\overset{\circ}{\ba}+\bW\ba)\cdot\bb+(\overset{\circ}{\bb}+\bW\bb)\cdot\ba=\overset{\circ}{\ba}\cdot\bb+\overset{\circ}{\bb}\cdot\ba+\ba\cdot(\bW+\bW^T)\bb=\overset{\circ}{\ba}\cdot\bb+\overset{\circ}{\bb}\cdot\ba
   \\
   \overset{\circ}{(f\bA)}&=\pa_t(f\bA)+\nablaM(f\bA)\bv + (f\bA)\bW-\bW(f\bA)   \\
   &=\pa_tf\bA+f\pa_t\bA+(f\nablaM\bA+\bA\otimes\nablaM f)\bv+f(\bA\bW-\bW\bA)
=\dot{f}\bA+f\overset{\circ}{\bA}
   \\
    \overset{\circ}{(\ba\otimes\bb)}&=\pa_t\ba\otimes\bb+\ba\otimes\pa_t\bb+\nablaM(\ba\otimes\bb)\bv+(\ba\otimes\bb)\bW-\bW(\ba\otimes\bb)
    \\
    &=(\pa_t\ba+(\nablaM\ba)\bv)\otimes\bb+\bb\otimes(\pa_t\bb+(\nablaM\bb)\bv)-\ba\otimes\bW\bb-\bW\ba\otimes\bb=\overset{\circ}{\ba}\otimes\bb+{\ba}\otimes\overset{\circ}{\bb}
\end{align*}
This concludes the proof.
\end{proof}
Assumption \ref{assum_Q} (kinematics of Q-tensors) dictates how the eigenframe of a conforming Q-tensor field is transported by a tangent flow. In the following lemma we characterize the passive transport $\overset{\circ}{\bq}=0$ of an eigenvector which is neither normal nor tangent to $\Gamma$.

\begin{lemma}[corotational derivative of a vector]\label{charact_corot}
The surface corotational derivative $\overset{\circ}{\bq}$ of a vector field $\bq:\Gamma\to\R^3$ is a vector field with the following normal and tangent components:
   \begin{align*}
      \bN\overset{\circ}{\bq}= \dot{q}_N\bn\,,\quad \bP\overset{\circ}{\bq}=\dot{\bq}_T-\bW_\Gamma\bq_T
\end{align*}
\end{lemma}
  \begin{proof} We recall Definition \ref{vector_transp} and $\eqref{scalar_trans}$. Since $\bv=\bP\bv$, we compute
\begin{align*}
\overset{\circ}{\bq}&=\partial_t(\bq_T+q_N\bn)+(\nablaM\bq_T+\bn\otimes\nabla_\Gamma{}q_N+q_N\bB)\bv-\bW_\Gamma\bq_T - q_N\bB\bv+(\bB\bv\cdot\bq_T)\bn
\\
&=(\partial_t\bq_T+\bP(\nablaM\bq_T)\bv)+(\partial_t{}q_N+(\nabla_\Gamma{}q_N)\cdot\bv)\bn+(\bN\nablaM\bq_T)\bv-\bW_\Gamma\bq_T +(\bB\bv\cdot\bq_T)\bn
\\
&=\bP(\dot{\bq}_T-\bW_\Gamma\bq_T)+(\dot{q}_N+\bn\cdot(\nablaM\bq_T)\bv+\bB\bv\cdot\bq_T)\bn=\bP(\dot{\bq}_T-\bW_\Gamma\bq_T)+\bN(\dot{q}_N\bn)
     \end{align*}  
where we used symmetry of $\bB$, $\bW_\Gamma=\bP\bW_\Gamma\bP$ and \eqref{Gauss} in the last step.
\end{proof}

\begin{remark}
    The passive transport $\overset{\circ}{\bq}=0$ of vector fields by the surface corotational derivative given in  Definition \ref{circle_def} has the following properties. If at the initial moment of time a vector field $\bq$ is tangent, then it remains to be tangent along the passive flow by a vector field $\bv$. The tangent component $\bq_T$ is subjected to the parallel transport \eqref{pt} and the instant rotation by the covariant spin tensor \eqref{strain} embedded in the flow along $\bv$). Also, the normal component $\bq\cdot\bn$  of a non-tangent vector field $\bq$  is transported by \eqref{scalar_trans} as a scalar field.
\end{remark}

Finally, we characterize the corotational transport  \eqref{corotation_defintiion2} of a general Q-tensor, namely non-conforming to $\Gamma$, via its eigenframe: the eigenvectors and the eigenvalues of a Q-tensor are embedded into the flow and are passively transported in the sense of the $\circ$ operator.
 
 \begin{theorem}[corotational derivative of a  tensor]\label{transport_decomp_Q}
 Given 
 a symmetric matrix field $\bQ\in{}C^1(\Gamma\times(0,T))^{3\times3}$ consider a point $\bx\in\Gamma$ and its neighborhood $U(\bx)\subset\Gamma$ such that 
there exists a spectral decomposition \eqref{spectral} with eigenvalues $\lambda_i\in{}C^1(U(\bx)\times(0,T))$ and the corresponding unit-length eigenvectors $\bq_i\in{}C^1(U(\bx)\times(0,T))^3$, $i=1, 2, 3$. If in $U(\bx)\times(0,T)$
\begin{itemize}
    \item all eigenvalues are distinct then
\begin{align*}
    \overset{\circ}{\bQ}=0\qquad \Longleftrightarrow{}\qquad  \overset{\circ}{\bq}_i=0\, ,\quad{}\dot{\lambda}_i=0\,,\quad i=1, 2, 3;
\end{align*}
    \item two eigenvalues are equal but distinct from the third one with eigenvector $\bq_m$ then
\begin{align*}
    \overset{\circ}{\bQ}=0\qquad \Longleftrightarrow{}\qquad  \overset{\circ}{\bq}_m=0\, ,\quad{}\dot{\lambda}_i=0\,,\quad i=1, 2, 3;
\end{align*}
   \item all eigenvalues are equal then
\begin{align*}
    \overset{\circ}{\bQ}=0\qquad \Longleftrightarrow{}\qquad \dot{\lambda}_i=0\, ,\quad{}i=1, 2, 3\,.
\end{align*}
\end{itemize}
\end{theorem}
\begin{remark}
Three cases appear in the statement because eigenvectors are not determined uniquely if some of the eigenvalues coincide. For example, in the case of two equal eigenvalues, one can discuss the transport of the corresponding planar eigenspace, but eigenvectors that form this eigenspace may undergo arbitrary deformations without affecting $\overset{\circ}{\bQ}=0$.
\end{remark}

\begin{proof} We start with the case of distinct eigenvalues. The spectral decomposition $\bQ(\bx)=\sum_{k=1}^3\lambda_k(\bx)(\bq_k(\bx)\otimes\bq_k(\bx))$ holds for all $\bx\in U(\bx)$ and $\bq_1(\bx), \bq_2(\bx), \bq_3(\bx)$ form an orthonormal basis. We take the corotational derivative and apply its properties from Lemma \ref{co_properties}:
\begin{align*}
  \overset{\circ}{\bQ}&=   \sum_{k=1}^3\left(\bq_k\otimes\bq_k\right)\dot{\lambda}_k+\lambda_k\left(\overset{\circ}{\bq_k}\otimes\bq_k+\bq_k\otimes\overset{\circ}{\bq_k}\right)
  \end{align*}
  from where the sufficiency follows immediately. 
  To show the necessity we contract the result with $\bq_i$ from the right and then with $\bq_j$ from the left:
  \begin{align*}
  \overset{\circ}{\bQ}\bq_i&=   \lambda_i\overset{\circ}{\bq}_i+\dot{\lambda}_i\bq_i+\sum_{k=1 ,k\neq{}i}^3\lambda_k(\overset{\circ}{\bq_k}\cdot\bq_i)\bq_k
  \\
  \bq_j\cdot\overset{\circ}{\bQ}\bq_i&=   \lambda_i\bq_j\cdot\overset{\circ}{\bq}_i+\dot{\lambda}_i\bq_j\cdot\bq_i+\sum_{k=1 ,k\neq{}i}^3\lambda_k(\overset{\circ}{\bq_k}\cdot\bq_i)(\bq_j\cdot\bq_k)
  \end{align*}
  where we used the identity $(\overset{\circ}{\bq_j}\cdot\bq_i)+(\overset{\circ}{\bq_i}\cdot\bq_j)=0$. Consider the diagonal, $i=j$, and off-diagonal, $i\neq{}j$, contractions separately:
  \begin{align*}
   \bq_i\cdot\overset{\circ}{\bQ}\bq_i&=   \dot{\lambda}_i\,,\qquad{} \bq_j\cdot\overset{\circ}{\bQ}\bq_i=   (\lambda_i-\lambda_j)\bq_j\cdot\overset{\circ}{\bq}_i
   \end{align*}
   If $\overset{\circ}\bQ=0$ and $\lambda_i\neq{}\lambda_j$ then $\dot{\lambda}_i=0$ and all the projections of $\overset{\circ}{\bq}_i$ on the basis vectors $\bq_1,\bq_2,\bq_3$ are zero.
   
   The case of two equal eigenvalues $\lambda_i=\lambda_j=\lambda$ is similar. We rearrange the spectral decomposition
\begin{align*}
    \bQ&=\lambda_m(\bq_m\otimes\bq_m)+\lambda(\bq_i\otimes\bq_i +\bq_j\otimes\bq_j)
    =\lambda_m(\bq_m\otimes\bq_m)+\lambda(\bI-\bq_m\otimes\bq_m)
    \\&=\lambda\bI+(\lambda_m-\lambda)\bq_m\otimes\bq_m
\end{align*}
compute the corotational derivative and contract it with $\bq_m$ from the right and from the left
\begin{align*}
  \overset{\circ}{\bQ}&=   \dot{\lambda}\bI+(\dot{\lambda}_m-\dot{\lambda})(\bq_m\otimes\bq_m)+(\lambda_m-\lambda)\left(\overset{\circ}{\bq}_m\otimes\bq_m+\bq_m\otimes\overset{\circ}{\bq}_m\right)
    \\
   \overset{\circ}{\bQ}\bq_m&=   \dot{\lambda}_m\bq_m+(\lambda_m-\lambda)\overset{\circ}{\bq}_m
   \,,\qquad\bq_m\overset{\circ}{\bQ}\bq_m=   \dot{\lambda}_m
    \end{align*} 
   which shows the equivalence. The case of three equal eigenvalues is trivial.
   \end{proof}
  
  \begin{remark}
  The preceding surface corotational derivatives $\overset{\circ}{\bq}$ and $\overset{\circ}{\bQ}$ have a physically intuitive explanation: they correspond to the three-dimensional,  \VY{Cartesian} corotational derivatives \eqref{bulk_corotation} along the special \textit{rotational extension}  $\bv^r$ of a surface flow $\bv$ from the surface $\Gamma$ to a bulk three-dimensional neighborhood of it. This rotational extension $\bv^r$  is described below.
  
 Consider the tangential projector $\bP$, the normal projector $\bN$ and the shape operator $\bB$ in the basis of principal directions  $\bt_1\,,\bt_2\,,\bn$ of $\Gamma$:
\begin{align*}
\bP&=\bt_i\otimes\bt_i\,,\quad \bN=\bn\otimes\bn\,,\quad\bB=\nabla_\Gamma\bn=\kappa_i\bt_i\otimes\bt_i
\end{align*}
where $\kappa_i$ are principle curvature fields on $\Gamma$. We may extend  a tangent velocity field $\bv=v^i\bt_i$ from $\Gamma$ to $\Omega_\delta$ in the normal direction \textit{constantly} or \textit{rotationally}:
\begin{align*}
    \bv^e&=(v^{i})^e\bt_i^e\,,\qquad{}
    \bv^r=(1+d\kappa_i^e)(v^{i})^e\bt^e_i=\bv^e+d\,\bB^e\bv^e
     \end{align*}
where $d$ is the signed distance to $\Gamma$. In other words, if a surface velocity $\bv$ is aligned with a principle direction, then its rotational extension $\bv^r$ changes linearly away from the center of curvature; if the velocity has two nonzero components along the surface principle directions, then the rotational extensions act on these components separately. \LB{ Using Proposition \ref{appendix} (product rules), the bulk gradient of $\bv^r$ evaluated on $\Gamma$ is
\[ \nabla\bv^r = \nabla \bv^e + (\bB^e\bv^e)\otimes \nabla d  +d \nabla(\bB^e\bv^e) = \nabla_M \bv +(\bB\bv)\otimes \bn .\]
Hence, $\frac{1}{2}(\nabla \bv^r - \nabla^T \bv^r)$ on $\Gamma$ is $\frac{1}{2}(\nabla \bv^r - \nabla^T \bv^r) = \bW_M +\bW_* = \bW.$}
Essentially, the surface corotational derivatives \eqref{corotation_defintiion} and \eqref{corotation_defintiion2} provide the same rate of change as the bulk corotational derivatives \eqref{bulk_corotation} evaluated on $\Gamma$ in which $\bv$ is set as the rotational extension $\bv^r$. \LB{ Although the surface corotational derivatives \eqref{corotation_defintiion} and \eqref{corotation_defintiion2} are independent of the normal extension, the bulk corotational derivatives \eqref{bulk_corotation} do depend on the extension, and $\bv^r$ is one particular choice of bulk extension such that \eqref{bulk_corotation} coincides with \eqref{corotation_defintiion} and \eqref{corotation_defintiion2}.} One reason it is called the rotational extension is that this extension is more physical for rotational velocity fields. For example, if we consider $\Gamma$ to be the unit circle with prescribed velocity $\bv = (\cos\theta, \sin\theta)$ on $\Gamma$, where $\theta$ is the polar angle, then $\Div{\bv^e} \neq0$ while $\Div{\bv^r} =0$ in this special case in two dimensions. Although $\Div{\bv^r} =0$ will not hold for more general surfaces, the above reasoning explains why $\bv^r$ is called the rotational extension.
  \end{remark}

 \section{Derivation of surface Beris--Edwards model} \label{Onsager}

We develop a model of fluidic liquid crystal films following \cite{yang2016hydrodynamic},\cite{unpublishedWang}, \cite{eck2009phase} and \cite{nochetto2014diffuse}. \VY{The modeling approach chosen in this paper is the so-called \textit{generalized Onsager principle} \cite{Wang2021}, \cite{doi2011onsager}} which is used as a tool.
\VY{The principle is formulated without referring to dimensions of the system and its environment, and it appears to be suitable for modelling of embedded surfaces. We refer to \cite{doi2011onsager} for the principle's thermodynamical premises and to  \cite{doi2015onsager}, \cite{wang2021onsager} for further details of its application to particular physical systems. The generalized Onsager principle}
is not an extremal principle which only needs a constitutive relation to complete the model (e.g. of an elastic body)  but rather a sequence of predetermined steps which guide the creation of a model with a thermodynamically consistent energy structure based on predetermined kinematic properties. We briefly outline these abstract steps \VY{(also see \cite{wang2021onsager}, Section~2.3)} as they should be applied to adapt the classical Beris--Edwards model \cite{beris1994thermodynamics} in flat domains to the case of curved surfaces $\Gamma$.

\begin{itemize}
    \item[Step 1:] \textit{Kinematics}. We choose the state variables of the forthcoming thermodynamical system on $\Gamma$ to be the tangent momentum field $\rho\bu$ and the Q-tensor field $\bQ$; density $\rho$ is constant. We postulate that the kinematics of the system are dictated by the surface material derivative $\dot{\bu}$ and the surface corotational derivative $\overset{\circ}{\bQ}$ given by Definitions \ref{vector_transp} and \ref{circle_def};
    
    \item[Step 2:]\textit{Energy landscape.} We define the total energy $E_{\rm{total}}$ of the system to be the sum of the kinetic energy and the Landau--de\,Gennes energy, and express its time derivative in terms of the rates $\overset{\circ}{\bQ}$ and $\dot{\bu}$ from Step 1;
    
    \item[Step 3:]\textit{Evolution laws.} We propose a suitable structure of the evolution laws involving the rates $\overset{\circ}{\bQ}$, $\dot{\bu}$ from Step 1 and several thermodynamical quantities (generalized forces) yet-to-be-determined. We split the latter into reversible and dissipative forces;
    
    \item[Step 4:] \textit{Reversible quantities}. The generalized reversible forces are responsible for the Hamiltonian structure of the system, whence they do not change the energy from Step 2 over time. Exploiting this fact, we determine these reversible quantities;
    
    \item[Step 5:]\textit{Dissipative quantities}. The generalized dissipative forces are responsible for the total energy decrease over time. We propose a nonequilibrium dissipative process in terms of a suitable least action principle, which identifies these dissipative quantities.
\end{itemize}

We apply in detail the five steps of the generalized Onsager principle in Section \ref{gen_ons_pr}, and summarize the resulting \textit{surface Beris--Edwards} model in Section \ref{surf_BE}, which satisfies a natural energy decay over time dictated by the dissipative quantities, namely
\[
\frac{d}{dt} E_{\rm{total}} \le 0.
\]

\subsection{Generalized Onsager principle} \label{gen_ons_pr}

In this section we apply the generalized Onsager principle following the steps outlined above. After the last step the resulting system is simplified using Lemma \ref{matrix_div}.


\medskip
{\it Step 1: Kinematics.}
We start by assuming that  a fluidic liquid crystal film is a nonequilibrium thermodynamical system on a stationary surface $\Gamma$ described by two state variables, namely the tangent incompressible velocity field $\bu$ and the three-dimensional symmetric and traceless Q-tensor field $\bQ$:
\VY{$$\bQ=\bQ(\bx,t) \,,\quad \bu=\bu(\bx,t)\,,\qquad \bx\in \Gamma.$$ 
}
Moreover, we describe the kinematical properties of $\bu$ and $\bQ$ via the surface material derivative $\dot{\bu}$ \VY{(or acceleration)} and the surface corotational derivative $\overset{\circ}{\bQ}$ introduced in Definitions \ref{vector_transp} and \ref{circle_def}. Such definitions are consistent with transport in the absence of any forces (passive transport), in which case they reduce to $\overset{\circ}{\bQ}=0$ and $\dot{\bu}=0$.

\medskip
{\it Step 2: Energy landscape.}
We postulate that the thermodynamical system possesses a total energy $E_{\rm{total}}=E_{LdG}+K$, given by the following Landau--de\,Gennes energy $E_{LdG}[\bQ,\nabla_M\bQ]$ and kinetic energy $K[\bu]$:
 \begin{align}\label{eldg} 
 &E_{LdG}[\bQ,\nabla_M\bQ]=\int_{\Gamma}e[\bQ,\nabla_M\bQ] :=\int_{\Gamma}\frac{L}2|\nabla_M\bQ|^2+\int_{\Gamma}F[\bQ]  
 \,,\qquad{}K[\bu]=\int_{\Gamma}\frac{\rho}2{|\bu|^2},
 \end{align}
 where $|\nabla_M\bQ|^2$ is the \textit{surface Frank energy} \cite{golovaty2017dimension} and $F[\bQ]$ is the double-well potential \eqref{dw}. \VY{More complicated forms of elastic energy $E_{LdG}$ can be postulated here, but we choose to consider the one-constant model of energy for the ease of presentation.} To compute the rate of change of the total energy, \LB{we use the fact that $\Gamma$ is a closed, time-independent surface and $\bu$ is tangential. This can be viewed as an application of Leibniz formula \cite[Lemma 2.1]{dziuk2007finite}. The resulting  change in total energy is}
 \begin{align*}
 \frac{d}{dt} E_{\rm{total}}[\bQ,\nablaM\bQ,\bu] = \int_{\Gamma}\left(\frac{\pa{}e}{\pa{}t}[\bQ,\nabla_M\bQ] +\frac{1}2\frac{\pa{}{(\rho\bu^2)}}{\pa{}t}\right)\, .
 \end{align*}
 We simply write $\int_\Gamma\rho\bu\cdot\pa_t\bu$ for the second term, while for the first term we have
 \begin{align*}
     \int_\Gamma\frac{\pa{}e}{\pa{}t}[\bQ,\nablaM\bQ]=\int_\Gamma\frac{\pa{}e}{\pa{}\bQ}:\frac{\pa{}\bQ}{\pa{}t}+\frac{\pa{}e}{\pa{}(\nablaM\bQ)}\tripledot\frac{\pa{}(\nablaM\bQ)}{\pa{}t},
 \end{align*}  
 where we recall the notation \eqref{contractions} for the contraction `$\tripledot$'. Commuting $\partial_t$ and $\nablaM$, because $\Gamma$ is stationary, and using Corollary \ref{3tensors_parts} \RHN{(external integration by parts)} yields for any matrix field $\bC$
\begin{equation*}
 \int_\Gamma \frac{\pa{}e}{\pa{}(\nablaM\bQ)}\tripledot\nablaM\bC 
 = \int_\Gamma L \nablaM\bQ \tripledot\nablaM\bC 
 = - \int_\Gamma L \divM\big(\nablaM\bQ \big) : \bC.
\end{equation*}
This implies
\begin{equation}\label{dens_rate} 
      \int_\Gamma\frac{\pa{}e}{\pa{}t}[\bQ,\nablaM\bQ]
      = \int_\Gamma \Big(- L \divM\big(\nablaM\bQ \big) + F'[\bQ] \Big):\pa_t\bQ
      = - \int_\Gamma \bH : \pa_t\bQ \, ,
 \end{equation}
 where the \textit{molecular field} $\bH$ is the traceless symmetric matrix
  \begin{equation}\label{def-H}
  \bH=\mathcal{P}\bigg(L\divM\nablaM\bQ-F'[\bQ]\bigg) \, ,
  \end{equation}
  and $\mathcal{P}$ is the projection operator on the subspace of symmetric and traceless matrices. Note that $\bH$ here is neither conforming nor flat-degenerate in the normal direction, in the sense of definitions \eqref{conf_def}  and \eqref{fd_def}, because $\bQ$ is general.
 
  We next intend to express $\frac{d}{dt} E_{\rm{total}}$ in terms of the surface material derivative $\dot{\bu}$ and surface corotational derivative $\overset{\circ}{\bQ}$, or equivalently to substitute $\pa_t\bu$ and $\pa_t\bQ$ by $\dot{\bu}$ and $\overset{\circ}{\bQ}$. To this end, we
  recall the kinematical properties from Definition~\ref{vector_transp} and Definition~\ref{circle_def},
  \begin{equation}\label{time-derivatives}
      \pa_t\bu=\dot{\bu}-(\nabla_\Gamma\bu)\bu\,,\qquad{}\pa_t\bQ=\overset{\circ}{\bQ}-(\nabla_M\bQ)\bu+\bS,
  \end{equation}
where \eqref{WM-W} and \eqref{star_spin} are used to split the tensor $\bS=\bS[\bu, \bQ] = \bW(\bu) \bQ - \bQ \bW(\bu)$ as follows:
    \begin{align}\label{S_def}
       \bS:=\bS_\Gamma+\bS_*\,,\qquad{}\bS_\Gamma:=\bW_\Gamma(\bu)\bQ-\bQ\bW_\Gamma(\bu)\,,\quad{}\bS_*:=\bW_*(\bu)\bQ-\bQ\bW_*(\bu)\,. 
      \end{align}
This, together with the fact that $(\bu, (\nabla_\Gamma\bu)\bu)_{\Gamma}=0$ according to \eqref{conv_cancel}, yields the following expression for the rate of change of total energy
  \begin{equation}\label{total-energy-rate}
  \begin{aligned}
  \frac{d}{dt} E_{\rm{total}}[\bQ,\nablaM\bQ,\bu] &= \int_{\Gamma}\big(-\bH:\pa_t\bQ +\rho\bu\cdot\pa_t\bu\big)
 \\
 &=-\left(\bH,\overset{\circ}{\bQ}\right)_{\!\!\Gamma}+\bigg
 (\rho\bu, \dot{\bu}\bigg)_{\!\!\Gamma}+\bigg(\bH,(\nabla_M\bQ)\bu-\bS[\bu,\bQ]\bigg)_{\!\!\Gamma}.
 \end{aligned}
 \end{equation}

\medskip
{\it Step 3: Evolution laws.}
Following the classical Beris--Edwards model in $\R^3$ we postulate that the surface model is driven by abstract evolution equations on the surface $\Gamma$ with the kinematics derived from Assumptions \ref{assum_vel} (kinematics of momentum) and \ref{assum_Q} (kinematics of Q-tensors). Note that only the structure is postulated while the required new quantities ($\ba^d_T$, $\bF^d$, $\ba^r$, $\bbf_T^r$) are yet-to-be-determined. We formulate these abstract evolution laws as follows.
\vskip 0.2cm
\begin{assumption}[evolution laws]\label{assum_Onsager}
The thermodynamics of the surface Beris--Edwards model has both a dissipative and a Hamiltonian  structure \cite{yang2016hydrodynamic}. The dissipative structure is due to a symmetric tangent stress $\ba^d_T=\bP\ba_T^d\bP$ and a symmetric tensor $\bF^d$. The  reversible Hamiltonian structure is due to a skew-symmetric stress $\ba^r$ and a tangent force $\bbf_T^r=\bP\bbf_T^r$. Motivated by the structure of the three-dimensional Beris--Edwards model, we propose the kinematic equations
 \begin{equation}\label{evol_laws}
 \begin{aligned}
     \rho\dot{\bu}&=\bP\divG(\ba^d_T+\ba^r)+\bbf_T^r\,, 
     \\
     \overset{\circ}{\bQ}&=\bF^d\, ;
 \end{aligned}
 \end{equation}
the terms $\bbf_T^r$ and $\ba_T^d$ in the surface momentum equation are assumed to be tangent to the surface $\Gamma$ because the model \VY{involves a two-dimensional viscous flow along $\Gamma$ which is expected to be recovered upon setting $\bQ=0$ in \eqref{evol_laws}; also, see Remark~\ref{Q=0}}.
 \end{assumption}
\LB{
\begin{remark}[alternative definition of surface divergence and evolution laws]
As mentioned in Remark \ref{rmk:alt-divG}, our definition of $\divG$ follows that in differential geometry. However, suppose we posit the evolution laws to be 
 \begin{equation*}
 \begin{aligned}
     \rho\dot{\bu}&=\bP\widehat{\divG}(\ba^d_T+\ba^r)+\widehat{\bbf}_T^r\,, 
     \\
     \overset{\circ}{\bQ}&=\bF^d\, ,
 \end{aligned}
 \end{equation*}
where $\widehat{\divG}$ is defined as the $L^2(\Gamma)$ adjoint of $\nabla_M$ as motivated by the integration-by-parts formula in Lemma \ref{matrix_parts} (covariant integration by parts).
The procedure outlined in Steps 4 and 5 below would yield exactly the same surface Beris-Edwards system as \eqref{k1} and \eqref{k2}. The resulting $\widehat{\bf f}_T^r$ would be slightly simpler than ${\bbf}_T^r$ because of using the $L^2(\Gamma)$ adjoint $\widehat{\divG}$, but the right-hand sides of the equations for $\dot{\bu}$ would be identical.
\end{remark}
}
Our next task is to combine Assumption \ref{assum_Onsager} with \eqref{total-energy-rate}.
We first invoke Lemma \ref{matrix_parts}, i.e.
\begin{equation*}
     \bigg(\bu,\rho\dot{\bu}\bigg)_{\!\!\Gamma}
     =-\bigg(\gradM\bu,\ba_T^d+\ba^r\bigg)_{\!\!\Gamma}+\bigg((\tr\bB)(\ba_T^d+\ba^r)\bn,\bu\bigg)_{\!\!\Gamma}
     +\bigg(\bu,\bbf_T^r\bigg)_{\!\!\Gamma} \, ,
\end{equation*}
and take into account the symmetry and tangentiality of $\ba_T^d=\bP\ba_T^d\bP$ to write $\big(\ba_T^d,\nablaM\bu \big)_{\Gamma} = \big(\ba_T^d,\bD_\Gamma(\bu) \big)_{\Gamma}$ because of \eqref{strainM} and \eqref{strain_relation}, as well as the skew-symmetry of $\ba^r$ and 
\eqref{strainM} to obtain
\begin{align}
     \bigg(\bu,\rho\dot{\bu}\bigg)_{\!\!\Gamma}=\bigg(\bu, \bbf_T^r+(\tr\bB)\ba^r\bn\bigg)_{\!\!\Gamma} -\bigg(\bD_\Gamma(\bu),\ba^d_T\bigg)_{\!\!\Gamma}-\bigg(\bW_M(\bu),\ba^r\bigg)_{\!\!\Gamma}\, .
     \end{align}
Similarly, \eqref{evol_laws} yields
     \[
    \left(\bH,\overset{\circ}{\bQ}\right)_{\!\!\Gamma}=\bigg(\bH,\bF^d\bigg)_{\!\!\Gamma}\,,
    \]
    and, in view of $\bS=\bS[\bu, \bQ] = \bW(\bu) \bQ - \bQ \bW(\bu)$, we see that
    \begin{equation*}
    \bigg(\bH,(\nablaM\bQ)\bu - \bS[\bu, \bQ]\bigg)_{\!\!\Gamma}=\bigg(\bu,\bH:\nablaM\bQ\bigg)_{\!\!\Gamma}-\bigg(\bH,\bW(\bu)\bQ-\bQ\bW(\bu)\bigg)_{\!\!\Gamma} \, .
  \end{equation*}
  We define the skew-symmetric \textit{Ericksen} stress $\bSigma$  and the \textit{Leslie} force $\bLambda$ to be
   \begin{align}\label{stresses}
       \bSigma=\bQ\bH-\bH\bQ, \qquad \bLambda=-\bH:\nablaM\bQ;
   \end{align}
the latter is tangent to $\Gamma$ due to \eqref{tangent_contr}, whereas the former is unrelated to $\Gamma$. Consequently,
\begin{equation*}
    \bigg(\bH,(\nablaM\bQ)\bu - \bS[\bu, \bQ]\bigg)_{\!\!\Gamma} = - \bigg(\bu,\bLambda\bigg)_{\!\!\Gamma}+\bigg(\bSigma,\bW(\bu)\bigg)_{\!\!\Gamma}\, .
  \end{equation*}

Using \eqref{WM-W} for $\bu=\bu_T$ tangential to $\Gamma$, namely $\bW(\bu)=\bW_M(\bu)+\frac12\bW_*(\bu)$, we deduce
\[
\bSigma:\bW(\bu) = \bSigma:\bW_M+\bB\bSigma\bn\cdot\bu
\]
because $\bSigma$ and $\bW_* (\bu)=\bB\bu\otimes\bn - \bn \otimes\bB\bu$ being antisymmetric yield
\begin{equation*}
 \frac12 \bSigma : \bW_* =
 \bSigma: (\bB\bu\otimes\bn) =
 \tr(\bn\bu^T\bB\bSigma) =
 \tr(\bu^T\bB\bSigma\bn) =
 \bB\bSigma\bn\cdot\bu.
 \end{equation*}
Hence the rate of change of the  total energy $E_{\rm{total}}[\bQ,\nablaM\bQ,\bu]$ is given by
 \begin{equation}\label{rate}
 \begin{aligned}
  \frac{d}{dt} E_{\rm{total}}[\bQ,\nablaM\bQ,\bu] & =-\bigg(\bH,\bF^d\bigg)_{\!\!\Gamma}- \bigg(\bD_\Gamma(\bu),\ba^d_T\bigg)_{\!\!\Gamma}
   \\
  &+\bigg(\bu, \bbf_T^r-\bLambda+(\tr\bB)\ba^r\bn +\bB\bSigma\bn\bigg)_{\!\!\Gamma} -\bigg(\bW_M(\bu),\ba^r-\bSigma\bigg)_{\!\!\Gamma} \, ,
 \end{aligned} 
 \end{equation}
 where the dissipative and reversible terms are collected in separate lines.

 \begin{remark}\label{key}
 Only the dissipative terms in the first line of \eqref{rate} should contribute to the energy rate $\frac{d}{dt}E_{total}$ as identified by Assumption~\ref{assum_Onsager}. Two reversible terms from the second line should cancel with each other for any possible dynamics of the system. From the perspective of theory of constitutive modeling, this gives rise to a plethora of possible models where the functional dependence of $\ba^r$ and $\bbf_T^r$ on $\bu$ and $\bQ$ varies even in the flat case. We will require in the Step 4 that each term in the second line vanishes separately. This modeling choice is consistent with the classical Beris--Edwards system in flat domains. One could try to attribute this choice to the principle of frame indifference, but it is beyond the scope of this paper.
 \end{remark}

 \medskip
 {\it Step 4: Reversible quantities.}
To find the reversible quantities $\ba^r$ and $\bbf_T^r$ we recall that these terms should not contribute to the time derivative of the total energy \eqref{rate}. Some of the reversible quantities in \eqref{rate} are paired with the velocity $\bu$ which represents a uniform motion of an infinitesimal material volume while others are paired with $\bW_M(\bu)$ which represents a rotation of an infinitesimal volume. 
As discussed in Remark~\ref{key}, none of these conjugated pairs should produce mechanical work, whence the last two terms in \eqref{rate} vanish for any $\bu$:
\begin{equation}\label{reversible-quantities}
\ba^r=\bSigma
\qquad
\bbf_T^r=\bLambda-(\bB+(\tr\bB)\bP)\bSigma\bn\, .
 \end{equation}  
Consequently, the reversible quantities in \eqref{evol_laws} are fully determined.

\vskip 0.2cm
{\it Step 5: Dissipative quantities.}
To find the dissipative quantities $\ba_T^d$ and $\bF^d$ we make an additional assumption regarding the  nonequilibrium thermodynamics of the model in the form of \textit{the least action principle} (see \cite{Wang2021} for details)
First, we define the dissipation functional 
  \begin{align*}
     \Phi[\bF^d,\ba_T^d] :=\int_{\Gamma}\left(\frac{|\bF^d|^2}{2M}+\frac{|\ba^d_T|^2}{4\mu}\right)
 \end{align*}
where the mobility $M$ and the viscosity $\mu$ are material constants. According to the least action principle, dissipative quantities should minimize the expression $\frac{d}{dt} E_{\rm{total}}+\Phi$ at every time during the evolution to be thermodynamically consistent. Therefore, its first variation must vanish
\begin{align*}
     \delta_{(\ba_T^d,\bF^d)}\left(\frac{d}{dt} E_{\rm{total}} [\bQ,\nablaM\bQ,\bu]+\Phi[\bF^d,\ba_T^d]\right)=0 \, .
\end{align*}
In view of the first line of \eqref{rate}, we discover that the tensors $\ba_T^d$ and $\bF^d$
satisfy

\begin{equation}\label{dissipative-quantities}
    -\bD_\Gamma(\bu)+ \frac{1}{2\mu}\ba^d_T=0
    \,,\quad -\bH+ \frac{1}{M}\bF^d=0\,.
 \end{equation}
Consequently, the dissipative quantities in \eqref{evol_laws} are fully determined.

 We have just finished the five steps of the generalized Onsager principle and are now
 ready to write the ensuing system of equations on $\Gamma$. Inserting \eqref{reversible-quantities} and \eqref{dissipative-quantities} into \eqref{evol_laws} yields
   \begin{align}
    \rho\dot{\bu}&=2\mu\bP\divG\bD_\Gamma(\bu)+\bP\divG\bSigma+\bLambda-\big(\bB+(\tr\bB)\bP\big)\bSigma\bn\,, \label{k1}
     \\
     \overset{\circ}{\bQ}&=M\bH\,,\label{k2}
    \end{align}
 which by construction enjoys the following energy structure.
 \begin{proposition}[energy law] The system of equations \eqref{k1}-\eqref{k2} on the surface $\Gamma$ satisfies
 \begin{align} \label{energy_law} \frac{d}{dt}E_{\rm{total}}[\bQ,\nablaM\bQ,\bu] = -2\mu\|\bD_\Gamma(\bu)\|_\Gamma^2-M\|\bH\|_\Gamma^2\,. \end{align}
 \end{proposition}
\begin{proof}
Simply replace \eqref{reversible-quantities} and \eqref{dissipative-quantities} into \eqref{rate}.
\end{proof}

A further simplication of \eqref{k1} is in order. We express the tangential force $\bP\divG\bSigma$ due to the Ericksen stress $\bSigma$ defined in \eqref{stresses} in terms of the \textit{tangent Ericksen stress} $\bSigma_\Gamma=\bP\bSigma\bP$ and a remainder.
We resort to Lemma \ref{matrix_div} to relate $\divG\bSigma$ to $\divG\bSigma_\Gamma$ as follows:
  \begin{align}\label{q1}
     \bP\divG\bSigma&=\bP\divG\bSigma_\Gamma+\tr(\bB)\bP\bSigma\bn+\bB\bSigma^T\bn\, .
 \end{align}
We observe that the term $\tr(\bB)\bP\bSigma\bn$ in \eqref{q1} cancels with the last term in
\eqref{k1}, while the skew-symmetry $\bSigma=-\bSigma^T$ of the Ericksen stress implies
$\bB\bSigma\bn-\bB\bSigma^T\bn=2\bB\bSigma\bn$. We thus end up with the following reduced form of the momentum equation
\begin{equation}\label{momentum-new}
\rho\dot{\bu}=2\mu\bP\divG\bD_\Gamma(\bu)+\bbf_E +\bLambda-\bbf_* \, ,
\end{equation}
with the {\it Ericksen force} $\bbf_E$ and {\it star force} $\bbf_*$ defined by
\begin{equation}\label{ericksen-star}
\bbf_E : = \bP\divG\bSigma_\Gamma,
\qquad
\bbf_* := 2\bB\bSigma\bn.
\end{equation}
Both forces are tangent to $\Gamma$, $\bbf_E$ due to the projection $\bP$ and $\bbf_*$ because the range of the shape operator $\bB$ is contained in the tangent plane at each point of $\Gamma$. It is worth realizing that thermodynamics consistency requires the presence of the novel force $\bbf_*$ in \eqref{momentum-new}. If the surface $\Gamma$ is flat, e.g. a domain in $\R^2$, then $\bB=0$ and $\bbf_*=0$. Moreover, $\bbf_*$ vanishes again provided $\bSigma\bn=0$ as it would happen if both $\bQ$ and $\bH$ are assumed to be conforming and flat-degenerate. The relaxation of these assumptions is the main contribution of this paper.

We further explore the extraction of the tangent part $\bSigma_\Gamma$ from $\bSigma$ in \eqref{q1}, that leads to the Ericksen force $\bbf_E$ of \eqref{ericksen-star}. In fact, we present a simple characterization of $\bbf_E$, which is of independent interest and quite useful to understand simulations in Section \ref{S:Ericksen-flat}.

\begin{lemma}[characterization of the Ericksen force]\label{L:alt_erick}
There exists a scalar function $\theta$ such that
\begin{align}\label{Ericksen_alternative}
     \bbf_E = \bn \times \nablaG\theta \, .
 \end{align}
\end{lemma}
\begin{proof}
Consider the right-handed basis of principal directions $\bt_1, \bt_2,\bn$ at a point $\bx\in\Gamma$. Any second order tensor $\bA$ may be represented in this basis via dyads as follows:
\begin{align*}
    \bA=\sum_{i,j=1}^2a_{ij}\bt_i\otimes\bt_j+\sum_{i=1}^2(a_{3i}\bn\otimes\bt_i+a_{i3}\bt_i\otimes\bn)+a_{33}\bn\otimes\bn
\end{align*}
for some components $a_{ij}$, $1\leq i,j\leq 3$.
Since the surface Ericksen stress $\bSigma_\Gamma=\bP\bSigma\bP$ is tangent, its representation does not include any dyads involving $\bn$. Moreover, skew-symmetry $\bSigma_\Gamma=-\bSigma_\Gamma^T$ implies
$a_{11}=a_{22}=0$ and $a_{21}=-a_{12} =\theta$, where $\theta$ is the only non-trivial component of $\bSigma_\Gamma$ and is a function of $\bx\in\Gamma$. Consequently,
\begin{equation*}
    \bSigma_\Gamma=\theta \left(\bt_2\otimes\bt_1-\bt_1\otimes\bt_2\right)\,,
\end{equation*}
and the second order tensor $\bOmega_\bn:=\bt_2\otimes\bt_1-\bt_1\otimes\bt_2$ maps $\bt_1$ to $\bt_2$ and $\bt_2$ to $-\bt_1$. It is thus a rotation by $\frac{\pi}{2}$ around the axis $\bn$ or, simply, the  cross product operator $\bOmega_\bn\ba = \bn \times \ba$ according to the corkscrew rule. It turns out that $\bOmega_\bn$ admits the following \textit{cross product} matrix representation in terms of the canonical
\VY{Cartesian} basis $\be_x, \be_y, \be_z$
\begin{align}\label{Omega-n}
\bSigma_\Gamma=\theta\bOmega_\bn=
\theta\begin{pmatrix}
    0 & -\bn_z & \bn_y\\
    \bn_z & 0 & -\bn_x\\
    -\bn_y & \bn_x & 0
  \end{pmatrix}\,.
\end{align}
Note that the function $\theta$ and the normal $\bn$ fully describe the surface Ericksen stress $\bSigma_\Gamma$. Moreover, using \eqref{divergence}, \eqref{e-n} together with Proposition \ref{appendix} (product rules), we calculate
\begin{align*}
 \divG \bSigma_\Gamma = \divM \bSigma_\Gamma=\div \bSigma_\Gamma^e=\div(\theta^e\bOmega_\bn^e)
              =\theta^e\div\bOmega^e_\bn+\bOmega^e_\bn\nabla\theta^e \, .
\end{align*}
We finally observe that $\div\bOmega^e_\bn = -\textrm{curl} \, \bn^e = 0$, because $\bn^e=\nabla d$, to obtain that $\divG \bSigma_\Gamma = \bn\times\nablaG\theta$ and that $\bbf_E=\bP\divG \bSigma_\Gamma$ is given by \eqref{Ericksen_alternative} as asserted.
\end{proof}

We point out that the orientation of $\bn$ is not unique. If we change $\bn$ to $-\bn$, then we also have to exchange $\bt_1$ with $\bt_2$ to have a right-handed basis and this flips the sign of $\theta$; hence the representation of \eqref{Ericksen_alternative} is well defined. Moreover, since $\divG \bSigma_\Gamma=\bn\times\nablaG\theta$ is already tangent to $\Gamma$, we deduce that the projection $\bP$ in the definition $\bbf_E=\bP\divG\bSigma_\Gamma$ is superfluous.

 \subsection{Surface Beris--Edwards model} \label{surf_BE}
 We are now in a position to present the novel model of fluidic liquid crystal films. Let a closed surface $\Gamma$ represent the liquid crystal film. The liquid crystal may be generally oriented in $\R^3$, but the material flows tangentially to $\Gamma$ so that $\Gamma$ does not change over time.
 The incompressible flow is described by the tangential velocity $\bu$, which is assumed to be divergence-free $\divG\bu=0$. Therefore, the density $\rho$ is constant and the scalar pressure field $p$ enforces $\divG\bu=0$ on $\Gamma$.
  
 The new model combines the equations \ref{k2} and \eqref{momentum-new} with the expressions \eqref{def-H} and \eqref{time-derivatives} and the constitutive relations \eqref{stresses} and \eqref{ericksen-star}. Given initial conditions $\bu_0$ and $\bQ_0$, the model reads: find symmetric and traceless matrix fields $\bH,\bQ$ as well as tangent velocity $\bu=\bP\bu$ and scalar pressure $p$ on $\Gamma$ such that for all times the following system of PDEs is satisfied on $\Gamma$
 
\begin{equation}\label{beris-edwards}
\begin{aligned}
    {\bf H} +\mathcal{P}F'[\bQ] &=   L\,\mathcal{P}\divM\nablaM{\bQ}\,, \\
                  \pa_t{\bf Q} +(\nablaM{\bf Q}){\bf u}    &= M{\bf H}+  ({\bf S}_\Gamma+{\bf S}_*)\,,  \\
         \rho\big(\pa_t\bu +(\nabla_\Gamma \bu)\bu+\nabla_{\Gamma}p \big) &= 2\mu \bP\divG\bD_\Gamma(\bu)-(\bbf_{\Gamma}+\bbf_*)  \,, \\
         \divG \bu &=0 \,,
\end{aligned}
\end{equation}
where  
\begin{align*}
    \bS_\Gamma&=\bS_\Gamma[\bQ,\bu]=\bW_\Gamma(\bu)\bQ-\bQ\bW_\Gamma(\bu)\,,&\quad
     \bS_*&=\bS_*[\bQ,\bu]=\bW_*(\bu)\bQ-\bQ\bW_*(\bu)\,,\\
   \bSigma&=\bSigma[\bQ,\bH]=\bQ\bH-\bH\bQ\,,&\quad   \bSigma_\Gamma&=\bSigma_\Gamma[\bQ,\bH]=\bP\bSigma\bP\,,
     \\
     \bbf_{\Gamma} &= \bbf_{\Gamma} [\bQ,\bH]=-\bP\divG{\bf \Sigma}_\Gamma+  \bH:\nabla{}_M\bQ\,,&\quad
     \bbf_*&=\bbf_*[\bQ,\bH]=2\bB\bSigma\bn\, .
    \end{align*}
   Here $\bS_\Gamma$ is the corotation \eqref{S_def} by the covariant spin tensor $\bW_\Gamma$ in \eqref{strain}, $\bS_*$ is the corotation \eqref{S_def} by the star spin tensor $\bW_*$ in \eqref{star_spin}, $\bSigma$ is the Ericksen stress \eqref{stresses} with the tangent part $\bSigma_\Gamma=\bP\bSigma\bP$, $\bbf_{\Gamma}$ is the tangent \textit{surface Beris--Edwards} force  which consists of the Ericksen force $\bbf_E=\bP\divG{\bf \Sigma}_\Gamma$ in \eqref{ericksen-star} and the Leslie force  $\bLambda=-\bH:\nablaM\bQ$ in \eqref{stresses}, and $\bbf_*$ is the star force \VY{from}
   \eqref{ericksen-star}.
    The first variation of the double-well potential $F[\bQ]$ in \eqref{dw} is given by
    \begin{align}\label{der_dw}
     \mathcal{P}F'[\bQ]=a{\bf Q} -b{\bf Q}^2 +\frac{b}{3}\text{tr}({\bf Q}^2 ){\bf I} +c\,\tr({\bf Q}^2){\bf Q}\,,
      \end{align}
      where $\mathcal{P}$ is the projection onto the subspace of traceless and symmetric matrices. The operator $\mathcal{P}$ acts likewise on the variation $\divM\nablaM\bQ$ of the elastic energy.
      The surface Beris--Edwards model \eqref{beris-edwards} obeys the energy law \eqref{energy_law} by construction.
   \begin{remark}
   We would like to stress that the star corotation tensor $\bS_*$ and the star force $\bbf_*$ distinguish our  model from the model in \cite{nestler2021active} where $\bQ$ is assumed to be conforming to the surface with a prescribed eigenvalue in the normal direction. These terms guarantee the thermodynamical consistency of our model for a non-flat surface. We demonstrate the behavior of conforming and non-conforming Q-tensors in our numerical experiments of Section \ref{numerics}. 
    \end{remark} 
   
   \begin{remark}\label{Q=0}
   If we disregard the Q-tensor equations and the coupling force $\bbf_\Gamma+\bbf_*$ we will be left with the well-known surface Navier--Stokes system \cite{jankuhn2018incompressible} which models an incompressible surface flow driven by inertia.
   \end{remark}

\section{Representation of Q-tensors on surfaces}\label{S:representation}
 In this section we define the notions of uniaxiality and flat-degeneracy, along with the biaxiality parameter $\beta[\bQ]$ which relates them. We also introduce the non-conformity parameter $r_\Gamma[\bQ]$. These parameters will be instrumental in describing and visualizing the numerical experiments in Section \ref{numerics}. 
 

 \subsection{Biaxiality parameter}\label{representation}

We start with a simple definition: unit vector fields $\bq \in \R^3$ will be called {\it director} fields.
\begin{definition}[flat-degeneracy]\label{flat-degenerate}
A Q-tensor $\bQ\in\R^{3\times3}$ is {\it flat-degenerate} if one of its eigenvalues is zero, say $\lambda_2=0$ whence $\lambda_1=-\lambda_3$. Therefore, if $r:=2\lambda_1$, then $\bQ$ reads
\begin{align}\label{e1}
    \bQ=\lambda_1 \left(\bq_1\otimes\bq_1 - \bq_3\otimes\bq_3\right)=r \left(\bq_1\otimes\bq_1 - \frac12\bP_{\bq_2}\right) \, ,
\end{align}
where $\bP_\bq:=\bI-\bq\otimes\bq$ is the projector onto the orthogonal plane to the director $\bq$. 
\end{definition}

This definition is consistent with \eqref{fd_def}.
A flat-degenerate state of the form \eqref{e1} is essentially a two-dimensional Q-tensor state in the plane orthogonal to $\bq_2$. Another important class of three-dimensional Q-tensor states is given by the following definition \VY{\cite{mottram2014introduction}, \cite{borthagaray2021q}, \cite{sonnet2012dissipative}}.
\begin{definition}[unixiality]\label{uni-bi} A Q-tensor $\bQ\in\R^{3\times3}$ is \textit{uniaxial} if it may be represented as
\begin{align}\label{uniax}\bQ=s\left(\bq\otimes\bq-\frac13\bI\right)
\end{align}
for some director $\bq$ and order parameter $s$. Otherwise the Q-tensor is \textit{biaxial}.
\end{definition}


The following biaxiality parameter \cite{majumdar2010equilibrium} relates the notions of flat-degeneracy and uniaxiality of Q-tensors, and allows for a classification of Q-tensor fields useful in simulations.  
\begin{definition}[biaxiality parameter]\label{biaxiality_def}
For a non-zero Q-tensor $\bQ\in\R^{3\times3}$ the \textit{biaxility parameter} is the real number
\begin{align}\label{biax}\beta[\bQ]=1-6\frac{(\tr{}\bQ^3)^2}{(\tr{}\bQ^2)^3} \, .
\end{align}
\end{definition}

\VY{It is well known that the vanishing of $\beta[\bQ]$ indicates that $\bQ$ is uniaxial \cite{majumdar2010equilibrium}, but the opposite limit, when it is equal to one, indicates the flat-degeneracy of liquid crystal state as shown in the following Lemma.}

\begin{lemma}[properties of biaxiality parameter]\label{L:biaxiality}
The biaxiality parameter satisfies
\begin{equation}\label{biaxiality-bounds}
0\le \beta[\bQ]\le 1 \qquad\forall \, \bQ \in \R^{3\times3}\,,\quad\bQ\neq 0.
\end{equation}
The minimal value, $\beta[\bQ]=0$, corresponds to $\bQ$ being uniaxial whereas the maximal value, $\beta[\bQ]=1$, corresponds to $\bQ$ being flat-degenerate.
\end{lemma}

\begin{proof}
We exploit the spectral decomposition \eqref{spectral} to write $\bQ^j=\sum_{k=1}^3\lambda_k^j (\bq_k\otimes \bq_k)$ for $j\in\mathbb{N}$, whence $\tr \bQ^j = \sum_{k=1}^3 \lambda_k^j$.
In view of the definition \eqref{biax}, we have to prove
\[
0 \le 6 \frac{\big(\tr \bQ^3 \big)^2}{\big( \tr \bQ^2 \big)^3} \le 1.
\]
The leftmost inequality is trivial. The rightmost one entails the following tedious computation.
Since $\tr \bQ=0$ we let $\lambda_2=-\lambda_1-\lambda_3$ and rewrite the desired traces in terms of $\alpha=\lambda_1^{-1}\lambda_3$
\begin{align*}
\tr \bQ^2 &= \lambda_1^2 \big( 1 + (1+\alpha)^2 + \alpha^2 \big) = 2 \lambda_1^2 (1+\alpha + \alpha^2),
\\
\tr \bQ^3 &= \lambda_1^3 \big( 1 - (1+\alpha)^3 + \alpha^3 \big) = - 3 \lambda_1^3 \alpha (1+\alpha),
\end{align*}
because $\bQ\ne 0$ implies either $\lambda_1\ne0$ or $\lambda_3\ne0$.
Consequently, we obtain the asserted inequality
\begin{equation*}
  \big( \tr \bQ^2 \big)^3 - 6 \big( \tr \bQ^3 \big)^2 =
  2 \lambda_1^6 \big(2+3\alpha-3\alpha^2-2\alpha^3 \big)^2 = 
  2 \lambda_1^6 (1-\alpha)^2 ( 2\alpha+1 )^2 (\alpha+2)^2 \ge 0.
\end{equation*}

Moreover, this explicit expression reveals that $\beta[\bQ]=0$ is equivalent to either
$\alpha=1, \alpha = -\frac12$ or $\alpha = -2$. This in turn corresponds to $\lambda_1=\lambda_3$,
$\lambda_1=-2\lambda_3$ (i.e. $\lambda_2 = \lambda_3$) or $\lambda_3=-2\lambda_1$ (i.e. $\lambda_2 = \lambda_1$). According to the characterization of a uniaxial Q-tensor after Definition \ref{uni-bi}, we deduce that $\beta[\bQ]=0$ is equivalent to $\bQ$ 
being uniaxial.

In contrast, $\beta[\bQ]=1$ is equivalent to $\tr \bQ^3 = 0$ which reduces to either $\alpha=0$ or
$\alpha=-1$. This in turn reads either $\lambda_3 = 0$ (or symmetrically $\lambda_1=0$) or $\lambda_3=-\lambda_1$ (i.e. $\lambda_2 = 0$). According to Definition \ref{flat-degenerate}, we infer 
that $\beta[\bQ]=1$ if and only if $\bQ$ is flat degenerate. This completes the proof.
\end{proof}

Finally, the following result is proved in \cite[Proposition 1]{majumdar2010equilibrium}, but we state it as a lemma for further reference in Section \ref{numerics}.

\begin{lemma}[minimizer of $F$]\label{L:double-well}
Let the parameters $a,b,c$ of the double-well potential $F$ in \eqref{dw} satisfy
$a<0\,,b>0\,,c>0$. Then $F[\bQ]$ is minimized by a uniaxial state \eqref{uniax} with $s$ given by
\begin{align}\label{s_min}
    s_+=\frac{b+\sqrt{b^2-24ac}}{4c} \, .
\end{align}
\end{lemma}
%

\subsection{Non-conformity parameter}\label{S:non-conformity}

In this section we turn to the representation of Q-tensors on surfaces and introduce a scalar field which quantifies the notion of non-conformity of Q-tensors. In fact, we extend the representation of conforming Q-tensors of \cite[equation (5)]{nestler2020properties} to arbitrary $\bQ\in\R^{3\times3}$.

Strictly speaking, since we are concerned with general Q-tensors there is no a priori relation between the surface $\Gamma$ and the eigenframe of $\bQ$. Nevertheless, we would like to split the general state of the Q-tensor into liquid crystal states related to the normal and tangent subspaces to $\Gamma $.

Given the unit normal $\bn$ to $\Gamma$ and a tangent director field $\bq_\Gamma$, i.e. $\bq_\Gamma\cdot\bn=0$, we consider sums $\mathbf{N}_\Gamma+\bT_\Gamma$ of arbitrary traceless matrices of the form
\begin{align}\label{homeo_form}
    \mathbf{N}_\Gamma=n_\Gamma \left(\bn\otimes\bn-\frac{1}{2}\bP\right)=\frac32 n_\Gamma \left(\bn\otimes\bn-\frac{1}{3}\bI\right)\,,\quad{}\bT_\Gamma=s_\Gamma\left(\bq_\Gamma\otimes\bq_\Gamma-\frac12\bP\right) \, ,
    \end{align}
where $\mathbf{N}_\Gamma$ is an uniaxial homeotropic Q-tensor (normal to $\Gamma$), and $\bT_\Gamma$ is a flat-degenerate Q-tensor tangent to $\Gamma$.
Note that $\beta[\bN_\Gamma]=0$ and $\beta[\bT_\Gamma]=1$ and that both $\bN_\Gamma$ and $\bT_\Gamma$ are always conforming. 
Clearly, an arbitrary Q-tensor $\bQ\in\R^3$ cannot be represented by such sums if its eigenframe does not include the normal vector. Therefore, for $\bQ\in\R^3$ we define its traceless conforming (normal and tangential) components $\bN_\Gamma[\bQ]$ and $\bT_\Gamma[\bQ]$ by minimizing the residual with respect to $n_\Gamma$ and $\bT_\Gamma$:
\begin{align}\label{minim}
&\min\limits_{n_\Gamma, \bT_\Gamma}|\bR_\Gamma{}|^2\,,\qquad\bR_\Gamma{}= \bQ-\mathbf{N}_\Gamma -\bT_\Gamma \,.
\end{align}

\begin{lemma}[homeotropic decomposition of $\bQ$ on $\Gamma$]\label{L:homeotropic-decomposition}
An arbitrary Q-tensor $\bQ\in\R^{3\times3}$ admits the orthogonal decomposition $\bQ=\bN_\Gamma[\bQ]+\bT_\Gamma[\bQ]+\bR_\Gamma[\bQ]$ into three traceless symmetric tensors, where $\bN_\Gamma[\bQ]$ is a uniaxial $Q$-tensor given by \eqref{homeo_form} with 
$n_\Gamma=\bn^T\bQ\bn$, $\bT_\Gamma[\bQ]$ is a flat-degenerate tangent Q-tensor, and $\bR_\Gamma[\bQ]$
has minimal Frobenius norm. Moreover, they satisfy 
$|\bQ|^2 = |\bN_\Gamma[\bQ]|^2 + |\bT_\Gamma[\bQ]|^2 + |\bR_\Gamma[\bQ]|^2$ and are given by the expressions
\begin{equation}\label{Q_decomp_traceless}
\begin{aligned} 
\bN_\Gamma[\bQ]&=n_\Gamma\left(\bn\otimes\bn-\frac12\bP\right)\,,\quad  \bT_\Gamma[\bQ]=\bP\bQ\bP+\frac{n_\Gamma}{2}\bP\,,
\\
\bR_\Gamma[\bQ]&=(\bQ-n_\Gamma\bI)\bn\otimes{}\bn+\bn\otimes{}(\bQ-n_\Gamma\bI)\bn\,.
\end{aligned}
\end{equation}
\end{lemma}
\begin{proof}
We expand the residual
\begin{align*}
|\bR_\Gamma{}|^2 &=(\bQ- \bT_\Gamma- n_\Gamma \bn\bn^T-\frac{n_\Gamma}{2}\bP ):(\bQ- \bT_\Gamma - n_\Gamma \bn\bn^T -\frac{n_\Gamma}{2}\bP)
\\
&=\bQ:\bQ-2\bQ:\bT_\Gamma+\bT_\Gamma:\bT_\Gamma- 2n_\Gamma(\bQ-  \bT_\Gamma):\bn\bn^T-n_\Gamma\bP:(\bQ-  \bT_\Gamma) + \frac32 n_\Gamma^2 \,,
\end{align*}
and compute its first variations. Since $n_\Gamma$ is scalar, we readily have
\begin{equation*}
0=\frac{\partial}{\partial{}n_\Gamma}|\bR_\Gamma{}|^2=- 2\bQ:\bn\bn^T+\tr{}(\bP\bQ\bP-\bT_\Gamma)+3n_\Gamma \, .
\end{equation*}
On the other hand, since $\bT_\Gamma$ is traceless and tangent to $\Gamma$, a general variation of $\bT_\Gamma$ reads $\bP \bC \bP - \frac12 \tr \big( \bP \bC \bP \big) \bP$ for an arbitrary symmetric matrix $\bC \in \R^{3\times3}$. Consequently, a tedious computation of $\frac{\partial}{\partial{}\bT_\Gamma}|\bR_\Gamma{}|^2 : \bC$ using that $\tr \bP=2$ yields
\begin{equation*}
0=\frac{\partial}{\partial{}\bT_\Gamma}|\bR_\Gamma{}|^2=- 2\bP\bQ\bP - \tr{(\bP\bQ\bP)}\bP +2\bT_\Gamma \, . 
\end{equation*}
Exploiting that $\tr (\bT_\Gamma) = 0$, these equations give the optimal values
\begin{equation*}
n_\Gamma=\frac{2}{3}\bn^T\bQ\bn-\frac13\tr{}(\bP\bQ\bP)\,,\qquad \bT_\Gamma=\bP\bQ\bP-\frac12\tr{(\bP\bQ\bP)}\bP\,.
\end{equation*}
Since
\begin{equation*}
  \bP \bQ \bP = \big(\bI - \bn\otimes\bn \big)\bQ \big(\bI - \bn\otimes\bn\big)
  = \bQ - \bn\otimes\bQ\bn - \bQ\bn\otimes\bn + \big(\bn^T\bQ\bn\big) \bn\otimes\bn \, ,
\end{equation*}
and $\bQ$ is traceless, we deduce $\tr\big( \bP \bQ \bP \big)=- \bn^T \bQ \bn$
whence $n_\Gamma = \bn^T \bQ \bn$ and the expressions for $\bT_\Gamma$ and $\bR_\Gamma$ in \eqref{Q_decomp_traceless} follow immediately. Moreover, $\bN_\Gamma[\bQ]$ and $\bT_\Gamma[\bQ]$ are orthogonal because
\begin{equation*}
    \bN_\Gamma[\bQ] : \bT_\Gamma[\bQ] = n_\Gamma \Big(\bn\otimes\bn - \frac12 \bP\Big) : \Big(\bP\bQ\bP + \frac{n_\Gamma}{2} \bP \Big) = \frac{n_\Gamma}{2} \big( \bn^T\bQ\bn - n_\Gamma \big) = 0
    \, ,
\end{equation*}
whence the minimization property \eqref{minim} is equivalent to the orthogonality of $\bR_\Gamma[\bQ]$ and $\bN_\Gamma[\bQ] + \bT_\Gamma[\bQ]$.
This concludes the proof.
\end{proof}

Now we are in a position to introduce a quantitative measure of non-conformity for an arbitrary Q-tensor $\bQ\in\R^3$. Since both $\bN_\Gamma[\bQ]$ and $\bT_\Gamma[\bQ]$ are conforming to $\Gamma$, possible non-conformity of $\bQ$ is dictated by the solution $\bR_\Gamma[\bQ]$ of the minimization problem \eqref{minim}. The relation $|\bR_\Gamma[\bQ]|\le |\bQ|$ motivates the forthcoming definition. 

\begin{definition}[non-conformity parameter] The \textit{non-conformity parameter} $r_\Gamma[\bQ]$ of an arbitrary $\bQ\in\R^{3\times3}$ on $\Gamma$ is the fraction $0\leq r_\Gamma[\bQ] \leq 1$ defined by
\begin{align}\label{nonconf_param}
r_\Gamma[\bQ]:=\frac{|\bR_\Gamma{}[\bQ]|}{|\bQ|} \, .
\end{align}
\end{definition}

\begin{remark}
We see that $r_\Gamma[\bQ]=0$ if and only if $\bQ=\bN_\Gamma[\bQ]+\bT_\Gamma[\bQ]$ or equivalently $\bR_\Gamma[\bQ]=0$. In contrast, $r_\Gamma[\bQ]=1$ if and only if $\bN_\Gamma[\bQ]=\bT_\Gamma[\bQ]=0$ or equivalently $n_\Gamma=0$ and
\begin{equation*}
    \bQ = \bR_\Gamma[\bQ] = \bQ\bn\otimes\bn+\bn\otimes\bQ\bn \, .
\end{equation*}
Therefore, if $r_\Gamma[\bQ]=1$ we infer that $\bQ\bn$ is tangent to $\Gamma$ because $n_\Gamma = \bn^T\bQ\bn=0$ and $\bQ\bq$ is perpendicular to $\Gamma$ for any tangent vector $\bq$ because $0=\bT_\Gamma[\bQ] \bq = \bP \bQ\bq$.
\end{remark}

\begin{remark}
If $\bQ$ is conforming to $\Gamma$, i.e. $r_\Gamma[\bQ]=0$, and its normal component $\bN_\Gamma[\bQ]=0$, then $\bQ=\bT_\Gamma[\bQ]$
is flat degenerate. Therefore, according to Lemma \ref{L:biaxiality}, $\bQ$ is biaxial and the biaxiality parameter
$\beta[\bQ]=1$ is maximal.
\end{remark}

\subsection{Enforcing Conformity: The Hess-Osipov Energy}\label{S:osipov-hess}
We consider enforcing conformity through penalization. Natural penalizations are the following physically justified energies, which can be found in \cite[Eq. (4)]{nestler2020properties} and \cite[Eq. (8)]{golovaty2017dimension}, and are closely tied to the energies found in \cite[Eq. (4)]{golovaty2015dimension} and \cite[Eq. (7)]{osipov1993density}:
\begin{align}\label{pen+norm}
E_{pen}[\bQ] := \gamma\int_\Gamma |\bP\bQ\bn|^2,
\quad
E_{norm}[\bQ] := \alpha\int_\Gamma \big|\bn^T\bQ\bn-\delta \big|^2.
\end{align}
The energy $E_{pen}[\bQ]$ with penalty parameter $\gamma$ weakly enforces that $\bQ\bn$ be normal. In fact,
in the limit $\gamma\rightarrow\infty$, the energy $E_{pen}[\bQ]$ is minimized provided $\bP\bQ\bn=0$ or equivalently if $\bQ\bn$ is normal. This is the strong form of conformity according to \eqref{conf_def}.

On the other hand, the energy $E_{norm}[\bQ]$ with penalty parameter $\alpha>0$ enforces a value $\delta$ of orientational order in the normal direction. For conforming Q-tensors the value $n_\Gamma=\bn^T\bQ\bn$ is the eigenvalue of $\bQ$ in the normal direction $\bn$ according to Lemma \ref{L:homeotropic-decomposition} (homeotropic decomposition of $\bQ$ on $\Gamma$). However, unless $\bn$ is an eigenvector of $\bQ$, the limit $\alpha\rightarrow\infty$ only penalizes the deviation of $\bn^T\bQ\bn$ from $\delta$, which may vary along $\Gamma$.

The role of \eqref{pen+norm} will be computationally explored in Section \ref{S:normal_anchoring} and Section \ref{S:conformity}.

\section{Exploration of the surface Beris--Edwards model}\label{numerics}
 In this section we explore computationally basic properties of the surface Beris--Edwards model presented in Section \ref{surf_BE}. We resort to the biaxiality parameter $\beta[\bQ]$ of Section \ref{representation} and the non-conformity parameter $r_\Gamma[\bQ]$ of Section \ref{S:non-conformity} to interpret and display our results. We start in Section \ref{LdGonsphere} with the kinematics of the \textit{surface Landau--de\,Gennes} equation without transport of momentum. In Section \ref{coupling_forces} we compute profiles of the Leslie force \eqref{stresses} and the Ericksen and star forces  \eqref{ericksen-star} on some simple Q-tensor configurations with a defect; this provides basic intuition on the thermodynamical coupling of the Q-tensor and the momentum transport on surfaces. 
 In section \ref{sec_nonconformity} we demonstrate computationally that the transition between two conforming states may occur through non-conforming intermediate states. Finally, we show in Section \ref{sec_instability} why the relaxation of the conformity assumption \eqref{conf_def} may be critical for the modeling of liquid crystal films. We consider a homeotropic, radially symmetric Q-tensor on a unit sphere and investigate the influence of the weak anchoring on the stability of this Q-tensor configuration.
 
\subsection{Landau--de Gennes dynamics on a sphere} \label{LdGonsphere}

In this section we consider the \textit{surface Landau--de\,Gennes} model from Section \ref{surf_BE} without the momentum equation, and explore the main kinematical and dynamical properties of this simplified model. For all experiments in this section, we set the mobility $M$, the elastic constant $L$, and the parameters of the double-well potential $a,b,c$ in \eqref{der_dw} to be
\begin{equation*}
    M=1, \quad L=1, \quad a = -5, \quad b=1, \quad c=10.
\end{equation*}
Consequently, the equilibrium value \eqref{s_min} of the order parameter is $s_+ \approx 0.60$.

\subsubsection{Passive corotational transport of a non-conforming Q-tensor}
The first numerical simulation demonstrates the action of the corotational derivative $\overset{\circ}{\bQ}$ defined in \eqref{corotation_defintiion2}. To this end,
we consider the passive velocity $\bv(x,y,z)=\pi\be_z\times(x,y,z)$
for $(x,y,z)\in\Gamma$ over the unit sphere $\Gamma$; $\bv$ is tangent to $\Gamma$ and corresponds to a rigid rotation of $\Gamma$ around the axis $\be_z$. We examine the passive transport equation $\overset{\circ}{\bQ}=0$ dictated by $\bv$ over $\Gamma$ where the initial condition $\bQ_0$ of $\bQ$ is uniaxial
\begin{align}\label{passive_model}
    \bQ_0(x,y,z) := s_0\left(\bq_0\otimes\bq_0-\frac13\bI\right)\,,\quad(x,y,z)\in\Gamma\,,
\end{align}
and the order parameter $s_0$ and director $\bq_0$ are given by
\begin{equation}\label{ICs}
s_0(x,y,z)=s_+ \big(1+\exp(-20(y-0.6)) \big)^{-1}\, ,
\quad 
\bq_0(x,y,z)=(1,1,0)/\sqrt{2}\, .
\end{equation}
We stress that $\bq_0$ is neither normal nor tangent, and $s_0$ localizes $\bQ_0$ to a spherical cap
    $0.6 \lesssim y$. We display the passive dynamics in Figure \ref{passive_corot_transport2}.

\begin{figure}[ht!]
\vskip0.2cm
\hskip 2cm
\begin{center}
{
 a)\quad
\begin{overpic}[abs,height=0.25\textwidth,
unit=1mm,
]
{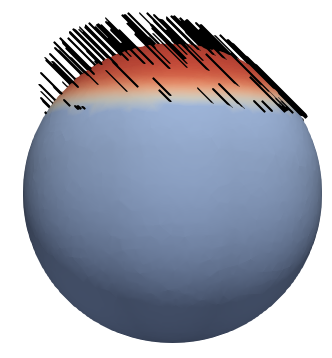}
\put(14, 38){\small{$t = 0$}}
\end{overpic}
\hskip 0.5cm
\begin{overpic}[abs,height=0.25\textwidth,
unit=1mm,
]
{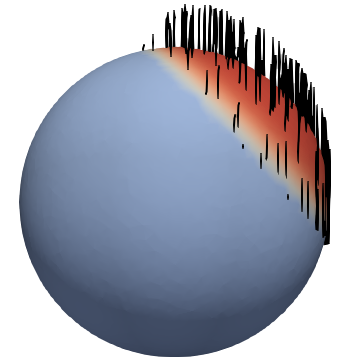}
\put(14, 38){\small{$t = 0.25$}}
\end{overpic}
\hskip 0.5cm
\begin{overpic}[abs,height=0.225\textwidth,
unit=1mm,
]
{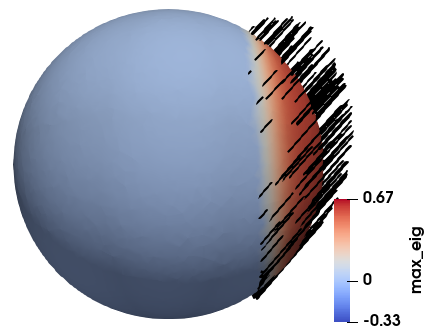}
\put(14, 38){\small{$t = 0.5$}}
\end{overpic}
}
\end{center}
 
\begin{center}
{
 b)\quad
\begin{overpic}[abs,height=0.25\textwidth,
unit=1mm,
]
{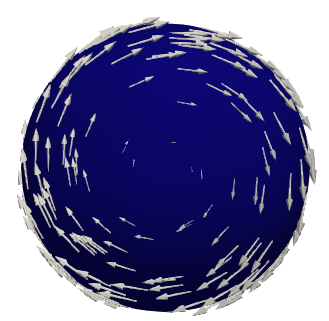}
\end{overpic}
\hskip 0cm
\begin{overpic}[abs,height=0.25\textwidth,
unit=1mm,
]
{corot/hair_biaxiality_0.png}
\end{overpic}
\hskip 0.1cm
\begin{overpic}[abs,height=0.235\textwidth,
unit=1mm,
]
{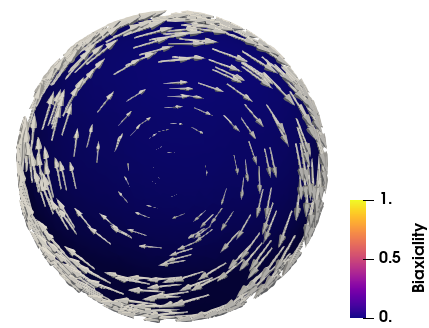}
\end{overpic}
}
\end{center}
\caption{\label{passive_corot_transport2}
\small{
Passive surface corotational transport of an initially uniaxial Q-tensor along a rotation of the unit sphere $\Gamma$ given by the prescribed velocity field $\bv=\pi\be_z\times(x,y,z)$ for $(x,y,z)\in\Gamma$;
all pictures display the $(x,y)$-plane so $\be_z$ is perpendicular to it.
a) largest eigenvalue of $\bQ$ and corresponding oblique eigenvector $\bq$.
b) biaxiality parameter $\beta[\bQ]$ of \eqref{biax} and velocity field $\bv$. Since $\beta[\bQ]$ stays close to zero, $\bQ$ remains uniaxial with respect to $\bq$. The uniaxial Q-tensor state $(s,\bq)$ is uniform on the spherical cap and rotates rigidly. Therefore, the entire Q-tensor eigenframe moves along the sphere as if the ambient space experiences the rotation.}  }
\label{pas_tr}
\end{figure}

According to Theorem \ref{transport_decomp_Q}, and the property that a Q-tensor with two equal eigenvalues is uniaxial, the solution to $\overset{\circ}{\bQ}$ with initial condition \eqref{passive_model} is the uniaxial Q-tensor
\begin{align}\label{w1}
\bQ(x,y,z,t)=s\left(\bq\otimes\bq-\frac13\bI\right)\,,\quad(x,y,z)\in\Gamma\,, 
\end{align}
where $s=s(x,y,z,t)$ and $\bq=\bq(x,y,z,t)$ satisfy the initial value problems on $\Gamma$
\begin{align}\label{ivp-passive}
    \dot{s}=0\,,\quad s(0)=s_0\,,\qquad\qquad\overset{\circ}{\bq}=0\,,\quad\bq(0)=\bq_0\,.
\end{align}
In view of Lemma \ref{charact_corot}, the director field $\bq$ admits the decomposition $\bq=\bq_T+q_N\bn$ in terms of normal component $q_N$ and tangential component $\bq_\Gamma$, which satisfy the following initial value problems:
\begin{align*}
    \dot{q}_N=0\,,\quad q_N(0)=(\bq_0)_N\,,\qquad \pa_t\bq_T+(\nablaG\bq_T)\bv-\bW_\Gamma(\bv)\bq_T=0\,,\quad\bq_T(0)=(\bq_0)_T\,,
\end{align*}
Since $\bv$ is a rotation of the sphere, the solution $(s(t),\bQ(t))$ of the initial value problem \eqref{ivp-passive} is just the rigidly rotated initial condition $(s_0,\bq_0)$. This solution of \eqref{ivp-passive} for $t\in [0,0.5]$ is shown in Figure~\ref{pas_tr}: at the final time $t=0.5$ the solution $(s(t),\bQ(t))$ has rotated $\pi/2$ around $\be_z$ and the biaxiality parameter $\beta[\bQ]\approx0$. This corroborates that $\bQ(t)$ remains uniaxial for all $t$.

It is worth realizing that if one did not use the covariant spin tensor $\bW_\Gamma$ defined in \eqref{strain} in the transport of the tangent component $\bq_T$, then the parallel transport \eqref{pt} would not result in the rotated solution $\bq(t)$ (and $\bQ(t)$). This numerical example highlights the importance of corotational derivatives \eqref{corotation_defintiion} and \eqref{corotation_defintiion2} for the kinematics of liquid crystal films.

\subsubsection{Diffusion of a uniaxial Q-tensor.}
In this example, we explore the so-called ``dry" case of the surface Beris--Edwards model. The Q-tensor changes are driven solely by the interaction of elastic energy and double-well potential $F$ in the Landau - de\,Gennes energy \eqref{eldg} in the absence of momentum transport. We thus set $\bu=0$ (no fluid) in the system \eqref{beris-edwards} from Section \ref{surf_BE}, thereby resulting in the gradient flow dynamics for the surface Landau--de\,Gennes energy $E_{LdG}[\bQ,\nabla_M\bQ]$
\begin{equation}\label{LdG_dynamics}
\begin{aligned}
    {\bf H} +\mathcal{P}F'[\bQ] &=   L\,\mathcal{P}\divM\nablaM{\bQ} \,, 
    \\
                  \pa_t{\bf Q}     &= M{\bf H} \,,
\end{aligned}
\end{equation}
which is supplemented with the initial condition $\bQ(0)=\bQ_0$ from \eqref{passive_model}.
The initial value problem \eqref{LdG_dynamics} is solved numerically on a unit sphere $\Gamma$ and the results are displayed on Figure \ref{diff_tr}. The numerical solution exhibits two crucial aspects of the Landau--de\,Gennes dynamics. First, since the initial condition is localized approximately to the spherical cap $y>0.6$, the Q-tensor state diffuses due to the term $L\,\mathcal{P}\divM\nablaM{\bQ}$ in \eqref{LdG_dynamics}. Second, the order parameter $s$ is zero away from a spherical cap that expands downwards (light blue on Figure~\ref{diff_tr}a). The nonlinear term $\mathcal{P}F'[\bQ]$ in \eqref{LdG_dynamics}, associated with the double-well potential $F[\bQ]$, drives the order parameter $s$ everywhere to the value $s_+$ that minimizes $F[\bQ]$ according to Lemma \ref{L:double-well}. In addition, the director field $\bq$ stays parallel to the initial value $\bq_0$ to minimize the elastic energy in \eqref{eldg}. The solution is thus uniaxial and given by \eqref{w1}. This is corroborated by Figure~\ref{diff_tr}b, which depicts the biaxiality parameter $\beta[\bQ]$ defined in \eqref{biax}. In fact, $\beta[\bQ]\approx0$ for all times in the entire surface, which is only possible if $\bQ$ is uniaxial according to Lemma \ref{L:biaxiality}. Therefore, the uniaxial evolution of the Q-tensor field is preferable to avoid competition between the elastic and potential energies that give rise to $E_{LdG}[\bQ,\nabla_M\bQ]$ in  \eqref{eldg}, provided the initial director field $\bq_0$ is constant and the corresponding elastic energy vanishes.

\vskip 0.5cm
\begin{figure}[ht!]
\begin{center}
{
 a)\quad
\begin{overpic}[abs,height=0.25\textwidth,
unit=1mm,
]
{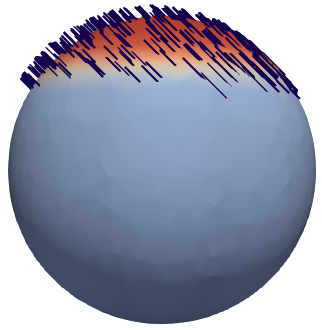}
\put(14, 40){\small{$t = 0$}}
\end{overpic}
\hskip 0.1cm
\begin{overpic}[abs,height=0.25\textwidth,
unit=1mm,
]
{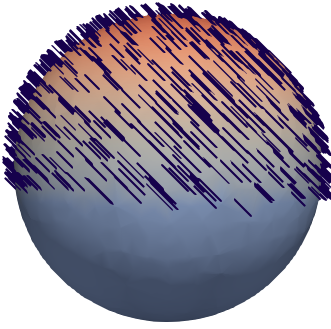}
\put(14, 40){\small{$t = 0.2$}}
\end{overpic}
\hskip 0.1cm
\begin{overpic}[abs,height=0.25\textwidth,
unit=1mm,
]
{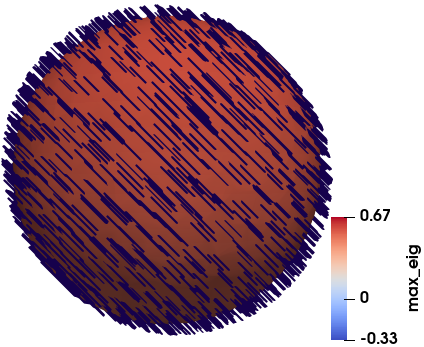}
\put(14, 40){\small{$t = 0.8$}}
\end{overpic}
}
\end{center}

\begin{center}
{
 b)\quad
\begin{overpic}[abs,height=0.25\textwidth,
unit=1mm,
]
{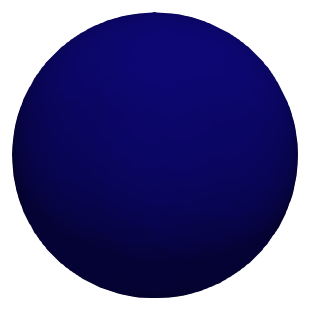}
\end{overpic}
\hskip 0cm
\begin{overpic}[abs,height=0.25\textwidth,
unit=1mm,
]
{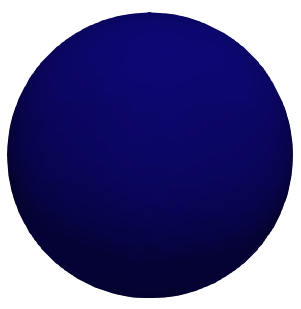}
\end{overpic}
\hskip 0.1cm
\begin{overpic}[abs,height=0.26\textwidth,
unit=1mm,
]
{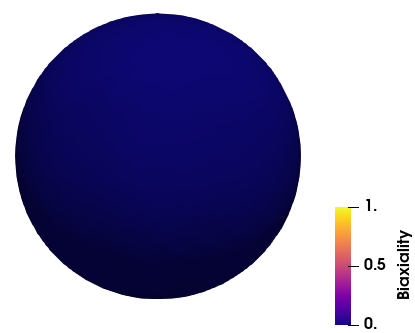}
\end{overpic}
}
\end{center}
\caption{
\small{Diffusion along the unit sphere $\Gamma$ of a Q-tensor $\bQ$ with initial uniaxial condition $\bQ_0$ given by \eqref{passive_model} localized to a spherical cap. All pictures show the $xy$-plane view: a) the largest eigenvalue and corresponding eigenvector of $\bQ$; b) the biaxiality parameter $\beta[\bQ]\approx 0$ 
for all times indicates that $\bQ$ is always uniaxial according to Lemma \ref{L:biaxiality}. The localized Q-tensor diffuses along $\Gamma$ while staying parallel to $\bQ_0$ to minimize the elastic energy. At the same time, the double-well potential $F[\bQ]$ drives the scalar order parameter $s$ of $\bQ$ to the
minimizer $s_+$ of $F[\bQ]$ stated in \eqref{s_min}     .} }
\label{diff_tr}
\end{figure}

\subsubsection{Evolution of a uniform Q-tensor under passive rotation.}
This example couples the Landau--de\,Gennes dynamics on the unit sphere $\Gamma$ of the previous example with a passive rotation. We prescribed the tangential velocity $\bv(x,y,z)=\pi\be_z\times(x,y,z)$ for $(x,y,z)\in\Gamma$ and replace the time derivative $\partial_t\bQ$ in \eqref{LdG_dynamics} with the corotational derivative $\overset{\circ}{\bQ}$ of \eqref{corotation_defintiion2}. We consider the initial
$\bQ_0=s_0 \big(\bq_0\otimes\bq_0-\frac13\bI\big)$ in \eqref{passive_model} with uniform director field $\bq_0$ given by \eqref{ICs} but with the non-equilibrium value $s_0=0.1$ of the order parameter $s$. This results in the following initial value problem
\begin{equation}\label{BE_dynamics}
\begin{aligned}
    {\bf H} +\mathcal{P}F'[\bQ] &=   L\,\mathcal{P}\divM\nablaM{\bQ}\,, 
    \\
                  \overset{\circ}{\bf Q}    &= M{\bf H} \,.
\end{aligned}
\end{equation}
with $\bQ(0)=\bQ_0$. Figure~\ref{noneq} documents the evolution for $t\in[0,1]$. Since the prescribed velocity $\bv$ is a rotation around the $z$-axis, the solution consists of the concatenation of diffusion without velocity with a rigid rotation. Since the initial condition $\bQ_0$ is uniform, the elastic energy is zero and only the double-well potential $F[\bQ]$ is active to drive the order parameter $s$. This is precisely what Figure~\ref{noneq}a illustrates: the eigenframe of $\bQ$ at each point of $\Gamma$ rotates by an angle $\pi$ in the plane orthogonal to $\be_z$ while $s$ evolves towards the minimizer $s_+$ of the potential $F[\bQ]$ given by \eqref{s_min}. Moreover, $\bQ$ remains uniaxial for all time because the biaxiality parameter $\beta[\bQ]\approx0$ in light of Figure~\ref{noneq}b, whence Lemma \ref{L:biaxiality} applies. This example reveals the essential role of the corotational derivative \eqref{corotation_defintiion2} in modeling liquid crystals on surfaces in that it does not generate spurious biaxial states during a passive dynamics of the eigenframe of $\bQ$.

\begin{figure}[ht!]
\vskip0.2cm
\begin{center}
{
 a)\quad
\begin{overpic}[abs,height=0.25\textwidth,
unit=1mm,
]
{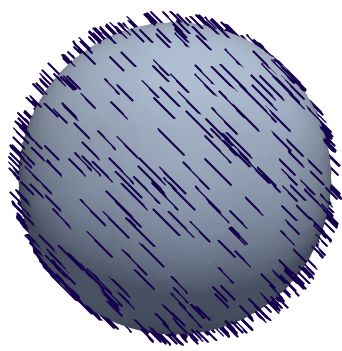}
\put(14, 38){\small{$t = 0$}}
\end{overpic}
\hskip 0.1cm
\begin{overpic}[abs,height=0.25\textwidth,
unit=1mm,
]
{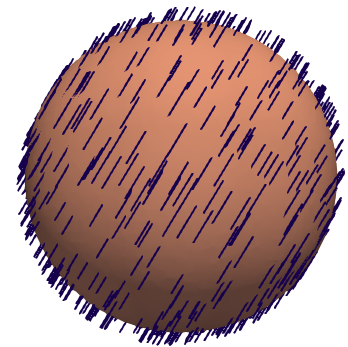}
\put(14, 38){\small{$t = 0.5$}}
\end{overpic}
\hskip 0.1cm
\begin{overpic}[abs,height=0.25\textwidth,
unit=1mm,
]
{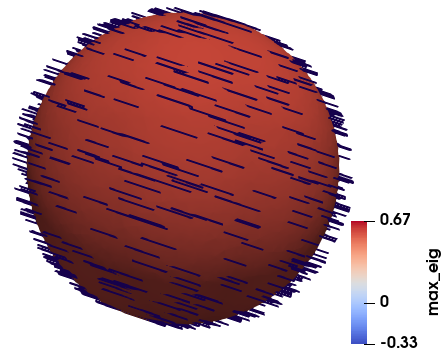}
\put(14, 38){\small{$t = 1$}}
\end{overpic}
}
\end{center}

\begin{center}
{
 b)\quad
\begin{overpic}[abs,height=0.25\textwidth,
unit=1mm,
]
{corot/hair_biaxiality_0.png}
\end{overpic}
\hskip 0cm
\begin{overpic}[abs,height=0.25\textwidth,
unit=1mm,
]
{corot/hair_biaxiality_0.png}
\end{overpic}
\hskip 0.1cm
\begin{overpic}[abs,height=0.25\textwidth,
unit=1mm,
]
{corot/hair_biaxiality_2.png}
\end{overpic}
}
\end{center}
\caption{\small{Evolution of a uniform Q-tensor under passive rotation around the $z$-axis with velocity field $\bv=\pi\be_z\times(x,y,z)$ for $(x,y,z)\in\Gamma$. All pictures show the $xy$-plane view: a) the largest eigenvalue and corresponding eigenvector of $\bQ$; b) the biaxiality parameter of \eqref{biax} satisfies $\beta[\bQ]\approx0$ for all times. Lemma \ref{L:biaxiality} implies that the Q-tensor remains uniaxial for all times. In fact, $\bQ$ is always uniform in space and rotates rigidly with $\bv$, whence the elastic energy vanishes. The order parameter $s$ evolves uniformly in space from $s_0=0.1$ to the minimizer $s_+$ of the potential energy $F[\bQ]$ in \eqref{s_min}.}}
\label{noneq}
\end{figure}

\subsection{Coupling forces in the momentum equation of surface Beris--Edwards model}\label{coupling_forces}
This set of experiments explore the action of the forces on the momentum equation \eqref{beris-edwards}c
\begin{equation}\label{momentum}
\rho\big(\pa_t\bu +(\nabla_\Gamma \bu)\bu+\nabla_{\Gamma}p \big) = 2\mu \bP\divG\bD_\Gamma(\bu)+ \bLambda + \bbf_E - \bbf_*  \,, 
\end{equation}
namely the Leslie force $\bLambda$ in \eqref{stresses}, and the Ericksen $\bbf_E$ and star $\bbf_*$ forces in \eqref{ericksen-star}
\begin{equation}\label{coupled-forces}
    \bLambda = - \bH : \nabla_M \bQ,
    \qquad
    \bbf_E = \bP \divG{\bSigma_\Gamma},
    \qquad
    \bbf_* = 2 \bB \bSigma \bn.
\end{equation}
We deal with the following configuration of $\bQ$ lying in the $xz$-plane and described in terms of polar coordinates $(r,\phi)$, i.e. $\phi=\mathrm{atan2}{(x, z)}\,, r=\sqrt{x^2+z^2}$.
Let $\omega\ge0$ be a parameter that controls the swirled director field $\bq_\omega$ perpendicular to $\be_y = (0,1,0)$
\begin{equation}\label{swirl_dir}
    \bq_\omega=\bq_\omega(r,\phi)= \big(\cos(\phi+\omega{}r),  0, \sin(\phi+\omega{}r) \big),
\end{equation}
and let the order parameter $s_{k,\xi}(r)$ vary between $0$ and $s_+$ defined in \eqref{s_min} via a regularized radial step function which is the logistic sigmoid with midpoint $\xi$ and width $k$
\begin{equation*}
    s_{k,\xi}=s_{k,\xi}(r)=\frac{s_+}{1+\exp{}(-2k(r-\xi))}.
\end{equation*}
The Q-tensor is uniaxial with eigenvector $\bP\bq_\omega$ tangential to $\Gamma$ and order parameter $s_{k,\xi}$, namely
\begin{equation}\label{swirl_Q}
\bQ = \bQ[k;\xi;\omega]={s_{k,\xi}}\left(\frac{\bP\bq_\omega}{|\bP\bq_\omega|}\otimes\frac{\bP\bq_\omega}{|\bP\bq_\omega|}-\frac13\bI\right)\, .
\end{equation}
This is a regularized degree $+1$ defect because at the origin, where $\bq_\omega$ becomes singular, the order parameter $s_{k,\xi}$ is about zero. The largest eigenvalue of $\bQ$ is $\lambda_{max}=2/3s_+ \approx 0.82$ with $s_+$ defined in \eqref{s_min}.
Moreover, the physical parameters of the fluid are its density $\rho=0.1$ and viscosity $\mu=0.1$ in \eqref{momentum}.

\begin{remark}
We point out that to generate non-zero Ericksen stresses $\bSigma$ we need a configuration of the Q-tensor with a swirled and regularized director field $\bq_\omega$. Figure~\ref{homeo_uniform} shows homeotropic and uniform Q-tensors on a spherical cap for which all coupling forces in \eqref{coupled-forces} are zero. Hence, no transport of momentum appears in such Q-tensor configurations.
\end{remark}

\begin{figure}[ht!]
\vskip0.2cm
\begin{center}
{
\begin{overpic}[abs,width=0.3\textwidth,
unit=1mm,
]
{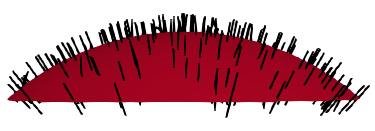}
\end{overpic}
\hskip 1cm
\begin{overpic}[abs,width=0.3\textwidth,
unit=1mm,
]
{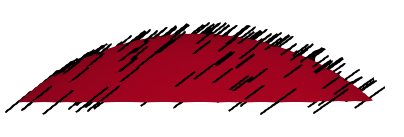}
\end{overpic}
}
\end{center}

\caption{\small{Q-tensors on spherical caps for which all the forces in \eqref{coupled-forces} are zero. Left: a homeotropic Q-tensor $\bQ={s_+}\left(\bn\otimes\bn-\bI/3\right)$, right: a uniform Q-tensor  $\bQ={s_+}\left((1,1,1)\otimes(1,1,1)-\bI\right)/3$.  }  }
\label{homeo_uniform}
\end{figure}

\subsubsection{Leslie force on a flat disc}\label{S:Leslie-flatdisc}

We first examine $\bLambda$ in \eqref{coupled-forces}. To this end, we consider a flat disc $\Gamma$ of radius 5 orthogonal to $(0,1,0)$  with Dirichlet boundary conditions for the Q-tensor $\bQ$ and the velocity $\bu$. We let $\omega=0$ and $\bQ$ be the radial uniaxial regularized defect $\bQ[5; 2.5; 0]$ of degree $+1$ defined in \eqref{swirl_Q} without a ``swirl". The order parameter $s_{k,\xi}$ of $\bQ$ is almost flat except near the circle of radius $r=\xi=2.5$, whence its gradient is radial and points outwards. Figure~\ref{leslie})b displays $\bQ$.

\begin{figure}[ht!]
\vskip0.2cm
\begin{center}
{
\begin{overpic}[abs,height=0.35\textwidth,
unit=1mm,
]
{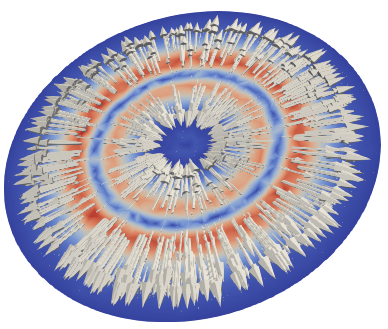}
\end{overpic}
\hskip 0.1cm
\hskip 0.1cm
\hskip 0.1cm
\begin{overpic}[abs,height=0.35\textwidth,
unit=1mm,
]
{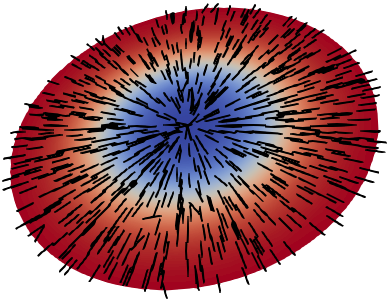}
\end{overpic}
}
\end{center}

\caption{\small{ Leslie force $\bLambda=-\bH:\nablaM\bQ$ (left) for the Q-tensor $\bQ[5; 2.5; 0]$ in \eqref{swirl_Q} on a flat disc of radius 5; this is a regularized defect of degree $+1$ with order parameter $s_{k, \xi}$ about zero at the origin (right). 
The gradient of $s_{k, \xi}$, which is radial and points outwards, is mostly responsible for the structure of $\bLambda$ (left). In fact, its concavity flips in the transition region near the circle of radius $r=\xi=2.5$, thereby resulting in a radial $\bLambda$ that point outwards for $r>\xi$ and inwards for $r <\xi$. }  }
\label{leslie}
\end{figure}
 
 The molecular field $\bH$ might be thought of approximately as the Laplacian of $\bQ$, whence it changes sign around $r=\xi$ where the convexity of $s_{k,\xi}$ flips to concavity. Since $\nabla_M\bQ$ must be radial, because of symmetry arguments,
 the Leslie force $\bLambda=-\bH:\nablaM\bQ$ is also radially symmetric and points inwards for $r < \xi$ and outwards for $r > \xi$. This is shown in Figure~\ref{leslie})a.

Since the disc $\Gamma$ is flat, the shape operator $\bB=\nabla_M\bn$ vanishes and so does the star force $\bbf_*=2\bB\bSigma\bn$ in \eqref{momentum}. In addition, computations reveal that the Ericksen tensor 
$\bSigma=\bQ\bH-\bH\bQ$ is zero and so is the Ericksen force $\bbf_E=\bP\divG\bSigma_\Gamma$. Therefore, the only active force is the Leslie force $\bLambda$, which is not divergence-free according to Figure~\ref{leslie}. Computations also show that $\bLambda$ does not produce fluid flow because the velocity
is $\bu=\bzero$, which in turn implies that $\bLambda$ is a gradient equilibrated by the pressure term to enforce the incompressibility condition $\divG\bu=0$.

\subsubsection{Ericksen force on a flat annulus}\label{S:Ericksen-flat}
We now examine the impact of the tangent Ericksen stress $\bSigma_\Gamma=\bP(\bQ\bH-\bH\bQ)\bP$ and corresponding Ericksen force $\bbf_E=\bP\divG\bSigma_\Gamma$ on the momentum equation \eqref{momentum}. We consider the flat annulus $\Gamma$ of inner radius 1 and outer radius 5 which is orthogonal to $\be_y=(0,1,0)=\bn$. Note that removing the inner disc gets rid of the defect at the origin. 
Throughout $\Gamma$ we take the order parameter $s$ to be the constant $s_+$ defined in \eqref{s_min} and the swirl parameter $\omega = 0.1$ in the definition \eqref{swirl_dir} of the director $\bq_\omega$. We consider the following uniaxial Q-tensor
\begin{equation}\label{swirl_Q_const}
\bQ ={s_{+}}\left({\bq_\omega}\otimes{\bq_\omega}-\frac13\bI\right)\, ,
\end{equation}
which is depicted in Figure~\ref{donut}b. The value $\omega\ne0$ is responsible for the Ericksen tensor
$\bSigma=\bQ\bH-\bH\bQ\ne0$, for otherwise radial symmetry forces $\bH$ and $\bQ$ to have the same conforming, radial eigenframe at every $\bx\in \Gamma$ and $\bSigma=0$. On the other hand, $s_+$ minimizes the double-well potential, according to Lemma \ref{L:double-well} (minimizer of $F$), and $\bQ$ in \eqref{swirl_Q_const} satisfies $\mathcal{P}F'[\bQ]=0$.

\begin{figure}[ht!]
\vskip0.2cm
\begin{center}
{
\begin{overpic}[abs,height=0.35\textwidth,
unit=1mm,
]
{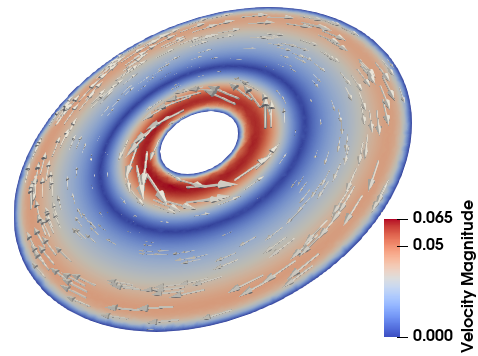}
\end{overpic}
\hskip 0.1cm
\hskip 1cm
\begin{overpic}[abs,height=0.35\textwidth,
unit=1mm,
]
{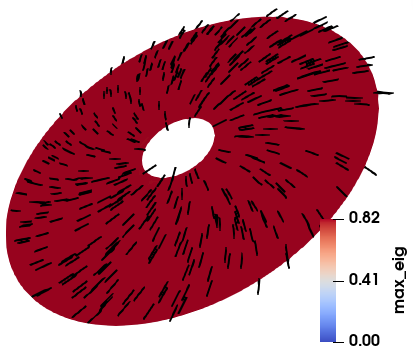}
\end{overpic}
}
\end{center}

\vskip0.2cm
\begin{center}
{
\begin{overpic}[abs,height=0.35\textwidth,
unit=1mm,
]
{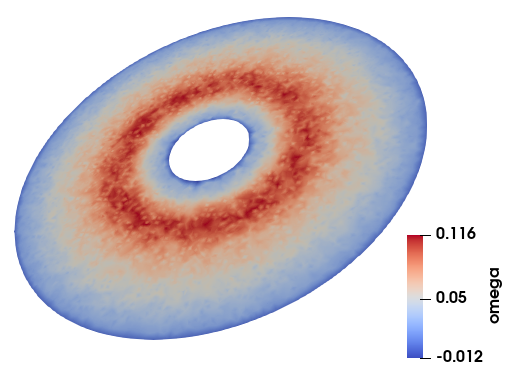}
\end{overpic}
\hskip 0.1cm
\hskip 0.1cm
\begin{overpic}[abs,height=0.35\textwidth,
unit=1mm,
]
{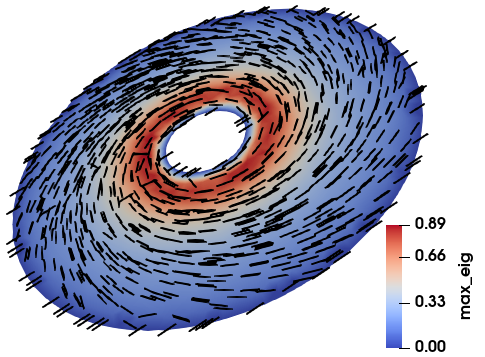}
\end{overpic}
}
\end{center}

\caption{\small{Velocity field $\bu$ produced by the Ericksen force $\divG\bSigma_\Gamma$ on an annulus $[1,5]\times{S}^1$ (top left). Uniaxial tensor $\bQ$ given by \eqref{swirl_Q} with swirl parameter $\omega=0.1$ and constant order parameter $s_{k,\xi}=s_+$ defined in \eqref{s_min} (top right). The complex flow exhibits two regions of rotation, the outer one clockwise and the inner one counterclockwise, separated by an stagnation layer of vanishing velocity. The parameter $\theta$ from \eqref{Ericksen_alternative} corresponds to the unit normal $\bn=(0, 1, 0)$ of $\Gamma$ pointing upwards (bottom left). The vector $\nablaG\theta$ is radial and points towards the stagnation layer in both the inner and outer annuli. Therefore, the Ericksen force $\bbf_E=\bn\times\nablaG\theta$ is rotational and mimics the velocity. The largest eigenvalue and corresponding eigenvector of the tensor $\bH$ (bottom right), which is also conforming with zero eigenvalue in the normal direction. The eigenframes of $\bQ$ and $\bH$ do not coincide.} }
\label{donut}
\end{figure}

We impose Dirichlet boundary conditions to both the Q-tensor $\bQ$ and the velocity $\bu$, and report the computational results in Fig. \ref{donut}. It turns out that the Ericksen force $\bbf_E$ generates a rotational incompressible flow with two distinct regions of rotation separated by a stagnation layer with zero velocity; this is displayed on Figure~\ref{donut}a. In the inner region the flow rotates \textit{counterclockwise} while in the outer region the liquid crystal material flows in the opposite direction. Since the Q-tensor, displayed in Figure~\ref{donut}b, has a uniformly \textit{clockwise} swirl on the entire annulus, one might wonder what originates this complex flow.

First, we investigate analytically the structure of the molecular field $\bH$ and surface Ericksen stress $\bSigma_\Gamma$ for $\bQ$ in \eqref{swirl_Q_const}. If we consider the basis $\bq_1=\bq_\omega, \bq_2=\bn\times\bq_\omega, \bn$, which is an eigenframe for the uniaxial Q-tensor $\bQ$, then $\bQ$ and the projector $\bP$ may be represented by
 \begin{align*}
     \bQ=\begin{pmatrix}
    \frac23s_+ & 0 & 0\\
    0 & -\frac13s_+ & 0\\
    0 & 0 & -\frac13s_+
  \end{pmatrix}\,,\qquad    \bP=\begin{pmatrix}
    1 & 0 & 0\\
    0 & 1 & 0\\
    0 & 0 & 0
  \end{pmatrix}\,.
 \end{align*}
 On the other hand, the tensor $\bH$ satisfies \eqref{BE_dynamics} with $\mathcal{P}F'[\bQ]=0$, whence
 \[
 \bH\bn = L \mathcal{P} \divM (\nablaM \bQ) \, \bn = 0
 \]
 due to the flatness of $\Gamma$ and conformity of $\bQ$. Since $\bH$ is traceless and symmetric, we get
 \begin{align}\label{H}
     \bH=\begin{pmatrix}
    a & b & 0\\
    b & -a & 0\\
    0 & 0 & 0
  \end{pmatrix}\,,
 \end{align}
 in the basis $\bq_1,\bq_2,\bn$ for suitable functions $a$ and $b$. The Ericksen stress
 $\bSigma=\bQ\bH-\bH\bQ$ reads
 \begin{align}\label{si}
     \bSigma = \begin{pmatrix}
    0 & {s_+ b} & 0\\
    -{s_+b} & 0 & 0\\
    0 & 0 & 0
  \end{pmatrix}
 \end{align}
 in the same basis and shows that $b\neq0$ is required for a nontrivial $\bSigma$.
 In other words, the eigenframe of $\bH$ should not coincide with that of of $\bQ$ for $\bSigma\ne0$, as alluded to earlier \VY{in Section~\ref{S:Ericksen-flat}}. Moreover, comparing $\bSigma=\bSigma_\Gamma$ in \eqref{si} with \eqref{Omega-n}, the function $\theta$ in \eqref{Ericksen_alternative} satisfies
 \begin{align}\label{theta-b}
     \theta=-s_+b\,,\qquad \bbf_E=\bP\divG\bSigma_\Gamma=-s_+ \bn\times\nablaG{}b \, .
 \end{align}
 
Intuitively, the molecular field $\bH$ enters the expression of $\pa_t\bQ$ in (\ref{beris-edwards}b), and a nonzero off-diagonal component, $b\neq{}0$, in \eqref{H} means that the eigenframe of $\bH$ rotates relative to that of $\bQ$. In this sense, the surface Ericksen force $\bbf_E$ in \eqref{theta-b} encodes the spatial rate of change of the eigenframe rotation: the linear momentum is the thermodynamic counterpart of the relative rotation of the molecular field $\bH$ from \eqref{H} due to the swirl structure of $\bQ$ in \eqref{swirl_Q_const}.

 We next provide a computational justification for the intriguing flow in Figure~\ref{donut}a. We resort to the parameter $\theta$ in \eqref{Ericksen_alternative}, which provides the representation $\bbf_E=\bn\times\nablaG\theta$ of the Ericksen force according to Lemma \ref{L:alt_erick}. The scalar field $\theta$ is displayed on Figure~\ref{donut}c, whence its gradient $\nablaG\theta$ is radially symmetric and pointing towards an annulus where $\theta$ exhibits its largest value; hence $\nablaG\theta$ changes orientation from an inner to an outer annular region. Therefore, the Ericksen force $\bbf_E$ is rotational and exhibits the same structure as the velocity field on Figure~\ref{donut}a with inner and outer regions of counterclockwise and clockwise orientation.


\subsubsection{A regularized swirled defect on a flat disc}

We next combine the effects of the Ericksen force $\bbf_E=\divG\bSigma_\Gamma$ and the Leslie force $\bLambda=-\bH~:~\nablaM\bQ$ in one single experiment. We consider the swirled regularized Q-tensor $\bQ[5; 2.5; 0.1]$ defined in \eqref{swirl_Q} with swirl parameter $\omega = 0.1$ and transition parameter $\xi=2.5$ on a flat disc $\Gamma$ of radius 5 orthogonal to $(0,1,0)$. The computational results are shown on Figure~\ref{ericksen}. The parameter $\xi$ characterizes the green layer on Figure~\ref{ericksen}b where the incompressible flow changes the direction of rotation. For $r<\xi$ the fluid rotation is counterclockwise
according Figure~\ref{ericksen}a, which also depicts $\bQ$, namely both the swirl director field $\bq_\omega$ and order parameter $s_{k,\xi}$ in \eqref{swirl_Q}. Moreover, in Figure~\ref{ericksen}b we display the
profiles of $\bbf_E$ and $\bLambda$ and realize that the fluid flow is consistent with the rotational character of $\bbf_E$ and the azimuthal component of $\bLambda$.

\begin{figure}[ht!]
\vskip0.2cm
\begin{center}
{
a)\hskip 0.2cm \begin{overpic}[abs,height=0.38\textwidth,
unit=1mm,
]
{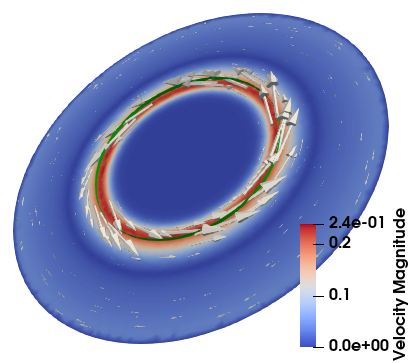}
\end{overpic}
\hskip 0.5cm
\hskip 0.5cm
\begin{overpic}[abs,height=0.38\textwidth,
unit=1mm,
]
{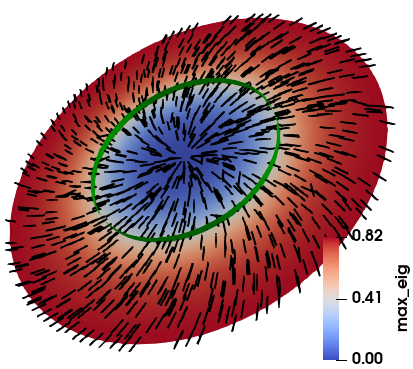}
\end{overpic}
\vskip 0.2cm
b)\hskip 0.2cm \begin{overpic}[abs,height=0.4\textwidth,
unit=1mm,
]
{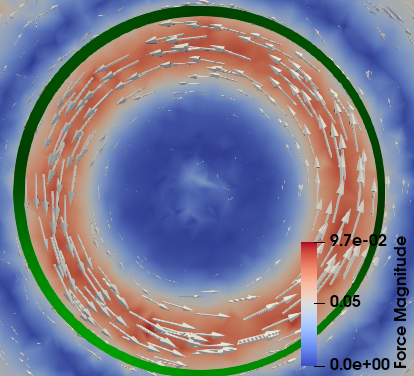}
\end{overpic}
\hskip 0.5cm
\hskip 0.5cm
\begin{overpic}[abs,height=0.4\textwidth,
unit=1mm,
]
{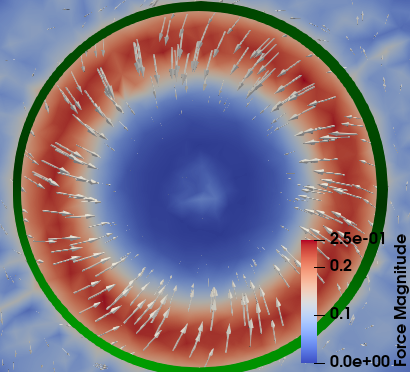}
\end{overpic}
}
\end{center}

\caption{\small{ a) Velocity field $\bu$ (left) created by the swirled regularized defect $\bQ[5; 2.5; 0.1]$ defined in \eqref{swirl_Q} (right) on a flat disc $\Gamma$ of radius 5. 
b) Ericksen force $\bbf_E=\divG\bSigma_\Gamma$ (left) and Leslie force $\bLambda=-\bH:\nablaM\bQ$ (right) near the transition region near  $r=\xi=2.5$ depicted by the green circle.  Note that $\bLambda$ exhibits a larger magnitude than $\bbf_E$, but only the rotational part of $\bLambda$ may generate incompressible flow.
The rotational flow is mostly due to $\bbf_E$.}  }
\label{ericksen}
\end{figure}

Notice that, in contrast to the experiment in Section \ref{S:Leslie-flatdisc} where $\omega=0$, the Leslie force $\bLambda$ has a radial and an azimuthal component. The former is absorbed into the pressure and does not create linear momentum as in Section \ref{S:Leslie-flatdisc}. However, the latter adds to the Ericksen force $\bbf_E$ to give rise to an inner region $r <\xi$ with clockwise rotational flow and an outer region $r>\xi$ with opposite flow. We observe that the velocity magnitude is much larger in the transition region $r \approx \xi$ than in the inner and outer regions and that, even though $\bbf_E$ is smaller in magnitude than $\bLambda$, it is mostly responsible for the counterclockwise flow.

\subsubsection{The star force on a unit sphere}
The star force $\bbf_*=2\bB\bSigma\bn$ is zero on flat geometries because the shape operator $\bB=\nabla_M\bn$ vanishes. In this experiment we consider the unit sphere $\Gamma$, for which $\bB=\bP$, to show the non-trivial behavior of $\bbf_*$ even for a surface with a simple shape operator. To demonstrate the action of $\bbf_*$ on $\Gamma$ we choose regularized unixial Q-tensors defined in \eqref{swirl_Q}: radial $\bQ[5, 0.5, 0]$ with $\omega=0$ and swirled $\bQ[5, 0.5, 1.5]$ with $\omega=1.5$. Figure~\ref{star} and Figure~\ref{swirl_star} show these configurations with unit vector $\be_y=(0,1,0)$ pointing upwards. The definition \eqref{swirl_Q} of $\bQ$ is relative the $xz$-plane perpendicular to $\be_y$, so the transition region occurs at $r =\sqrt{x^2+z^2} = \xi = 0.5$.

\begin{figure}[ht!]
\vskip0.2cm
\begin{center}
{

\begin{overpic}[abs,height=0.4\textwidth,
unit=1mm,
]
{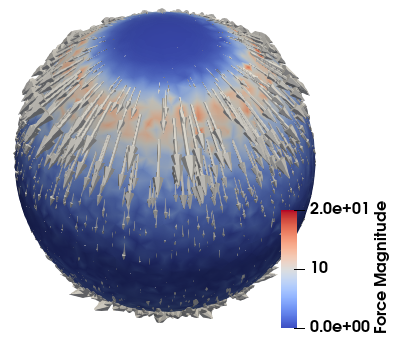}
\end{overpic}
\hskip 1cm
\begin{overpic}[abs,height=0.4\textwidth,
unit=1mm,
]
{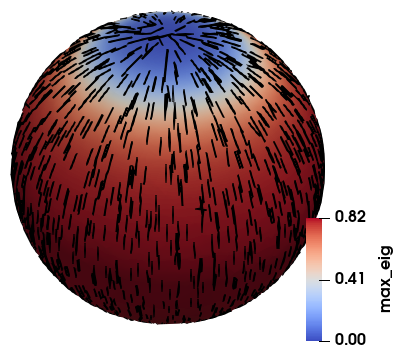}
\end{overpic}
}
\end{center}
\caption{\small{Star force $\bbf_*=2\bB\bSigma\bn$ (left) for a regularized radial defect $\bQ[5, 0.5, 0]$ of degree $+1$ defined in \eqref{swirl_Q} on the unit sphere $\Gamma$ with $\be_y$ pointing upwards (right). The Leslie force $\bLambda$ is similar to the flat disc (Figure \ref{leslie}) and the Ericksen force $\bbf_E$ is zero (neither is shown). The total force generates no flow.}  }
\label{star}
\end{figure}

\begin{figure}[ht!]
\vskip0.2cm
\begin{center}
{
\begin{overpic}[abs,height=0.4\textwidth,
unit=1mm,
]
{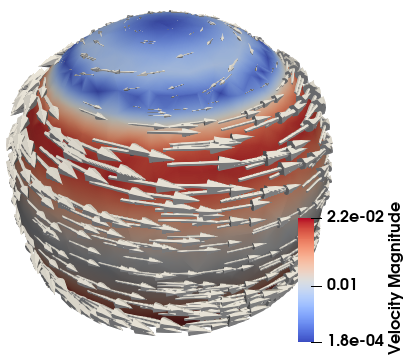}
\end{overpic}
\hskip 0.8cm
\begin{overpic}[abs,height=0.4\textwidth,
unit=1mm,
]
{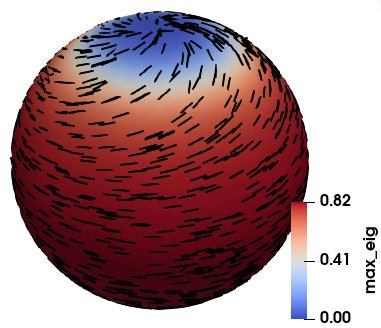}
\end{overpic}
}
\end{center}
\caption{\small{Velocity field $\bu$ (left) produced by the star force $\bbf_*=2\bB\bSigma\bn$ for a regularized swirled defect $\bQ[5, 0.5, 1.5]$ of degree $+1$ defined in \eqref{swirl_Q} (right) on a unit sphere $\Gamma$.  The field $\bu$ rotates clockwise near the defect at the north pole and counterclockwise near the equator.}  }
\label{swirl_star}
\end{figure}



It turns out that the ``radial" tensor $\bQ[5, 0.5, 0]$ generates no flow. This is because the star force $\bbf_*=2\bP\bSigma\bn=2\bSigma\bn$ has a radial structure and is localized near the transition region $r\approx\xi$, where the Ericksen stress $\bSigma = \bQ\bH - \bH\bQ$ is non-zero; hence $\bbf$ seems to be a corotational gradient that is compensated by $\nabla_\Gamma p$ in \eqref{momentum}. Moreover, the Leslie force $\bLambda=-\bH:\nablaM\bQ$ (not shown on Figure~\ref{star}) is also radial and similar to that in Figure~\ref{leslie}, whence it can also be absorbed into the pressure term. However, the profiles of $\bbf_*$ and $\bLambda$ are quite different: the former has a direction pointing towards the equator in both the upper and lower spherical caps of $\Gamma$ (see Figure~\ref{star}a), whereas the latter flips its direction near the transition region as in Figure~\ref{leslie}a. Finally, the Ericksen force $\bbf_E=\bP\divG\bSigma_\Gamma$ appears to be zero while $\bSigma\bn$ is clearly not.

In contrast, the swirled tensor $\bQ[5, 0.5, 1.5]$, shown on Figure~\ref{swirl_star}b, creates a force $\bbf_*$ that generates flow. Viewed from the north pole, such a flow develops an outer region $r>\xi$, where the velocity rotates counterclockwise, as well as an inner region $r<\xi$, where the velocity rotates clockwise but is much smaller in magnitude than the former (see Figure~\ref{swirl_star}a). On the othet hand, $\bQ$ swirls clockwise 
(see Figure~\ref{swirl_star}b).

\subsection{Relaxation of a flat-degenerate state}\label{sec_nonconformity}

In Section \ref{intro} we argued that assuming Q-tensors to be conforming, namely to obey \eqref{conf_def}, may be inconsistent with their surface dynamics unless an additional (penalty) energy enforces this configuration. We now explore such inconsistency computationally on a simple configuration of the Q-tensor on the unit sphere $\Gamma$. The initial configuration is a flat-degenerate Q-tensor field with zero normal eigenvalue (see Definition \ref{flat-degenerate}), while the final configuration is uniaxial with $\bn^T\bQ\bn=\frac23s_+$ (see Definition \ref{uni-bi}) and $s_+=1.5$ given in \eqref{s_min}. We will see that the intermediate states are generally non-conforming even if we penalize the lack of conformity, unless the penalty parameters are sufficiently large.

To describe $\bQ_0 = \bQ(0)$ in Figure \ref{nonconformity}, let $\be_y=(0,1,0)$ point upwards and let the director field $\bm=\bP\be_y/|\be_y|$ be tangent to the unit sphere $\Gamma$, where
$\bP=\bI-\bn\otimes\bn$. Then, let
\begin{align}\label{tangentQ}
    \bQ_0 = s_+\left(\bm\otimes\bm-\frac12\bP\right)
\end{align}
be a flat-degenerate Q-tensor with degree $+1$ defects at both north pole $y=1$ and south pole $y=-1$. Therefore, the biaxiality parameter $\beta[\bQ_0]=1$ defined in \eqref{biax} attains the largest possible value, according to Lemma \ref{L:biaxiality}), at all points of $\Gamma$ except for the defects.
Since minimizers of the double-well potential $F[\bQ]$ are uniaxial states  (Lemma~ \ref{L:double-well}) and $\bQ_0$ is far from uniform and carries large elastic energy at the defects, we expect $\bQ_0$ to be far from a minimizer of the Landau--de\,Gennes energy $E_{LdG}[\bQ,\nabla_M\bQ]$ in \eqref{eldg}. In fact, the final configuration is a uniaxial state \eqref{uniax} with director field $\bq=\bn$ and orientational order $s=s_+$,
whence
\begin{equation}
    \bn^T \bQ \bn = \frac{2}{3} s_+
\end{equation}
is the eigenvalue in the normal direction.
Flat-degenerate Q-tensors are prototypical for simulations in flat, two-dimensional domains. However, we stress that the evolution of $\bQ_0$ involves non-conforming Q-tensors with three non-zero eigenvalues.

\begin{figure}[ht!]
\vskip 1cm
\begin{center}
{
\begin{overpic}[abs,width=0.24\textwidth,
unit=1mm,
]
{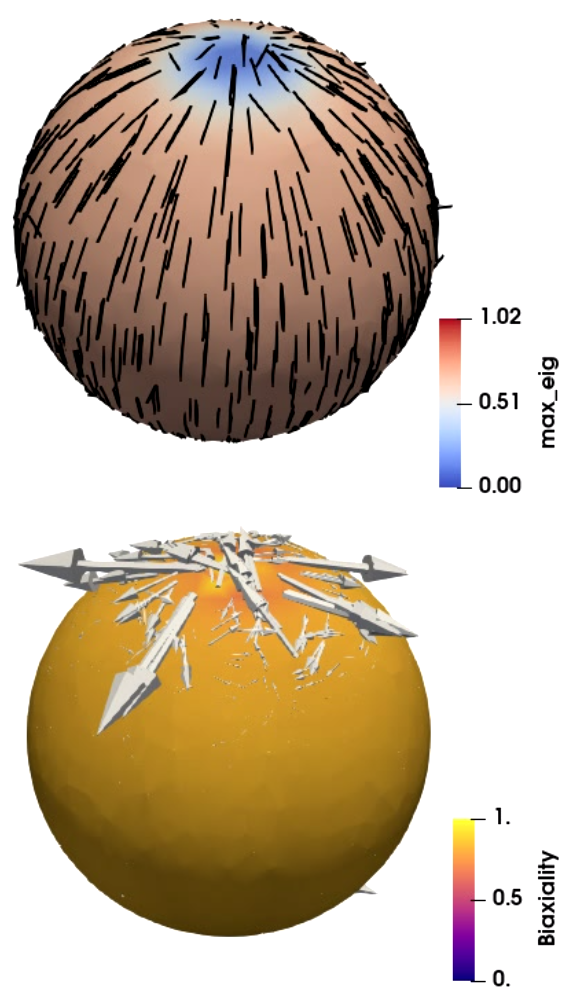}
\put(8,65){\small{$t = 0.005$}}
\end{overpic}
\begin{overpic}[abs,width=0.24\textwidth,
unit=1mm,
]
{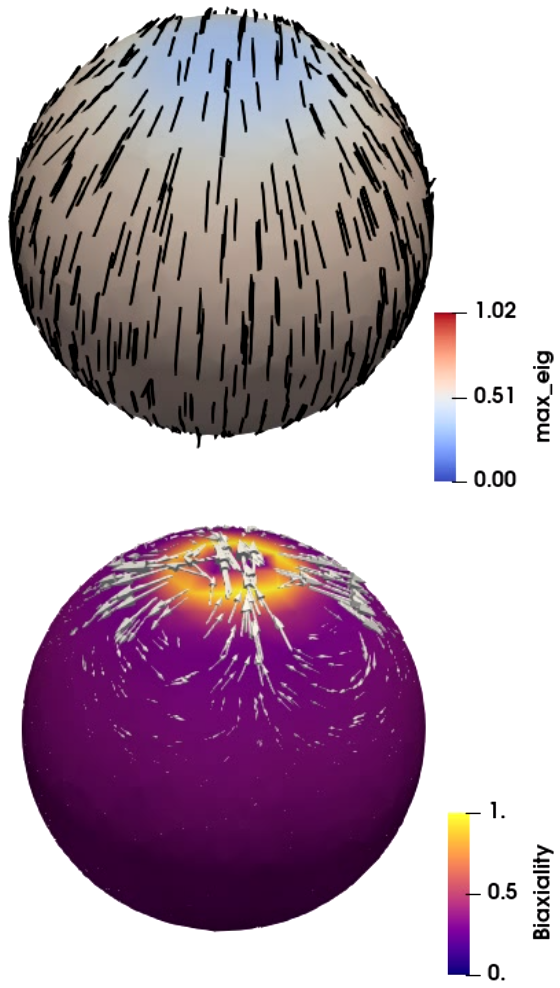}
\put(8,65){\small{$t = 0.015$}}
\end{overpic}
\hskip 0.1cm
\begin{overpic}[abs,width=0.24\textwidth,
unit=1mm,
]
{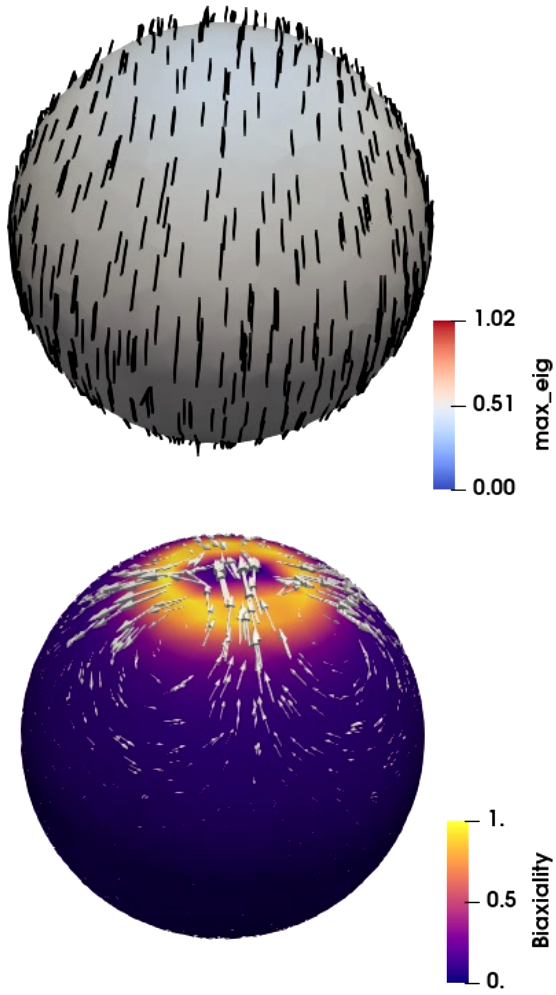}
\put(8,65){\small{$t = 0.023$}}
\end{overpic}
\begin{overpic}[abs,width=0.24\textwidth,
unit=1mm,
]
{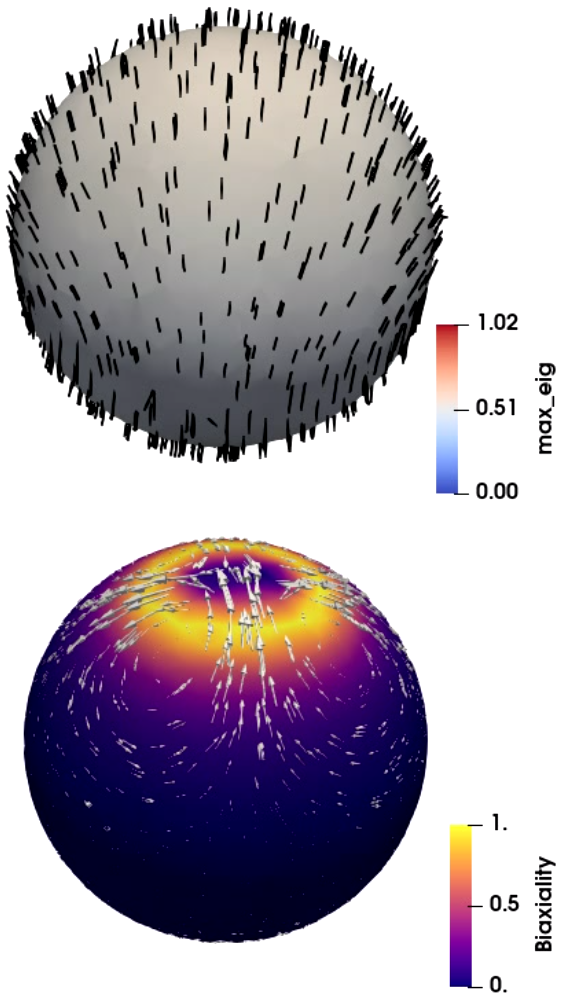}
\put(8,65){\small{$t = 0.03$}}
\end{overpic}
}
\end{center}
\vskip 1cm
\begin{center}
{
 \quad
\begin{overpic}[abs,width=0.24\textwidth,
unit=1mm,
]
{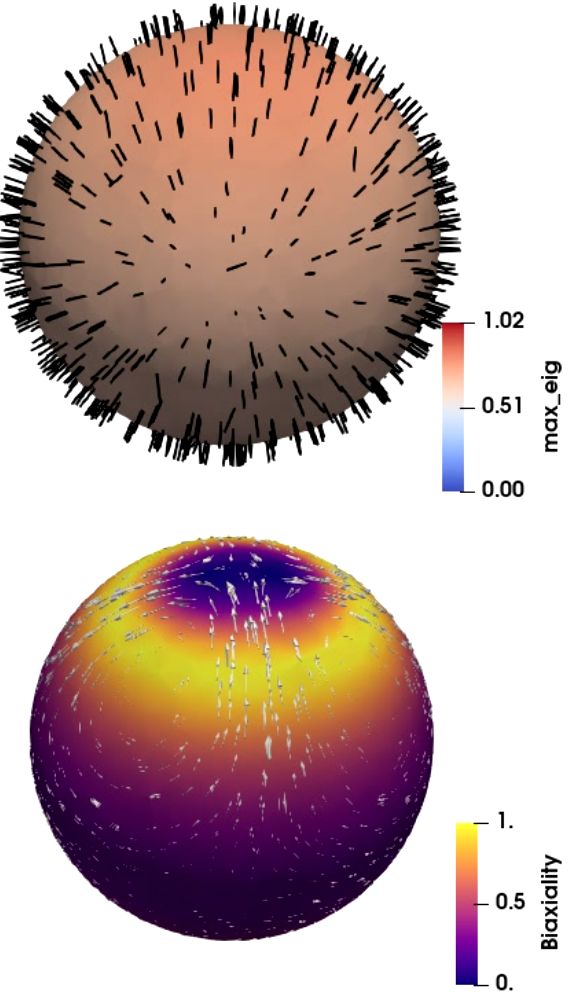}
\put(8,65){\small{$t = 0.05$}}
\end{overpic}
\hskip 0.1cm
\begin{overpic}[abs,width=0.24\textwidth,
unit=1mm,
]
{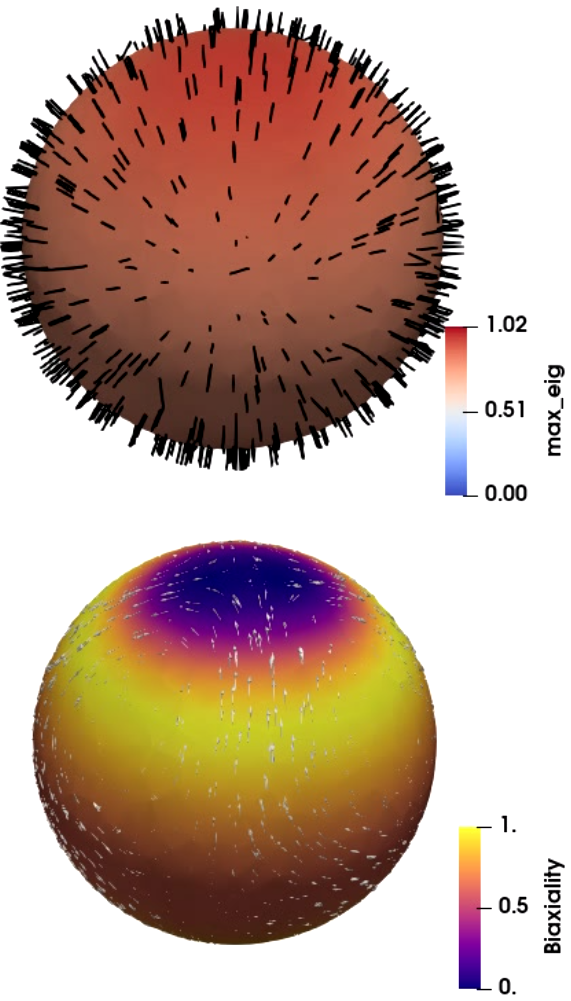}
\put(8,65){\small{$t = 0.075$}}
\end{overpic}
\begin{overpic}[abs,width=0.24\textwidth,
unit=1mm,
]
{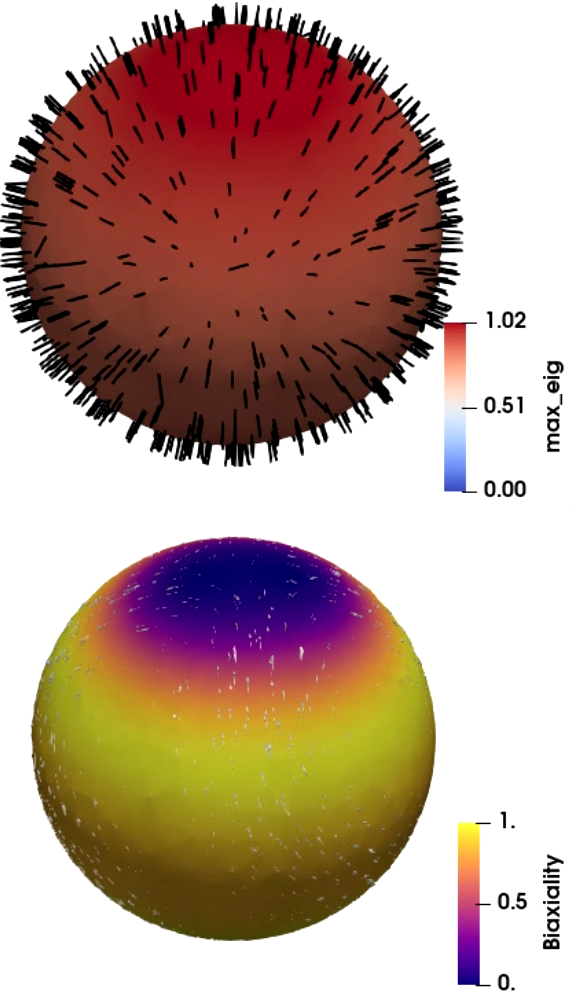}
\put(8,65){\small{$t = 0.1$}}
\end{overpic}
\begin{overpic}[abs,width=0.24\textwidth,
unit=1mm,
]
{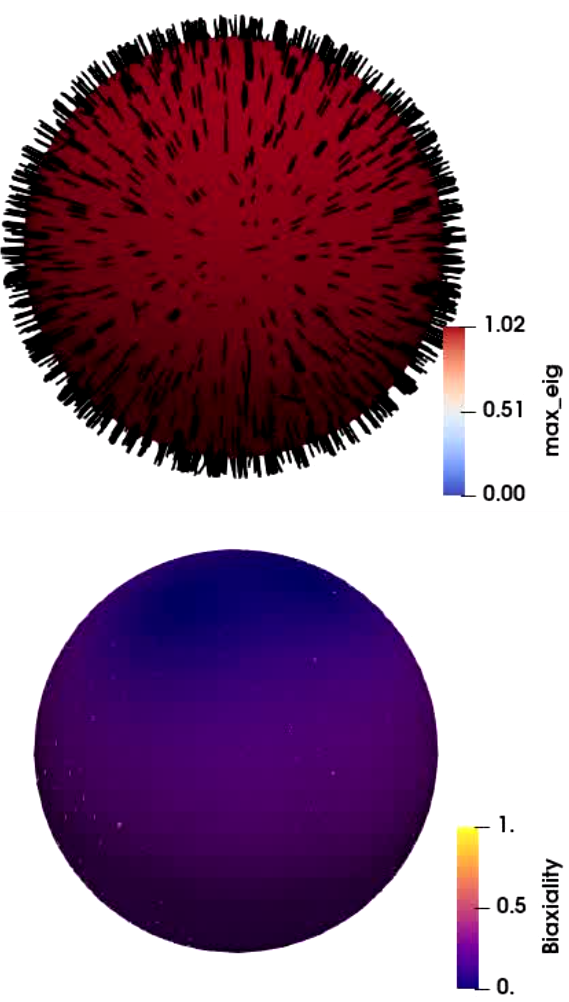}
\put(8,65){\small{$t = 0.4$}}
\end{overpic}
}
\end{center}
\caption{\small{Relaxation of the Q-tensor from the axisymmetric flat-degenerate state $\bQ_0$ in \eqref{tangentQ} to the uniaxial state \eqref{uniax} with $\bq = \bn$ and $s=s_+$. The interface parameters
in the energy \eqref{E_norm} are $\alpha=10$, $\delta=\frac{2}{3}s_+$. The Q-tensor relaxes from  $\bn^T\bQ\bn=0$ to $\bn^T\bQ\bn=\frac{2}{3}s_+=1$ passing through non-conforming states. Top: maximum eigenvalue and corresponding eigenvector of $\bQ$ evolve from tangential to normal to $\Gamma$. Bottom: biaxiality parameter $\beta[\bQ]$ and velocity field $\bu$ (scaled by $0.5$); $\beta[\bQ]$ varies from $0$ to $1$ with intermediate alternating regions of biaxiality. Vertical direction corresponds to the director $\be_y$.}  }
\label{nonconformity}
\end{figure}

\subsubsection{Normal anchoring penalization}\label{S:normal_anchoring}

As discussed in Section \ref{S:osipov-hess}, 
the Landau--de\,Gennes energy \eqref{eldg} of a liquid crystal film may include, in some applications \cite{golovaty2015dimension,golovaty2017dimension,nestler2020properties,osipov1993density}, the energy $E_{norm}[\bQ]$ with penalty parameter $\alpha>0$ defined in \eqref{pen+norm}, which enforces a value $\delta$ of orientational order in the normal direction
\begin{align}\label{E_norm}
E_{norm}[\bQ]=\alpha\int_\Gamma \big|\bn^T\bQ\bn-\delta \big|^2\,.
\end{align}

The dynamics of the Beris--Edwards system is dictated by the competition of several energies: the double-well potential $F[\bQ]$ promotes uniaxial states with order parameter $s_+$; the elastic energy $L |\nablaM\bQ |^2/2$ promotes uniform states in $\R^3$; and the energy $E_{norm}[\bQ]$ promotes a certain degree of orientational order, but does not affect the conformity.
A rigid condition of the form $\bn^T\bQ\bn=c$, as discussed in Section \ref{intro}, is often a modelling assumption postulated along with the conformity assumption. If this condition is relaxed but the conformity assumption is still applied, then $\bn^T\bQ\bn$ is an additional scalar variable representing the normal orientational order. One could model a transition from the conforming flat-degenerate state \eqref{tangentQ} with $\bn^T\bQ\bn=0$ to the conforming uniaxial state with $\bn^T\bQ\bn=\delta$ enforcing conformity of $\bQ$ for all intermediate times. However, our simulations show that our Beris-Edwards model find non-conforming intermediate states more energetically favorable.

We simulate the full surface Beris--Edwards system \eqref{beris-edwards} with initial conditions $\bu(0)=0$ and $\bQ(0)=\bQ_0$ given by \eqref{tangentQ}, as well as the augmented Landau--de\,Gennes energy \eqref{eldg} by \eqref{E_norm}. This leads to the following variant of \eqref{beris-edwards}a
\begin{equation}\label{beris-edwards_aug}
    {\bf H} +\mathcal{P}(F'[\bQ]+E_{norm}'[\bQ]) =   L\,\mathcal{P}\divM\nablaM{\bQ}\,,
\end{equation}
where $E_{norm}'[\bQ]$ is the variational derivative of $E_{norm}$. We choose the parameters
\[
a=-1, \, b=1, \, c=1;
\quad
M=1, \, L=1, \, \rho=0.1, \, \mu=0.1;
\quad
\alpha=10, \, \delta=\frac23s_+ = 1.0
\]
in \eqref{dw}, \eqref{beris-edwards} and \eqref{E_norm} respectively. We report  on Figure \ref{nonconformity} the numerical results for time evolution of the augmented surface Beris--Edwards system \eqref{beris-edwards}--\eqref{beris-edwards_aug}. The Q-tensor relaxes from the flat-degenerate state with $\bn^T\bQ\bn=0$ to the uniaxial state with $\bn^T\bQ\bn=\frac{2}{3}s_+$ passing through non-conforming states. The biaxiality parameter $\beta[\bQ]$ is uniform at the beginning and end of the simulation, with values $\beta[\bQ]=1$ (biaxial) to
$\beta[\bQ]=0$ (uniaxial) respectively, and exhibits alternating and space-dependent values in between. The energy landscape is complex with non-conforming intermediate states.

\subsubsection{Non-conformity penalization}\label{S:conformity}

To further check our claim of non-conformity on the transition from flat-degenerate to uniaxial
configurations, we develop a second experiment. To enforce that $\bQ\bn$ be normal, whence $\bn$ be an eigenvector of $\bQ$, we incorporate the physically justified anchoring energy 
\cite{golovaty2015dimension,golovaty2017dimension,nestler2020properties,osipov1993density}, 
\begin{align}\label{e_pen}
    E_{pen}[\bQ]=\gamma\int_\Gamma |\bP\bQ\bn|^2\, .
\end{align}
already discussed in \eqref{pen+norm}.
Therefore, the limit $\gamma\rightarrow\infty$ imposes the strong conformity condition $\bQ\bn=\lambda\bn$
because \eqref{e_pen} is minimized if $\bQ\bn$ is normal.

\begin{figure}[ht!]
\vskip 1cm
\begin{center}
{
 \quad
\begin{overpic}[abs,height=0.25\textwidth,
unit=1mm,
]
{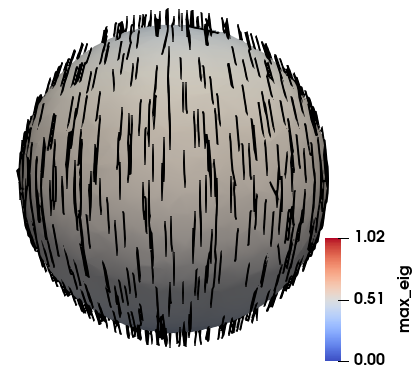}
\put(14, 38){\small{$\gamma = 0$}}
\end{overpic}
\hskip 0.1cm
\begin{overpic}[abs,height=0.25\textwidth,
unit=1mm,
]
{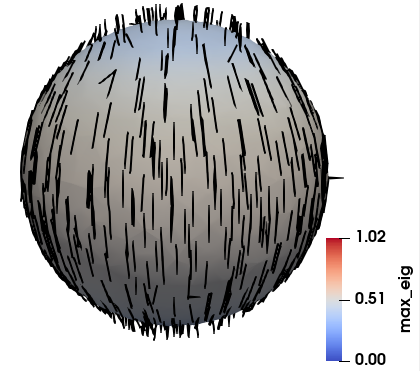}
\put(14, 38){\small{$\gamma= 1e2$}}
\end{overpic}
\hskip 0.1cm
\begin{overpic}[abs,height=0.25\textwidth,
unit=1mm,
]
{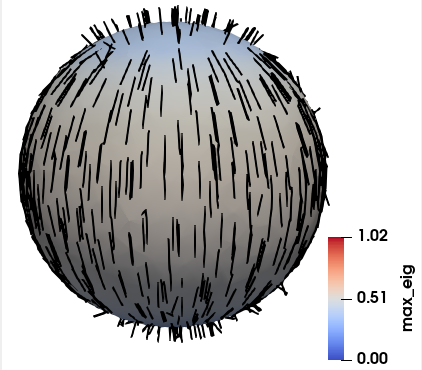}
\put(14, 38){\small{$\gamma = 1e4$}}
\end{overpic}
}
\end{center}
\begin{center}
{
 \quad
\begin{overpic}[abs,height=0.25\textwidth,
unit=1mm,
]
{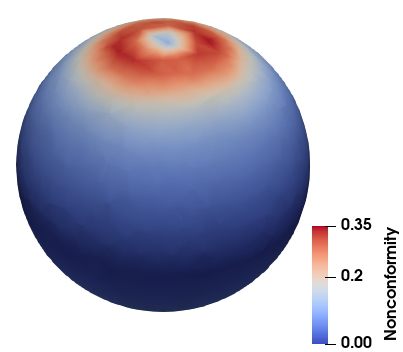}
\put(14, 38){\small{$\gamma = 0$}}
\end{overpic}
\hskip 0.1cm
\begin{overpic}[abs,height=0.25\textwidth,
unit=1mm,
]
{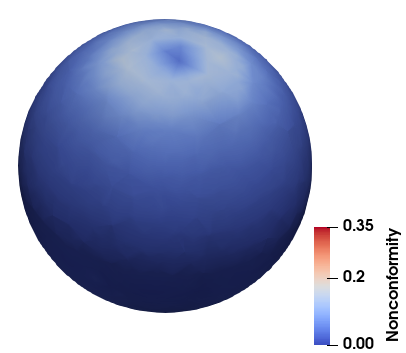}
\put(14, 38){\small{$\gamma= 1e2$}}
\end{overpic}
\hskip 0.1cm
\begin{overpic}[abs,height=0.25\textwidth,
unit=1mm,
]
{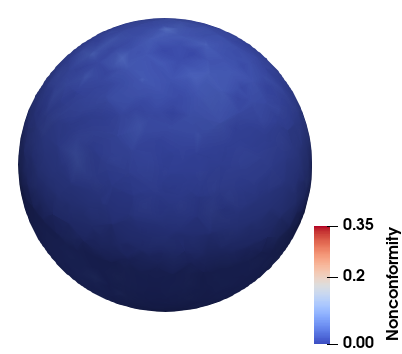}
\put(14, 38){\small{$\gamma = 1e4$}}
\end{overpic}
}
\end{center}

\caption{\small{Relaxation of the flat-degenerate Q-tensor field \eqref{tangentQ}, with director $\be_y$ pointing upwards, for different values of the parameter $\gamma$ in \eqref{e_pen} that penalizes non-conformity; $\gamma=0$ corresponds to Figure \ref{nonconformity}. All snapshots are taken for the time $t=0.023$, which is far from the final steady state. Small to moderate values of $\gamma$ give rise to intermediate non-conforming Q-tensor fields. Top: maximal eigenvalue and corresponding eigenvector. Bottom:  non-conformity parameter $r_\Gamma[\bQ]$ defined in \eqref{nonconf_param}. 
\VY{Note that, for $\gamma=1e4$, at each point of the sphere one of the eigenvectors is almost exactly normal. What is shown on the Figure is the eigenvector with the largest eigenvalue so a discontinuity may appear where two eigenvalues are equal and are largest.  }}  }
\label{conformity}
\end{figure}

We now repeat the preceding simulation of the augmented system \eqref{beris-edwards}--\eqref{beris-edwards_aug} but this time adding $E_{pen}[\bQ]+E_{norm}[\bQ]$ to the Landau--de\,Gennes energy \eqref{eldg}.  We choose
\[
\gamma=0, \quad \gamma=100, \quad \gamma=10000
\]
in \eqref{e_pen}, and display on Figure \ref{conformity} (top) the maximal eigenvalue and corresponding eigenvector of $\bQ$ at the fixed time $t=0.023$ far from equilibrium. We also report the non-conformity parameter $r_\Gamma[\bQ]$ defined in \eqref{nonconf_param} on Figure \ref{conformity} (bottom). The parameter $\gamma=0$ corresponds to the
simulations in Figure~\ref{nonconformity}.
As expected, large values of  $\gamma$ promote conformity of the Q-tensor for all times, while for small to moderate values of $\gamma$  intermediate states are non-conforming.

\subsubsection{Enforcing conforming and flat-degenerate Q-tensor dynamics}
Inspired by dynamic simulations of a conforming and flat-degenerate Q-tensor on a unit sphere from \cite{nestler2021active}, we explore the predictions of our model and numerical approach in the same
context. In fact, we show that enforcing the Q-tensor dynamics to be conforming and flat-degenerate in the normal direction via \eqref{E_norm} and \eqref{e_pen} leads to the so-called \textit{tetrahedral} configuration. This minimizing equilibrium configuration consists of four $+1/2$-defects located at the vertices of a regular tetrahedron inscribed in the unit sphere, as depicted in Figure~\ref{test-comp}.

\begin{figure}[ht!]
\vskip 1cm
\begin{center}
{
 \quad
\begin{overpic}[abs,height=0.25\textwidth,
unit=1mm,
]
{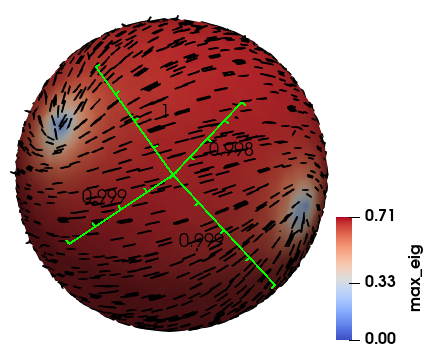}
\put(14, 38){\small{$t=0.01$}}
\end{overpic}
\hskip 0.1cm
\begin{overpic}[abs,height=0.25\textwidth,
unit=1mm,
]
{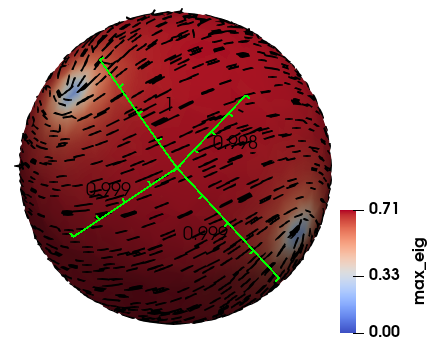}
\put(14, 38){\small{$t= 0.1$}}
\end{overpic}
\hskip 0.1cm
\begin{overpic}[abs,height=0.25\textwidth,
unit=1mm,
]
{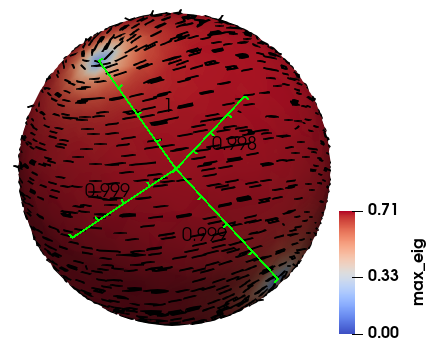}
\put(14, 38){\small{$t=2$}}
\end{overpic}
}
\end{center}

\caption{\small{Evolution towards a tetrahedral minimizer starting from a conforming,  flat-degenerate Q-tensor field \eqref{tangentQ2} with defects in the $XZ$-plane (horizontal). The Q-tensor stays conforming and flat-degenerate for all times via penalization \eqref{E_norm} and \eqref{e_pen}. The maximal eigenvalue and corresponding eigenvector are shown at times $t=0.01, 0.1, 2$, along with green segments connecting the center of the unit sphere with the vertices of a fixed regular tetrahedron. Two of the four $+1/2$-defects are visible and move towards their final positions at the vertices of the tetrahedron. The angles between the green segments are in the range of $109\pm3$ degrees which is close to the angle of the regular tetrahedron.  The vertical direction corresponds to the $Y$-axis.
}  }
\label{test-comp}
\end{figure}

We take the initial condition proposed in \cite{nestler2021active}. It consists of two tangent vector fields on the unit sphere $\Gamma$ given by
\begin{align}
    \bq^x=\frac{\bP(0,y,z)^T}{|\bP(0,y,z)^T|}\,,\quad \bq^z=\frac{\bP(x,y,0)^T}{|\bP(x,y,0)^T|},
\end{align}
which have $+1$-defects at $(\pm 1,0,0)$ and $(0,0,\pm1)$, respectively. Next a composite vector field $\bm^{xz}$ on $\Gamma$ is defined as follows,
\begin{align}
   \bm^{xz}=\bq^x,\,y\geq{}0\,,\qquad\qquad  \bm^{xz}=\bq^z,\,y<0\,,
\end{align}
 With the help of this composite vector field we construct a conforming, flat-degenerate Q-tensor,
\begin{align}\label{tangentQ2}
    \bQ_0 = s_+\left(\bm^{xz}\otimes\bm^{xz}-\frac12\bP\right) \, .
\end{align}

The advantage of this initial configuration is that $\bQ_0$ quickly transforms into a \textit{planar} configuration of four $+1/2$-defects located in the $XZ$-plane and resembling tennis ball patches  (Figure~\ref{test-comp} (left)). This planar configuration slowly evolves towards the minimizing tetrahedral  configuration depicted on Figure~\ref{test-comp} (right). 
In this simulation we set
\[
a=-10, \, b=1, \, c=10;
\quad
M=1, \, L=0.1, \, \rho=0.1, \, \mu=0.1;
\quad
\alpha=100, \, \gamma=10000, \delta=0 \, .
\]
The choice of penalization parameters $\alpha, \gamma, \delta$ ensures that the Q-tensor stays conforming and flat-degenerate in the normal direction for all times.

We observe from Figure~\ref{test-comp} that the final equilibrium configuration is still conforming and flat-degenerate  and it corresponds to the expected tetrahedral arrangement of four $+1/2$-defects that maximizes the distance between defects \cite{nestler2021active}. This shows the flexibility of our model to accommodate $Q$-tensor conformity via the Hess-Osipov energy described in Section~\ref{S:osipov-hess}.
However, this desirable consistency does not mean that our model always reduces to that in \cite{nestler2021active} in the limit $\alpha,\gamma\to\infty$ without further structural assumptions on the molecular field $\bH$. This crucial discovery in under current investigation.

\subsection{Homeotropic state: Instability and weak anchoring of Q-tensors}\label{sec_instability}
Although it should be clear that the general kinematics of Q-tensors introduced in Section \ref{kinematics} is inconsistent with the conformity assumption, we would like to demonstrate how this assumption affects the behavior of the Beris--Edwards model in a concrete example which is of standalone interest. Consider an initial condition on a unit sphere which is \textit{homeotropic} (i.e. conforming and uniaxial with respect to the normal $\bn$):
\begin{align*}
    \bQ=s_+\left(\bn\otimes\bn-\frac13\bI\right)
\end{align*}
with the constant order parameter $s_+$ from \eqref{s_min} that minimizes the double-well potential.  The configuration is rotationally symmetric and it is interesting to check if it is a stable one.
  
\begin{figure}[ht!]
\vskip 1cm
\begin{center}
{
 \quad
\begin{overpic}[abs,height=0.25\textwidth,
unit=1mm,
]
{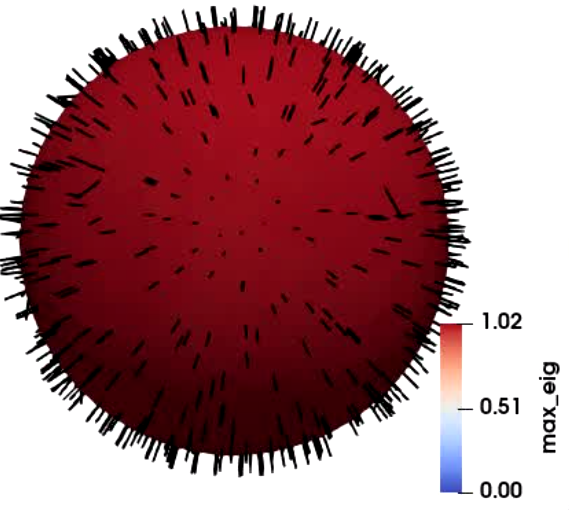}
\put(14, 38){\small{$t = 0$}}
\end{overpic}
\hskip 0.1cm
\begin{overpic}[abs,height=0.25\textwidth,
unit=1mm,
]
{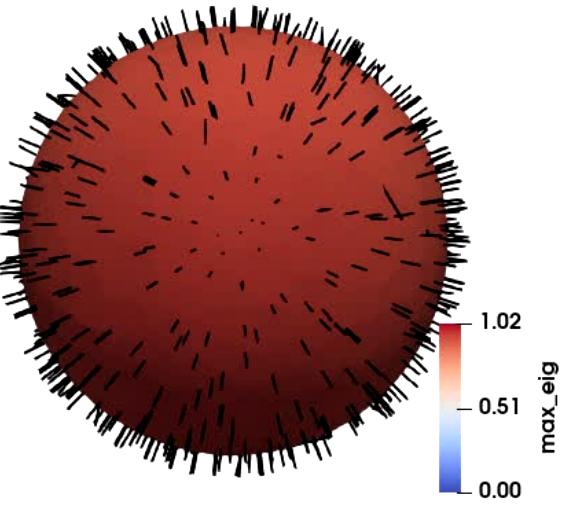}
\put(14, 38){\small{$t = 1$}}
\end{overpic}
\hskip 0.1cm
\begin{overpic}[abs,height=0.25\textwidth,
unit=1mm,
]
{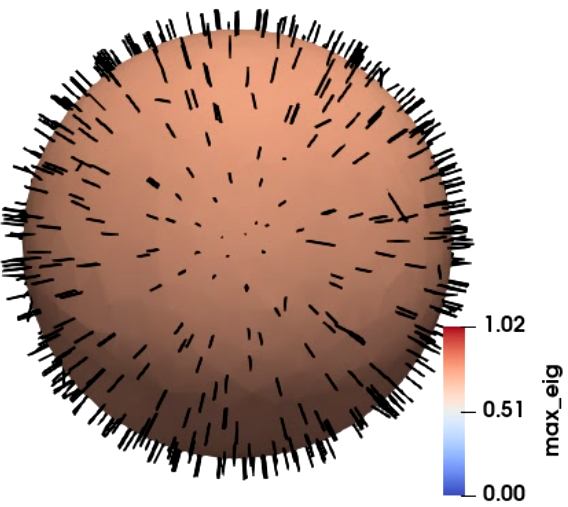}
\put(14, 38){\small{$t = 4$}}
\end{overpic}
}
\end{center}

\caption{\small{Stability of a hometropic Q-tensor anchored to the unit sphere by $\gamma=10$ in the Beris--Edwards model. The order parameter of the initial condition minimizes the double-well potential. The competition between the elastic energy and the double-well potential drives the configuration to a new homeotropic state with a constant order parameter $s$ close to $0.6$ which is shown at the most right snapshot. The Q-tensor stays  conforming and uniaxial for all times. The vertical direction corresponds to the $X$ axis.}}  
\label{stability}
\end{figure}

\begin{figure}[ht!]
\vskip 1cm
\begin{center}
{
 \quad
\begin{overpic}[abs,width=0.25\textwidth,
unit=1mm,
]
{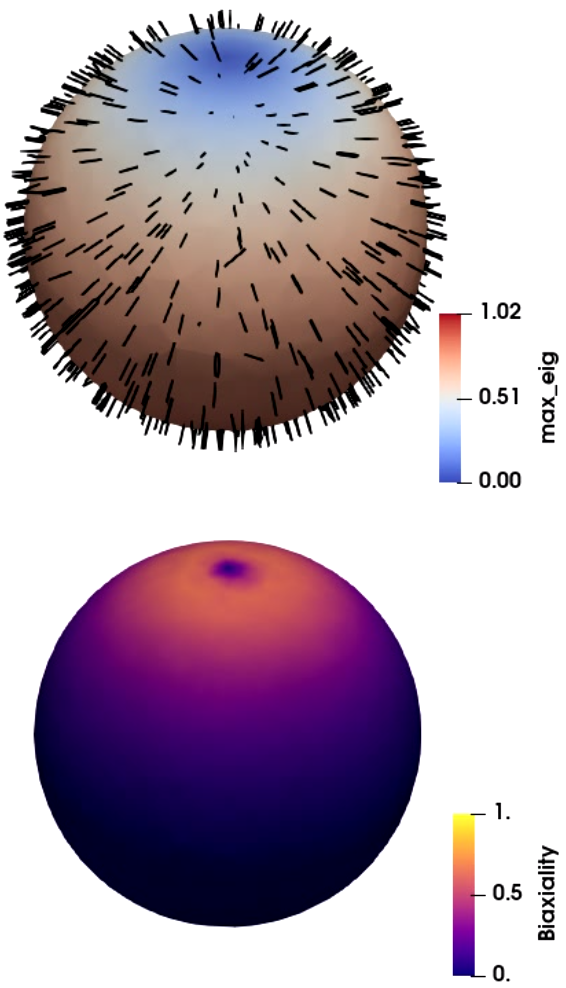}
\put(12,65){\small{$t = 5$}}
\end{overpic}
\hskip 0.1cm
\begin{overpic}[abs,width=0.25\textwidth,
unit=1mm,
]
{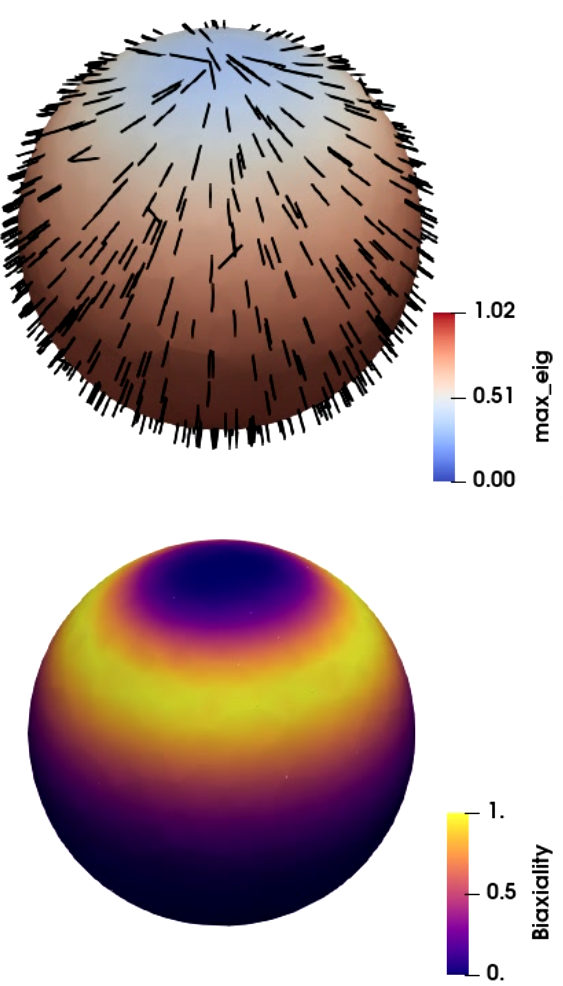}
\put(12,65){\small{$t = 10$}}
\end{overpic}
\hskip 0.1cm
\begin{overpic}[abs,width=0.25\textwidth,
unit=1mm,
]
{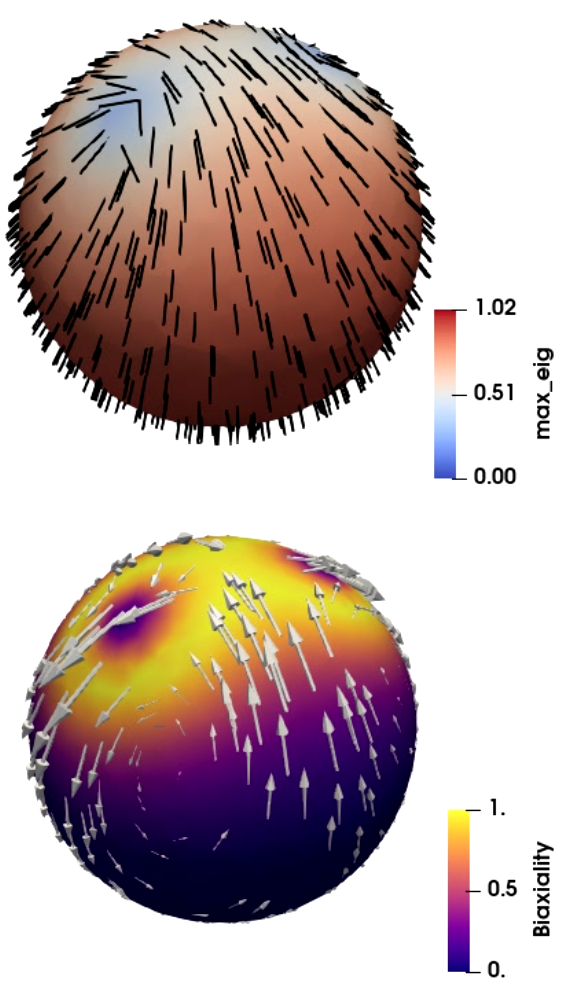}
\put(12,65){\small{$t = 12$}}
\end{overpic}
}
\end{center}
\vskip 1cm
\begin{center}
{
 \quad
\begin{overpic}[abs,width=0.25\textwidth,
unit=1mm,
]
{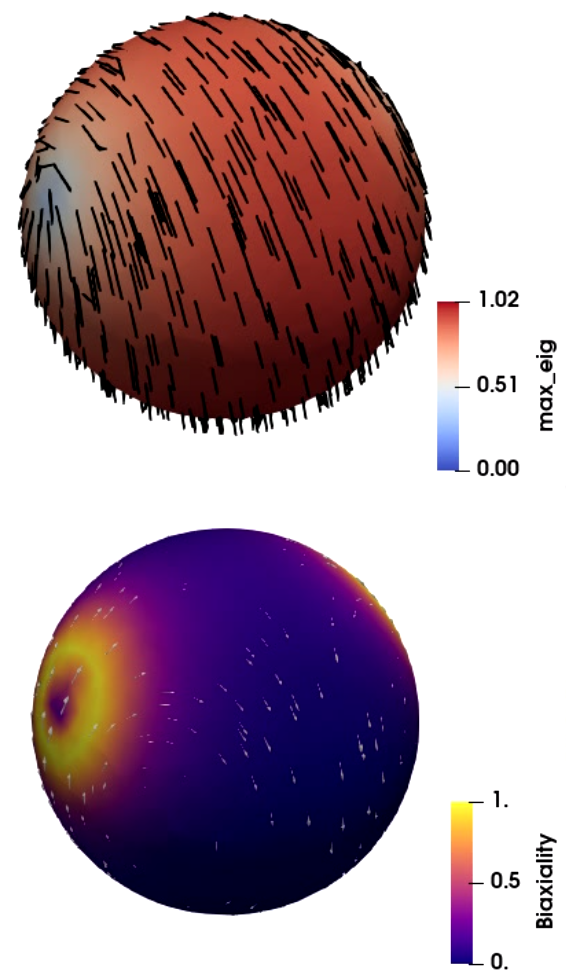}
\put(12,65){\small{$t = 16$}}
\end{overpic}
\hskip 0.1cm
\begin{overpic}[abs,width=0.25\textwidth,
unit=1mm,
]
{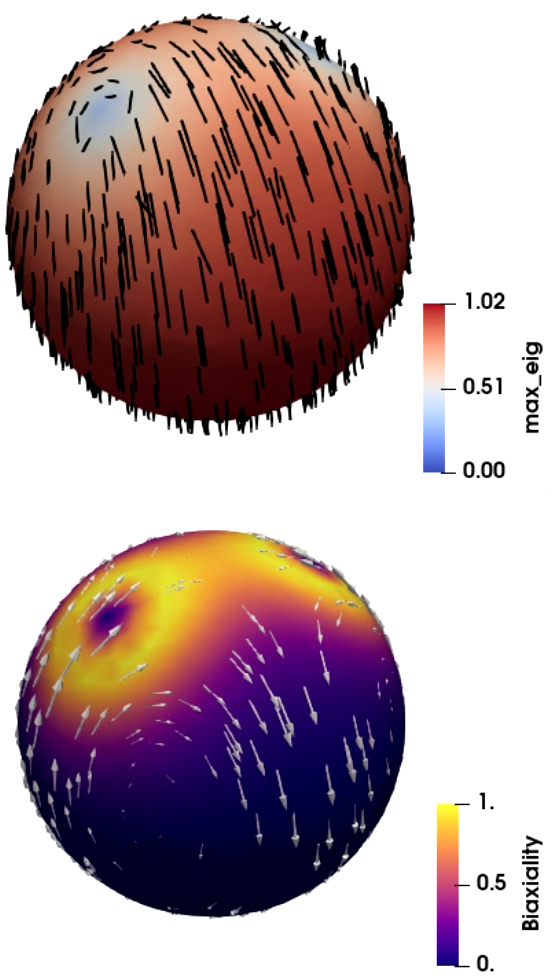}
\put(12,65){\small{$t = 21$}}
\end{overpic}
\hskip 0.1cm
\begin{overpic}[abs,width=0.25\textwidth,
unit=1mm,
]
{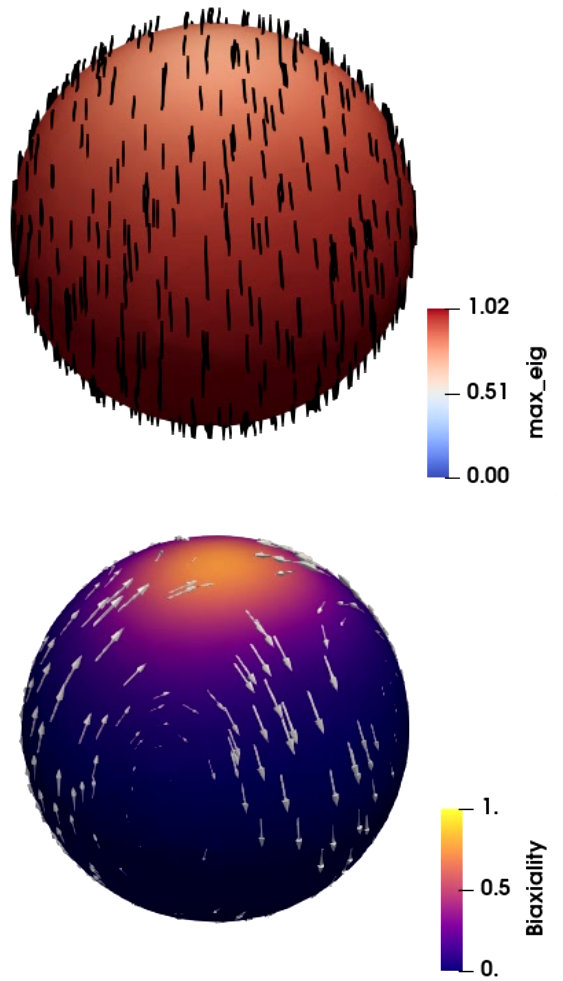}
\put(12,65){\small{$t = 23$}}
\end{overpic}
}
\end{center}
\caption{\small{Instability of the initially homeotropic, radially symmetric Q-tensor $\bQ_0$ in \eqref{perturb}, with director $\be_y=(0,1,0)$ pointing upwards, weakly anchored to the unit sphere $\Gamma$ via \eqref{e_pen} with parameter $\gamma=0.1$. The initial condition $\bQ_0$ is the same as in Figure \ref{stability} but $\gamma$ is much smaller. Since the anchoring is not sufficiently strong, the Q-tensor loses stability through non-conforming configurations, starting with the formation of a biaxial ring, followed by the splitting of the $+1$ defect in the north pole by two $+1/2$ defects that repel each other initially. They eventually coalesce to form the global minimizer - a uniaxial state, uniform in $\R^3$, with $s=s_+$ as the order parameter. Top: maximal eigenvalue and corresponding eigenvector. Bottom: biaxiality parameter $\beta[\bQ]$ and velocity field $\bu$ (scaled by 5).}  }
\label{instability}
\end{figure}

Assume the Q-tensor has to stay conforming for all times. Since the elastic energy of such radial configuration is nonzero, the \VY{elastic part of } Landau--de\,Gennes energy can be minimized by either evolving the order parameter $s$ from $s_+$ to a smaller value or even by generating a biaxial state which would brake the radial symmetry. We choose
\[
a=-1, \, b=1, \, c=1; \quad L=1, \, M=1, \, \rho=0.1, \, \mu=0.1; \quad \gamma = 10,
\]
in the double-well potential \eqref{dw},
the Beris-Edwards system \eqref{beris-edwards}, and the anchoring energy \eqref{e_pen} respectively. The effect of \eqref{e_pen} is to penalize the lack of conformity, whence the Q-tensor field evolves according to the first scenario, which slightly reduces the order parameter while keeping the Q-tensor radially symmetric and homeotropic provided $\gamma$ is large. In fact, for $\gamma\to\infty$ we expect a strong imposition of conformity. Figure \ref{stability} documents this claim for $\gamma=10$, and
reveals that the final radially symmetric, conforming, homeotropic Q-tensor configuration is stable for
moderate values of $\gamma$.

\begin{figure}[ht!]
\vskip 1cm
\begin{center}
{
 \quad
\begin{overpic}[abs,height=0.4\textwidth,
unit=1mm,
]
{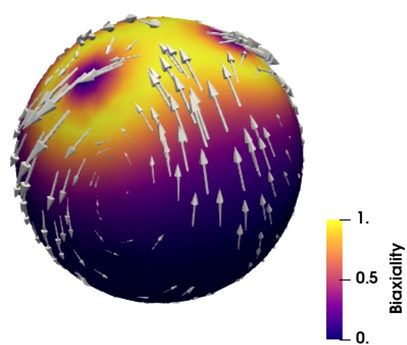}
\put(22, 58){\small{$t = 12$}}
\end{overpic}
\hskip 0.1cm
\begin{overpic}[abs,height=0.4\textwidth,
unit=1mm,
]
{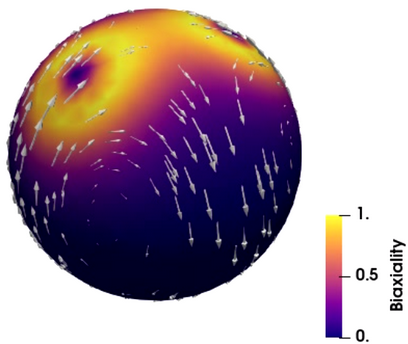}
\put(22, 58){\small{$t = 21$}}
\end{overpic}
}
\end{center}

\caption{\small{Velocity fields (scaled by 5) for specific times $t=12$ (left) and $t=21$ (right) of the same experiment as in Figure~\ref{instability}. Coloring corresponds to the biaxiality parameter \eqref{biaxiality_def}: two dark spots with low biaxiality, corresponding to Q-tensor defects, are surrounded by annuli with high biaxiality. The velocity fields for splitting (left) and merging (right) of defects are similar but opposite. This saddle-like pattern of velocity is consistent with the splitting and merging of defects for flat domains. }}  
\label{stability_zoom}
\end{figure}

To trigger the onset of instability we perturb the initial director as follows
\begin{align}\label{perturb}
    \widetilde{\bQ}_0=s_+\left(\widetilde{\bn}\otimes\widetilde{\bn}-\frac13\bI\right)\,,\quad \widetilde{\bn}=\frac{\bn+0.2(1,0,0)}{|\bn+0.2(1,0,0)|}\, ,
\end{align}
 and examine the full surface Beris--Edwards model \eqref{beris-edwards} augmented with \eqref{e_pen}
 via the small parameter value $\gamma=0.1$. The numerical results, displayed on Figure~\ref{instability}, reveal that the initial configuration $\bQ_0$ loses stability because the weak anchoring provided by $\gamma$ is not strong enough. The ensuing dynamics is quite rich: the instability manifest first via the formation of a biaxial ring with $+1$ defect in the north pole ($t=10$), which splits into two $+1/2$ defects ($t=12$).
 The nature of these defects is apparent in the display of the Q-tensor in the first row. These defects initially repel from each other ($t=16$), but later they coalesce ($t=21$ and $t=23$). Figure~\ref{stability_zoom} displays the velocity fields for these extreme stages of splitting and merging: they are similar but point in opposite directions and resemble saddle-like patterns already documented in flat cases.
 The final Q-tensor state is uniaxial uniform in $\R^3$, and close to the global minimizer $\bQ=s_+\left(\be_y\otimes\be_y-\frac13\bI\right)$ of the double-well potential with $\be_y=(0,1,0)$ pointing vertically in Figure~\ref{instability} . The final orientation is affected by the perturbation \eqref{perturb} of the initial condition.
 
 This experiment indicates that the stability and evolution of simple Q-tensor configurations depend on the penalization parameter $\gamma$ which controls the anchoring energy $E_{pen}[\bQ]$. It is thus conceivable that the actual size of $\gamma$ coming from materials science applications might not be sufficiently large to enforce the conformity assumption.

\section{Conclusions}

This paper derives and explores a novel model of liquid crystal films with general orientational order.
For a given smooth, stationary and closed surface $\Gamma$, the main contributions are:

\begin{itemize}[leftmargin=*]
    \item {\it Non-conforming Q-tensors}: We develop a new notion of Q-tensor kinematics on surfaces, which hinges on Assumption \ref{assum_vel}. We introduce the surface corotational derivative of Q-tensors \eqref{corotation_defintiion2} to transport a generically oriented Q-tensor field. This allows for transport of Q-tensors such that the unit normal vector $\bn$ to $\Gamma$ is not an eigenvector (non-conformity). In this vein, Assumption \ref{assum_Q} dictates how the eigenframe of a conforming Q-tensor is transported.
    
    \item {\it Energy law}: We invoke the generalized Onsager principle to derive a model with an energy structure that mimics the Beris--Edwards model in $\R^3$. We impose Assumption \ref{assum_Onsager} to define the structure of the evolution laws. The derivation employs extrinsic calculus in $\R^3$, thereby avoiding surface parametrizations and making finite element discretizations in $\R^3$ readily available. The surface model contains three distinct forces: the Leslie force $\bLambda$ and
    Ericksen force $\bbf_E$, which already exist in flat domains, as well as the new star force $\bbf_*$ which is responsible for thermodynamics consistency for non-conforming Q-tensors.
    
    \item {\it Simulations}: We conduct a systematic computational study of the surface Beris--Edwards model to unravel the role of several forces and mechanisms. This includes 
    
    \begin{itemize}
        \item experiments with the ``dry" surface Landau--de\,Gennes (gradient flow dynamics without linear momentum) to examine the novel Q-tensor kinematics; 
    
       \item experiments that illustrate the role of the three forces $\bLambda, \bbf_E$ and $\bbf_*$ and their profiles for a Q-tensor configuration with a degree $+1$ defect; 
       
       \item experiments of the dynamics connecting two confoming states which undergo more energetically favorable non-conforming intermediate states;
       
       \item conforming dynamics of four $+\frac12$ defects enforced via penalization that lead to a regular tetrahedral structure consistent with \cite{nestler2021active}.
       
       \item  simulations of the instability of a radially symmetric, homeotropic Q-tensor on a unit sphere due to insufficient anchoring, which showcases the effect of non-conformity.
    \end{itemize}
    The relaxation of the conformity assumption, via a thermodynamically consistent model, and computational exploration of its consequences are the main novelties of our work.
\end{itemize}    

\section{Acknowledgements}
We would like to thank Qi Wang who shared his unpublished note \cite{unpublishedWang} on a surface model of active liquid crystals in which the machinery of the generalized Onsager principle was applied using the language of local coordinate systems. We are also grateful to Qi Wang for several discussions about the Beris-Edwards model on flat domains.

Lucas Bouck was supported by the NSF grant DGE-1840340. Ricardo H. Nochetto and Vladimir Yushutin were partially supported by NSF grant DMS-1908267.

\section{Data Availability}
Data will be made available on reasonable request.

\section{Conflict of interest}
The authors declare that they have no conflict of interests.

\appendix
\section{General notation}\label{apx:calc}
\VY{In this appendix preliminary definitions and notations are clearly presented for the ease of further reading.} We adopt the matrix notation where a vector $\bx$ is represented by the column of its components $\bx_j=(\bx)_j$ in the standard basis $\be_j$ of $\R^n$, $j\in\{1,...,n\}$. A linear operator $\bA$ is represented by the $n\times{}n$ matrix $\bA_{ij}=(\bA_j)_i$ where the vector $\bA_j=\bA\be_j$ is the image of $\be_j$ under $\bA$. The gradient of a scalar function $f$ is a column of partial derivatives, $(\nabla{}f)_j=\pa_j{}f$. The matrix of the gradient of a vector field $\bu$ consists of rows of transposed gradients of the field components, $(\nabla{}\bu)_{ij}=\pa_j{}\bu_i$.   We denote $\partial_k\bA$ the matrix (vector) of $k$-th partial derivatives of a matrix (vector) $\bA$ applied component-wise. 
  
  We will often use dyads. The dyadic or tensor product of two vectors, $\bu\otimes\bv$, is a linear operator with the matrix representation $\bu\bv^T$, while the inner product $\bu\cdot\bv$ denotes the scalar $\bu^T\bv$. A linear operator $\bA$ may be represented by dyads involving either its column-vectors, $\sum_{j=1}^n\bA_j\otimes\be_j$,  or its row-vectors, $\sum_{j=1}^n\be_j\otimes(\bA^T)_j$.  Higher order tensor products can be derived from the associativity of $\bu\otimes\bv\otimes\bq$, e.g. $(\bA\otimes\bu)\bv=(\bu\cdot\bv)\bA$ and $(\bu\otimes\bA)\bv=\bu\otimes(\bA\bv)$.
  
  The standard gradient operator in $\R^n$ for a scalar field $f$, a vector field $\bu$  and a matrix field $\bA$ is defined in the language of vector algebra as follows \cite{jankuhn2018incompressible}
\begin{equation}\label{nabla}
\begin{aligned}
    \nabla{}f&=\sum_{j=1}^n\be_j\partial_j f\,,\qquad\nabla{}\bu=\sum_{j=1}^n\partial_j\bu\otimes\be_j=\sum_{j=1}^n\be_j\otimes\nabla\bu_j\,,\quad
    \\
    \nabla\bA &= \sum_{j=1}^n\partial_j\bA\otimes\be_j=\sum_{j=1}^n\be_j\otimes\nabla(\bA^T)_j 
   \end{aligned}
   \end{equation}
   The associated  
   directional derivative in $\Rn$ along a vector $\bv$ is given by
   \begin{equation}\label{dir_der}
   \begin{aligned}
    (\nabla{}^T f)\bv &=\sum_{j=1}^n\bv_j\partial_j{}f\,,\qquad(\nabla\bu)\bv=\sum_{j=1}^n\bv_j\partial_j\bu=\sum_{j=1}^n (\nabla\bu_j\cdot\bv)\be_j{}\,,
    \\
         (\nabla\bA)\bv&=\sum_{j=1}^n\bv_j\partial_j\bA
    =\sum_{j=1}^n(\nabla{}\bA_j)\bv\otimes{}\be_j 
    \end{aligned}
    \end{equation}
  \VY{Here and later $\nabla^T \! f$ is a short notation for $(\nabla f)^T$.}  We define the pointwise inner product of two matrices $\bA$ and $\bC$ as well as their gradients 
    \begin{align}\label{contractions}
    \bA:\bC=\sum_{j=1}^n\bA_j\cdot\bC_j=\sum_{i,j=1}^n\bA_{ij}\bC_{ij}\,,\quad
    \nabla\bA \tripledot \nabla{}\bC=\sum_{j=1}^n\partial_j\bA:\partial_j\bC=\sum_{j=1}^n\nabla{}\bA_j:\nabla{}\bC_j
\end{align}
The inner product of tensor fields on  $\Omega\subset\Rn$ is defined as follows:
 \begin{align*}
    (f,g)_\Omega=\int_\Omega fg\,,\qquad{}  (\bu,\bv)_\Omega=\int_\Omega \bu\cdot\bv\,,\qquad{} (\bA,\bC)_\Omega=\int_\Omega \bA:\bC
 \end{align*}
The divergence operator is given by
\begin{align}\label{bulk_div}
\div\bu&=\tr(\nabla\bu)\,,\quad    \div\bA=\sum_{j=1}^n\be_j\,\div (\bA^T)_j
               \,,\quad
               \div(\nabla\bA)= \sum_{j=1}^n\be_j\otimes\div\nabla(\bA^T)_j
\end{align}
and we want to stress that the vector divergence is applied to rows of a matrix, and the divergence of the gradient of a matrix is defined accordingly. 


The following proposition summarizes some straightforward useful identities which are consistent with the adopted notation.

\begin{proposition}[product rules]\label{appendix}
For any scalar field $f$, vector fields $\bu,\bv, \bq$ and matrix field $\bA$, we have
   \begin{align*}
 \nabla(f\bu)&=f\nabla\bu+\bu\otimes\nabla{}f\,,\quad &
 \nabla(\bu\cdot\bv)&=(\nabla^T\bu)\bv+(\nabla{}^T\bv)\bu,
 \\
 \nabla(f\bA)&=f\nabla\bA+\bA\otimes\nabla{}f
 \,,\quad &
 \nabla(\bA\bv)&=\sum_{i=1}^n\bv_i\nabla\bA_i+\bA\nabla{}\bv,
 \\
  \nabla(\bu\otimes\bq)\bv&=(\nabla\bu)\bv\otimes\bq+\bu\otimes(\nabla\bq)\bv\,,\quad&
 \div{(f\bu)}&=f\div\bu+\bu\cdot\nabla{}f,
 \\
 \div(\bu\otimes\bv)&=(\div\bv)\bu+(\nabla\bu)\bv
 \,,\quad&
     \div(f\bA)&=f\div\bA+\bA\nabla{}f,  
     \\\div(\bA^T\bu)&=\div\bA\cdot\bu+\bA:\nabla\bu, & \div(\bu\otimes\bA)&=\bu\otimes\div\bA+(\nabla\bu)\bA^T.
 \end{align*}
\end{proposition}
 \VY{Again, $\nabla^T \! \bu$ is a short notation for $(\nabla \bu)^T$.} The following proposition summarizes some integration by parts rules which are consistent with the adopted notation.
\begin{proposition}[integration by parts]
For any scalar field $f$, vector field $\bu$ and matrix field $\bA$ defined on a \VY{open} domain $\Omega\subset\Rn$ with boundary $\pa\Omega$ and outer unit normal $\bn$, we have
\begin{align*}
    (\div\bu,f)_\Omega&=(f\bu, \bn)_{\partial\Omega}-(\bu,\nabla{}f)_\Omega
  \,,\quad
    (\div\bA,\bu)_\Omega=(\bA^T\bu, \bn)_{\partial\Omega} - (\bA,\nabla\bu)_\Omega.
\end{align*}\label{bulk_part}
\end{proposition}

\section{Integration by parts and tangential decomposition of tensors on surfaces}\label{apx:by_parts}

\VY{In this appendix we present some results concerning the {external} and covariant tensor derivatives used in this paper. More specifically, we derive formulas that show the connection between the integration by parts on surfaces and tangential decomposition of tensors. Note that throughout the paper the integrals are taken component-wise, with respect to the ambient space $\mathbb{R}^3$.}
\begin{lemma}[Gauss-Weingarten] \label{GW} Given a vector field $\bu$ such that $\bu(\bx)\cdot\bn(\bx)=const$, for all $\bx\in\Omega_\delta$, the covariant and \VY{external} derivatives are related by the shape operator as follows:
\begin{align}\label{Gauss}
    \nabla_M\bu=\nablaG\bu-\bn\otimes\bB\bu\,,\quad\bx\in\Omega_\delta\,.
\end{align}
\end{lemma}
\begin{proof} For all $j\in[1,n]$ we compute
\begin{align*}
   0&=\pa_j(\bu\cdot\bn)=\pa_j\bu\cdot\bn+\bu\cdot\pa_j\bn=(\nabla\bu)_j\cdot\bn+\bu\cdot(\nabla\bn)_j\,,\quad\forall\bx\in\Omega_\delta
\end{align*}
which, due to the symmetry of $\bB=\nabla\bn=\bB^T$,  implies 
\begin{align}\label{Bu}
    (\nabla^T\bu)\bn+\bB\bu=0\,,\quad\bx\in\Omega_\delta
\end{align} 
From the definitions and the property $\bB=\bP\bB=\bB\bP$, we deduce
\begin{align*}
   \gradM\bu-\gradG\bu&=(\bI-\bP)\gradM\bu=\bN\nabla\bu\bP=\bn\bn^T\nabla\bu\bP=\bn\otimes (\bP(\nabla^T\bu)\bn)=-\bn\otimes \bP\bB\bu \, .
\end{align*}
This concludes the proof.
\end{proof}

The following corollary of Lemma \ref{GW} \RHN{is used} in the development of the surface model of liquid crystal flows.
\begin{corollary}[relation between spin tensors] For tangent vector fields, $u_N=0$, the covariant and \VY{external} spin tensors in \eqref{strain} and \eqref{strainM} are related through
    \begin{equation} \label{skewBu}
 \bW_M(\bu)=\bW_\Gamma(\bu)+\frac12  (\bB\bu\otimes\bn-\bn\otimes\bB\bu )\,.
 \end{equation}
\end{corollary}
\subsection{Vector fields }
We decompose a vector field $\bu$ on $\Omega_\delta$ into the {\it tangent} component $\bu_T$ and the {\it normal} component $u_N\bn$ as follows:
\begin{align}\label{u_decomp}
    \bu=\bu_T+u_N\bn\,,\quad\bu_T=\bP\bu\,,\quad{}u_N\bn=\bN\bu\,.
\end{align}
\begin{lemma}[covariant divergence]\label{vector_div}
For a vector field $\bu=\bu_T+u_N\bn$ on $\Omega_\delta$, we have
\begin{align*}
    \divG\bu
    &=\divG\bu_T+u_N\,\tr\bB \,.
\end{align*}
\end{lemma}
\begin{proof}
 Using the definition \eqref{surf_div}, Lemma \ref{appendix}, the cyclic property of traces and $\bP\bN=0$ we compute
\begin{align*}
\div_\Gamma\bu
    &=  \div_\Gamma(\bP\bu+\bN\bu)= \div_\Gamma(\bP\bu)+ \tr(\nabla_\Gamma(\bN\bu))= \div_\Gamma\bu_T+ \tr(\bP\nabla(\bN\bu)\bP)
    \\
&=\div_\Gamma\bu_T+\tr\bB(\bu\cdot\bn)
\end{align*}
where the last step is due to the following identity
\begin{align*}
\bP\nabla(\bN\bu)&=\sum_{i=1}^n\bu_i\bP\nabla\bN_i+\bP\bN\nabla{}\bu=\sum_{i=1}^n\bu_i\bP\nabla(\bn_i\bn)=\sum_{i=1}^n\bu_i\bP(\bn_i\nabla\bn+\bn\otimes\nabla\bn_i)
\\
&=\sum_{i=1}^n\bu_i\bP(\bn_i\nabla\bn)=(\bu\cdot\bn)\bB\,.
\end{align*}
This gives the assertion.
\end{proof}

\begin{lemma}[normal flux]\label{normal_flux}
For a normally extended vector field $\bu=\bu^e$, we have
  \begin{align*}
 \lim_{\delta\rightarrow0}&\frac{1}{2\delta}\left(\int_{\Gamma_\delta^+}\bu\cdot\bn-\int_{\Gamma_\delta^-}\bu\cdot\bn\right)=\int_\Gamma(\tr\bB)\bu\cdot\bn\,.
 \end{align*}
 \end{lemma}
 \begin{proof}
 First consider a normally extended scalar $f$:
  \begin{align*}
     \lim_{\delta\rightarrow0}&\frac{1}{2\delta}\left(\int_{\Gamma_\delta^+}f^e-\int_{\Gamma_\delta^-}f^e\right)=\frac{d}{d\delta}\left(\int_{\Gamma_\delta}f^e\right)|_{\delta=0}
     =\frac{d}{d\delta}\left(\int_{\Gamma}\det(1+\delta\bB_\delta)f^e\right)|_{\delta=0}
     \\
     &=\left(\int_{\Gamma}\frac{d}{d\delta}\det(1+\delta\bB_\delta)f^e\right)|_{\delta=0}
     =\int_\Gamma(\tr\bB)f\,.
 \end{align*}
 The proof concludes by applying this formula to products $\bu^e_i\bn_i$ of components of a vector field $\bu^e$ and the normal $\bn$, and summation over $i$.
 \end{proof}
\begin{lemma}[covariant integration by parts]\label{vector_parts}
For a vector field $\bu$ and a scalar field $f$ on $\Omega_\delta$, the integration by parts over a closed surface $\Gamma$ reads
\begin{align*}
(\divG\bu,f)_\Gamma
    &= (\tr(\bB)f\bu,\bn)_\Gamma-(\bu,\nablaG{}f)_\Gamma
    \,.\end{align*}
\end{lemma}
\begin{proof}
Since  $\divG\bu$ depends only on the values of $\bu$ on $\Gamma$ we first restrict $\bu$ to $\Gamma$ and then extend it normally obtaining $\bu^e$. We use Lemmas \ref{bulk_div} and \ref{normal_flux}  and take the limit $\delta\rightarrow0$ in
\begin{align*}
\lim_{\delta\rightarrow0}\frac{1}{2\delta} (\div{}\bu^e, f^e)_{\Omega_\delta}=\lim_{\delta\rightarrow0}\frac{1}{2\delta}\left((f^e\bu^e,\bn)_{\Gamma_\delta^+}-(f^e\bu^e,\bn)_{\Gamma_\delta^-}\right)-\lim_{\delta\rightarrow0}\frac{1}{2\delta}(\bu^e,\nabla{}f^e)_{\Omega_\delta}
\end{align*}
to obtain
\begin{align*}
(\divG{}\bu, f)_\Gamma=-(\bu,\nablaG{}f)_\Gamma+(f(\tr\bB)\bu,\bn)_{\Gamma}
\end{align*}
because $\divG\bu=\divM\bu=\div\bu^e$ and $\nablaG f=\nablaM f=\nabla f^e$ .
\end{proof}
 In view of \eqref{nablaG} and \eqref{divergence}, Lemma \ref{vector_parts} extends to tangential derivatives.

\subsection{Matrix fields}
We introduce a tangential decomposition of matrices. For an arbitrary matrix $\bA$ we compute
\begin{align*}
    \bP\bA\bP=(\bI-\bN)\bA(\bI-\bN)=\bA-\bP\bA\bN -\bN\bA\bP-\bN\bA\bN
\end{align*}
which suggests the following tangential decomposition
\begin{align}\label{matrix_decomp}
    \bA=\bA_\Gamma+\bA_N+\bA_{N\Gamma{}}\,,\quad{}\bA_\Gamma=\bP\bA\bP\,,\quad{}\bA_N=\bN\bA\bN\,,\quad{}\bA_{N\Gamma}=\bN\bA\bP+\bP\bA\bN
\end{align}

\begin{lemma}[matrix decomposition]\label{matrix_div}
For a matrix field $\bA$ on $\Omega_\delta$ the vector field $\divG\bA$ has the following tangential and normal components:
\begin{align*}
    \bP\divG\bA&=\bP\divG\bA_\Gamma+\tr(\bB)\bP\bA\bn+\bB\bA^T\bn\,,\qquad\bN\divG\bA=\bN\divG\bA_\Gamma+\divG(\bA^T\bn)\bn .
\end{align*}

\end{lemma}
\begin{proof} In view of \eqref{convG} and \eqref{surf_div}, we deduce
  $\divG\bA=\divG(\bA_\Gamma+\bA\bN+\bN\bA-\bN\bA\bN)$ and treat each term separately. By definition \eqref{surf_div}, Lemma \ref{appendix} and the identity $\divG(f\bn)=f\divG\bn$ we obtain
\begin{align*}
      (\divG(\bA\bN))_j&=\divG (\bN\bA^T)_j
   = \divG \bN(\bA^T)_j
   \\&= \divG (\bn\cdot(\bA^T)_j)\bn=(\bn\cdot(\bA^T)_j)\divG\bn 
=\tr(\bB)(\bA\bn)_j
\end{align*}
and $(\divG(\bN\bA\bN))_j =\tr(\bB)(\bN\bA\bn)_j$. We use the identity $\divG(f\bu)=f\divG\bu+\bu\cdot\nablaG{}f$ in 
\begin{align*}
      (\divG(\bN\bA))_j&= \divG (\bA^T\bN)_j
   =\divG \bn_j(\bA^T\bn)=\bn_j \divG(\bA^T\bn)+\bA^T\bn\cdot\nablaG\bn_j,                   \\
    &=\divG(\bA^T\bn)\bn_j+ \bA^T\bn\cdot\bB_j
      =\divG(\bA^T\bn)\bn_j+(\bB\bA^T\bn)_j
   \end{align*}
   and put together all the components as follows
   \begin{align*}
     \divG\bA&=\divG\bA_\Gamma+\tr(\bB)\bA\bn+\divG(\bA^T\bn)\bn+\bB\bA^T\bn-\tr(\bB)\bN\bA\bn
     \\
     &=\divG\bA_\Gamma+\divG(\bA^T\bn)\bn+\tr(\bB)\bP\bA\bn+\bB\bA^T\bn\,.
\end{align*}
The claim follows due to $\bB=\bP\bB$.
\end{proof}

\begin{lemma}[covariant integration by parts]\label{matrix_parts}
For a vector $\bu$  and a matrix $\bA$ on $\Omega_\delta$, we have
\begin{align*}
    (\divG \bA,\bu)_\Gamma &=
    -(\bA,\nabla_M\bu)_\Gamma+ ((\tr\bB)\bA\bn,\bu)_\Gamma
   \end{align*}
 \end{lemma}
\begin{proof}
  Since  $\divG\bA$ depends only on the values of $\bA$ on $\Gamma$ we first restrict $\bA$ and $\bu$ to $\Gamma$ and then extend them normally obtaining $\bA^e$ and $\bu^e$. We use Proposition \ref{bulk_part} to obtain
\begin{align*}
    (\div\bA^e,\bu^e)_{\Omega_\delta}
    &=-(\bA^e,\nabla\bu^e)_{\Omega_\delta} + (\bA^T\bu^e,\bn)_{\Gamma^{+}_\delta}-(\bA^T\bu^e,\bn)_{\Gamma^{-}_\delta}
 \end{align*}
 We take the limits $\delta\rightarrow{}0$:
 \begin{align*}
 &\lim_{\delta\rightarrow0}\frac{1}{2\delta}\left(\div\bA^e, \bu^e\right)_{\Omega_\delta}=\left(\div\bA^e, \bu\right)_{\Gamma}
 \,,\qquad \lim_{\delta\rightarrow0}\frac{1}{2\delta}(\bA^e,\nabla\bu^e)_{\Omega_\delta}= (\bA,\nabla\bu^e)_{\Gamma}
     \\
     &\lim_{\delta\rightarrow0}\frac{1}{2\delta}\left( (\bA^T\bu^e,\bn)_{\Gamma^{+}_\delta}-(\bA^T\bu^e,\bn)_{\Gamma^{-}_\delta}\right)=((\tr\bB)\bA^T\bu,\bn)_{\Gamma}=((\tr\bB)\bA\bn,\bu)_{\Gamma}
      \end{align*}
      and conclude the proof by noticing  $\div\bA^e=\divM\bA=\divG\bA$ and $\nabla{}\bu^e=\nablaM\bu\neq\nablaG\bu$.
\end{proof}

\begin{remark}
We point that for both vectors and matrices the decompositions in Lemma \ref{vector_div} and \ref{matrix_div} contain terms that have counterparts in Lemma \ref{vector_parts} and \ref{matrix_parts}. This fact is used in the derivation of the model in Section \ref{Onsager}.
\end{remark}



\begin{corollary}[\VY{external} integration by parts]\
    \label{3tensors_parts}
For matrix fields $\bA,\bC$ on $\Omega_\delta$, we have
\begin{align*}
    (\divM \nablaM\bA,\bC)_\Gamma &=
    -(\nabla_M\bA,\nabla_M\bC)_\Gamma\,.
   \end{align*}
\end{corollary}
\begin{proof}
The assertion follows from Lemma \ref{matrix_parts} and $(\nablaM\bA^T)\bn=0$. In fact, we have
\begin{align*}
    &(\divM \nablaM\bA,\bC)_\Gamma = \Big(\sum_{j=1}^n\be_j\otimes\divM\nablaM(\bA^T)_j, \sum_{j=1}^n\be_j\otimes(\bC^T)_j \Big)_\Gamma
     =\sum_{j=1}^n(\divM\nablaM(\bA^T)_j, (\bC^T)_j)_\Gamma
     \\
    & =  -\sum_{j=1}^n(\nablaM(\bA^T)_j,\nabla_M(\bC^T)_j)_\Gamma+ ((\tr\bB)(\nablaM(\bA^T)_j)\bn,\bu)_\Gamma
   =-\sum_{j=1}^n(\nablaM(\bA^T)_j,\nabla_M(\bC^T)_j)_\Gamma
   \\
   & =- \Big(\sum_{j=1}^n\be_j\otimes\nablaM(\bA^T)_j,\sum_{j=1}^n\be_j\otimes\nablaM(\bC^T)_j\Big)_\Gamma =-(\nabla_M\bA,\nabla_M\bC)_\Gamma .
   \end{align*}
   This concludes the proof.
\end{proof}

\bibliographystyle{plain}
\bibliography{literatur0}
\end{document}